\documentclass[11pt]{article}

\usepackage{geometry}
\usepackage{titlesec}
\usepackage{titlepic}
\usepackage[utf8]{inputenc}
\usepackage[T1]{fontenc}
\usepackage{lmodern}
\usepackage{amsfonts}
\usepackage{amssymb}
\usepackage{amsthm}
\usepackage{mathtools}
\usepackage{color}
\usepackage{chngcntr}
\usepackage{dsfont}
\usepackage{framed}
\usepackage{graphicx}
\usepackage{color}
\usepackage{mathrsfs}
\usepackage{setspace}
\usepackage[normalem]{ulem}
\usepackage{enumerate}
\usepackage[shortlabels]{enumitem}
\usepackage[all]{xy}
\usepackage{url}
\usepackage{hyperref}
\usepackage{marvosym}
\usepackage{booktabs}
\usepackage{etoolbox}
\usepackage{tocloft}
\usepackage{lipsum}
\usepackage{comment}
\usepackage{stmaryrd}
\usepackage{xpatch,amsthm}
\usepackage{todonotes}

\counterwithin{figure}{section}
\counterwithin{table}{section}
\counterwithin{equation}{section}
\makeatletter
\xpatchcmd{\@thm}{\fontseries\mddefault\upshape}{}{}{} 
\makeatother
\newtheorem{theorem}[figure]{Theorem}
\newtheorem{proposition}[figure]{Proposition}
\newtheorem{lemma}[figure]{Lemma}
\newtheorem{corollary}[figure]{Corollary}
\theoremstyle{definition}
\newtheorem{definition}[figure]{Definition}

\newtheorem{assumption}[figure]{Assumption}
\theoremstyle{remark}
\newtheorem{examplex}[figure]{Example}

\newtheorem{remarkx}[figure]{Remark}

\newenvironment{example}{\pushQED{\qed}\examplex}{\popQED\endexamplex}

\newenvironment{remark}{\pushQED{\qed}\remarkx}{\popQED\endremarkx}

\newcommand{\dd}{\mathrm{d}}
\providecommand{\diff}{\dd}

\newcommand{\E}{\mathbb{E}}
\newcommand{\N}{\mathbb{N}}
\newcommand{\Pb}{\mathbb{P}}
\newcommand{\Q}{\mathbb{Q}}
\newcommand{\R}{\mathbb{R}}

\newcommand{\F}{\mathbb{F}}
\newcommand{\FF}{\mathbb{F}}

\newcommand{\cE}{\mathcal{E}}
\newcommand{\cF}{\mathcal{F}}
\newcommand{\cG}{\mathcal{G}}
\newcommand{\cL}{\mathcal{L}}
\newcommand{\cP}{\mathcal{P}}

\newcommand{\cS}{\mathbb{S}}
\newcommand{\ccS}{\mathcal{S}}
\newcommand{\cM}{\mathcal{M}}

\newcommand{\cU}{\mathcal{U}}
\newcommand{\id}{\text{id}}

\newcommand{\TB}{\overline{\Theta}}

\newcommand{\vt}{\vartheta}

\newcommand{\sint}{\stackrel{\mbox{\tiny$\mskip-6mu\bullet\mskip-6mu$}}{}}

\newcommand{\Pas}{\Pb\text{-a.s.}}%
\newcommand*{\EX}[2][]{\E^{#1}\left [ #2 \right ]}

\newcommand*{\as}[1]{#1\text{-a.s.}}

\newcommand*{\ol}[1]{\bar{#1}}

\newcommand{\eq}[1]{\mathrel{=}_{#1}}

\DeclareMathOperator*{\loc}{loc}
\DeclareMathOperator*{\tr}{td}
\DeclareMathOperator*{\trr}{tr}
\DeclareMathOperator*{\nt}{nt}

\DeclareMathOperator*{\Var}{Var}
\DeclareMathOperator*{\Cov}{Cov}

\DeclareMathOperator*{\MV}{MV}
\DeclareMathOperator*{\MVH}{MVH}
\DeclareMathOperator*{\ex}{ex}

\DeclareMathOperator*{\esssup}{ess\,sup}
\DeclareMathOperator*{\essinf}{ess\,inf}
\DeclareMathOperator*{\unif}{unif}

\newcommand{\I}{\mathds{1}}
\providecommand{\one}{\I}

\newcommand{\ptimes}{\otimes}

\renewcommand*{\leq}{\leqslant}
\renewcommand*{\geq}{\geqslant}

\usepackage[font=scriptsize,margin=3em]{caption}

\newcommand*{\doi}[1]{\href{http://dx.doi.org/\detokenize{#1}}{doi}}
\renewcommand*{\url}[1]{\href{\detokenize{#1}}{\footnotesize\detokenize{#1}}}

\allowdisplaybreaks
\hypersetup{hidelinks,colorlinks,urlcolor=blue,linkcolor=[rgb]{.5,0,0},citecolor=[rgb]{.5,0,0}}
\setlist{topsep=0em, partopsep=0pt, itemsep=0pt}
\setlist[enumerate]{label=\rm(\roman*)}

\usepackage{url}
\usepackage[normalem]{ulem}

\newcommand\Tau{\mathcal{T}}

\usepackage{etoolbox}

\usepackage{todonotes}

\begin{document}

\setlength{\parskip}{.25em}
	
\title{\Large Incomplete-market CAPM equilibria in continuous time}
\author{Nikolaos Constantinou\thanks{University of Stuttgart, Department of Mathematics, 57 Pfaffenwaldring, 70569 Stuttgart, Germany, email: \texttt{nikolaos.constantinou@isa.uni-stuttgart.de}.} \and Christoph Czichowsky\thanks{London School of Economics and Political Science, Department of Mathematics, Columbia House, Houghton Street, London WC2A 2AE, UK, email: \texttt{c.czichowsky@lse.ac.uk}.} \and Martin Herdegen\thanks{University of Stuttgart, Department of Mathematics, 57 Pfaffenwaldring, 70569 Stuttgart, Germany, email: \texttt{martin.herdegen@isa.uni-stuttgart.de}.} \and David Martins\thanks{ETH Z\"{u}rich, Department of Mathematics, R\"{a}mistrasse 101, 8092 Z\"{u}rich, Switzerland, email: \texttt{davidptmartins@gmail.com}.}}
\maketitle

\begin{abstract}

We study continuous-time CAPM equilibria with general mean-variance preferences of the form $U(\mu,\sigma)$, where $U$ is strictly quasiconcave, increasing in the mean $\mu$ and decreasing in the volatility $\sigma$ of terminal wealth. Our framework accommodates both financial and productive assets with nontradeable endowments. We allow for equilibria driven by general semimartingales in the case of linear mean-variance preferences and by continuous semimartingales in the case of general mean-variance preferences. The key insight is a connection between mean-variance efficiency and quadratic utility maximisation. We first characterise quadratic equilibria through an explicit family of equilibrium price processes indexed by the aggregate risk tolerance. We then show that each agent's optimal mean-variance strategy corresponds to a quadratic optimum with an endogenously determined risk tolerance. This correspondence reduces the mean-variance equilibrium problem to a one-dimensional fixed-point problem, whose solvability we establish using fine properties of (dynamic) mean-variance hedging and a novel stability result for quadratic BSDEs with BMO growth at the origin. Finally, examples illustrate how unspanned stochastic-volatility risk interacts with aggregate endowments to affect equilibrium asset prices, and how nonlinear mean-variance preferences influence equilibrium risk tolerance and asset-price comparative statics.

\end{abstract}




\section{Introduction}

The capital asset pricing model (CAPM) provides an important benchmark equilibrium framework for understanding how the risks borne by investors are reflected in asset prices.
The existence of CAPM equilibria has mainly been studied in one-period economies or in complete markets.
Extending this framework to continuous-time incomplete markets presents a fundamental difficulty: investors' optimal trading decisions depend on their hedging opportunities, which are themselves determined by asset prices that must be consistent with market clearing.
This interaction becomes particularly relevant when agents are exposed to nontradeable endowment risks and when the economy includes both financial assets and productive (or real) assets paying terminal dividends.

In this paper, we establish existence of continuous-time CAPM equilibria in incomplete markets with general mean--variance preferences and nontradeable endowments under suitable assumptions.
We allow for equilibria driven by general semimartingales for linear mean--variance preferences and by continuous semimartingales for general mean--variance preferences.
The central step is a reduction of the equilibrium problem to a one-dimensional fixed-point problem for aggregate quadratic risk tolerance within an explicit family of quadratic equilibria.
This reduction exploits the correspondence between mean--variance optimality and quadratic utility maximisation, with individual quadratic risk tolerances determined endogenously.
Our setting extends the classical CAPM of Treynor~\cite{treynor:62}, Sharpe~\cite{sharpe:64}, Lintner~\cite{lintner:65a,lintner:65b} and Mossin~\cite{mossin:66}, which, despite its well-known limitations, remains influential in investment practice; see Levy~\cite{levy:11} for an overview.

Our study is situated in the tradition of Radner equilibria, which combine individual optimality with market clearing in sequential markets; see Radner~\cite{radner:82}.
The existence and uniqueness of CAPM equilibria were investigated by Nielsen~\cite{nielsen:88,nielsen:90a,nielsen:90b} and Allingham~\cite{allingham:91}, followed by Dana~\cite{dana:99}, Hens et al.~\cite{hens:al:02}, Wenzelburger~\cite{wenzelburger:10} and Koch-Medina/Wenzelburger~\cite{koch-medina:wenzelburger:18}.
These contributions study preferences specified directly through the mean and volatility of terminal wealth.
Related work by Berk~\cite{berk:97} investigates conditions under which the CAPM arises within an expected-utility framework.
Much of this rigorous equilibrium literature concerns one-period economies, although dynamic extensions of the CAPM have long been considered; see, for example, the multi-period model of Stapleton/Subrahmanyam~\cite{stapleton} and the continuous-time consumption-based model of Breeden~\cite{breeden}.
We extend the equilibrium existence theory to continuous-time incomplete markets, building in particular on the one-period setting of Koch-Medina/Wenzelburger~\cite{koch-medina:wenzelburger:18}.

Market incompleteness plays a central role in our analysis.
In complete or sufficiently spanned markets, agents can transfer or hedge their idiosyncratic endowment risks, so representative agent arguments can substantially simplify equilibrium analysis.
When endowment risks contain unspanned components, this risk-sharing picture changes.
As shown in the one-period incomplete-market analysis of Koch-Medina/Wenzelburger~\cite{koch-medina:wenzelburger:18}, agents still hedge their individual endowments as well as possible, but equilibrium prices reflect an extended market portfolio associated with the hedgeable part of aggregate endowment.
Unspanned endowment risk can nevertheless affect agents' demand for risk and hence equilibrium prices.
Such effects are particularly relevant in continuous time, where trading occurs dynamically and the underlying risk factors need not be fully spanned.
Mean--variance preferences also provide a useful benchmark beyond the classical CAPM: Markowitz~\cite{markowitz:2010,markowitz:2014} discusses how mean--variance efficient portfolios can approximate the outcomes of expected-utility maximisation for broader classes of risk-averse preferences.

The general mathematical framework of no-arbitrage pricing and utility maximisation in semimartingale markets is by now well developed, but existence of equilibria with endogenous price dynamics remains more challenging.
Several continuous-time incomplete-market equilibrium results concern exponential utility and particular stochastic environments.
Related contributions include Escauriaza et al.~\cite{escauriaza:al:21} and Kardaras et al.~\cite{kardaras:al:22}, in which equilibrium is characterised by coupled systems of backward stochastic differential equations (BSDEs) with quadratic growth in their control variables.
Xing/\v{Z}itkovi\'{c}~\cite{xing:zitkovic:18} establish solvability results for Markovian quadratic BSDE systems, while Escauriaza et al.~\cite{escauriaza:al:21} also treat discontinuities in the control variable in a Markovian setting, where analytical conditions on the state process play an important role.
Our approach exploits the particular structure of mean--variance efficiency.
The quadratic equilibrium characterisation applies in a general semimartingale setting; for the existence theorem under general mean--variance preferences, we impose a continuous-filtration assumption together with boundedness and nondegeneracy conditions.

We proceed in two steps.
First, we study quadratic utility maximisation and characterise equilibrium prices for prescribed individual quadratic risk tolerances.
The quadratic structure permits a nonstandard representative-agent construction even when the market is incomplete: market clearing is equivalent to the representative agent having no incentive to trade.
This yields an explicit equilibrium pricing measure and formulas for the price processes of both financial and productive assets, together with existence and uniqueness criteria within the quadratic model. The resulting family of equilibrium prices is indexed by a single aggregate quadratic risk-tolerance parameter, providing the basis for the subsequent fixed-point construction.
We also provide an example illustrating how unspanned stochastic-volatility risk affects prices and showing when an equilibrium fails to exist because the candidate price process is incompatible with admissible trading.

Second, we solve the individual mean--variance optimisation problem in a fixed market using the connection between quadratic utility, mean--variance hedging and the mean--variance efficient frontier. For each agent, the mean--variance optimum coincides with a quadratic optimum for an endogenously determined risk-tolerance parameter.
We therefore search for an aggregate quadratic risk tolerance consistent with the individual parameters implied by the corresponding equilibrium prices.
This reduces the continuous-time mean--variance equilibrium problem to a one-dimensional fixed-point problem within the family of quadratic equilibria.
Under suitable assumptions we establish existence of a fixed point and hence of a mean--variance equilibrium, without claiming uniqueness of the latter.
We also obtain formulas for the individual quadratic risk tolerances in equilibrium and an explicit aggregate formula in symmetric economies.
An example with nonlinear mean--variance preferences illustrates the resulting comparative statics for equilibrium risk tolerance and asset prices.

The fixed-point argument requires a delicate continuity analysis.
Its main analytical ingredient is the dynamic description of mean--variance hedging in terms of coupled BSDEs; see Bobrovnytska/Schweizer~\cite{bobrovnytska:schweizer:04} and Jeanblanc et al.~\cite{jeanblanc:al:12}.
The quadratic-equilibrium prices enter the BSDE drivers nonlinearly, and these drivers exhibit quadratic growth in the control variables and BMO growth at the origin.
The requisite stability theorem is established in the companion paper by Constantinou et al.~\cite{constantinou:al:25b}, building on the work of Briand/Hu~\cite{briand:hu:08}, Mocha/Westray~\cite{mocha:westray:12} and Frei et al.~\cite{frei:al:12}.
Combined with properties of dynamic mean--variance hedging, this result provides the continuity needed for the fixed-point construction.

The remainder of the article is organised as follows.
Section~\ref{sec:model:prelim:results} introduces the market model, agents' preferences and the notion of equilibrium.
Section~\ref{sec:quadratic} establishes the quadratic equilibrium characterisation and gives examples highlighting incomplete-market pricing effects and nonexistence.
Section~\ref{sec:MV} develops individual mean--variance optimality, the connection with quadratic equilibria and the fixed-point existence result of Theorem~\ref{thm:fixedpoint}; it also derives equilibrium risk-tolerance formulas and concludes with an explicit symmetric example.
Section~\ref{sec:conclusions} summarises the contributions.
The mean--variance hedging results used throughout are collected in Appendix~\ref{app:mean:var:hedging:prelim}, while technical proofs are provided in Appendices~\ref{sec:proofs} and~\ref{app:auxiliary}.

\section{Model and preliminary results}
\label{sec:model:prelim:results}

\subsection{Notation}

We work on a filtered probability space $(\Omega, \cF, \FF = (\cF_t)_{0 \leq t \leq T}, \Pb)$ with a fixed finite time horizon~$T\in(0,\infty)$. We assume that the filtration $\FF$ satisfies the usual conditions of right-continuity and completeness. Moreover, we assume that $\cF_0$ is $\Pb$-trivial and $\cF_T = \cF$. 

For $n \in \N$, we denote by $\cM(\R^n)$, $\cM^2(\R^n)$, $\cM_{\loc}^2(\R^n)$ and $\cP(\R^n)$, the sets of $\R^n$-valued martingales, $L^2$-martingales, local $L^2$-martingales and predictable processes on $[0,T]$, respectively; we denote the set of $\R^n$-valued continuous local $L^2$-martingales starting from zero by $\cM_{0,\loc}^2(\R^n)$. For $X \in \cM_{\loc}^2(\R^n)$, we denote $L_{\loc}^2(X) = L_{\loc}^2(X; \R^n)$ the set of $\R^{n\times1}$-valued predictable processes $(\xi_t)_{t\in[0,T]}$ that are $X$-integrable in the sense of Jacod/Shiryaev~\cite[III.6.17]{jacod:shiryaev:03}. For $m \in \N$, we also define the set
\[
L_{\loc}^2(X; \R^{n\times m}) = \{\xi = (\xi^1,...,\xi^m): \xi^1,...,\xi^m \in L_{\loc}^2(X)\},
\]
and for $\xi \in L_{\loc}^2(X; \R^{n\times m})$, we write $\xi \bullet X = (\xi^1 \bullet X, ..., \xi^m \bullet X)^\top$ with values in $\R^{m\times 1}$. We identify $\xi^{(1)},\xi^{(2)} \in L_{\loc}^2(X; \R^{n \times m})$ if they are \emph{$X$-equivalent}, i.e., if $\xi^{(1)} \bullet X$ and $\xi^{(2)} \bullet X$ are indistinguishable; we write in that case $\xi^{(1)} =_X \xi^{(2)}$.

Now, for $n_1,n_2,m_1,m_2 \in \N$, we have for all local martingales $X^{(1)} \in \cM_{\loc}^2(\R^{n_1})$, $X^{(2)} \in \cM_{\loc}^2(\R^{n_2})$, $\xi^{(1)} \in L_{\loc}^2(X^{(1)}; \R^{n_1\times m_1})$ and $\xi^{(2)} \in L_{\loc}^2(X^{(2)}; \R^{n_2 \times m_2})$ that the (predictable) quadratic covariation $\langle \xi^{(1)} \bullet X^{(1)}, \xi^{(2)} \bullet X^{(2)} \rangle$ takes values in $\R^{m_1 \times m_2}$ and satisfies
\[
\dd \langle \xi^{(1)} \bullet X^{(1)}, \xi^{(2)} \bullet X^{(2)} \rangle_t = (\xi_t^{(1)})^\top \dd \langle X^{(1)},X^{(2)} \rangle_t \xi_t^{(2)}, \qquad t \in [0,T].
\]
For $m_3 \in \N$, we have for $\xi^{(3)} \in L_{\loc}^2(\xi^{(1)}\bullet X^{(1)}; \R^{m_1\times m_3})$ the associativity property
\[
\xi^{(3)} \bullet (\xi^{(1)} \bullet X^{(1)}) = (\xi^{(1)}\xi^{(3)}) \bullet X^{(1)}.
\]

Take an equivalent probability measure $\Q \approx \Pb$ for the notational background. For $0 \leq s \leq t \leq T$, we denote by $\Tau_{s,t}$ the set of $[s,t]$-valued stopping times; for $\tau \in \Tau_{0,T}$, we write $\E_\tau^\Q[\,\cdot\,]$ to mean the $(\Q,\cF_\tau)$-conditional expectation. For $p \in [1,\infty)$, we denote by $\cL^p(\Q)$ the space of $\R$-valued, progressively measurable processes $X = (X_t)_{t\in[0,T]}$ satisfying $\|X\|_{\cL^p(\Q)}^p := \E^\Q[\int_0^T |X_u|^p \,\dd u] < \infty$, $\cL^\infty$ the space of $\R$-valued, progressively measurable, $\Q \ptimes \dd t\text{-a.e.}$ bounded processes, and $\cL_{\unif}^\infty$ the space of $\R$-valued, predictable, uniformly bounded processes. Moreover, we define $BMO(\Q)$ the space of $\R$-valued local martingales $M = (M_t)_{t\in[0,T]}$ satisfying
\[
\|M\|_{BMO_2(\Q)}^2 := \sup_{\tau\in\Tau_{0,T}}\Big\|\E_\tau^\Q[\langle M \rangle_T - \langle M \rangle_\tau]\Big\|_\infty < \infty.
\]
We shall drop the $\Q$-qualifiers whenever working under the real-world probability measure $\Pb$.

Moreover, we say that a strictly positive martingale $Z \in \cM(\R)$ satisfies the \emph{reverse H\"{o}lder inequality} ($R_2(\Pb)$) if there exists a constant $C \geq 1$ such that
\begin{equation}
\label{eq:R2P}
\E\Bigg[\Big(\frac{Z_T}{Z_t}\Big)^2 \,\Big|\, \cF_t\Bigg] \leq C, \qquad t \in [0,T].
\end{equation}
Similarly, we say that an equivalent measure $\Q \approx \Pb$ satisfies $R_2(\Pb)$ if the density process $(Z_t^{\Q})_{t\in[0,T]}$ of $\Q$ satisfies $R_2(\Pb)$.

Finally, for a continuous local martingale $M\in \cM_{\loc}^2(\R^n)$, define the predictable increasing process $I_t^M = \trr \langle M \rangle_t$, $t \in [0,T]$. There exists a Radon-Nikod\'{y}m derivative for $\langle M \rangle$ with respect to $I^M$, i.e., a $\Pb \ptimes I^M\text{-a.e.}$ unique predictable process $(\kappa_t^{M})_{t\in[0,T]}$ with values in the set of $n \times n$ symmetric positive semidefinite matrices such that
\begin{equation}
\label{eq:RN}
\langle M \rangle_t = \int_0^t \kappa_s^{M} \,\dd I_s^M, \qquad t \in [0,T].
\end{equation}

\subsection{Financial market}
\label{subsec:fin:mark:primitives}

We consider a financial market consisting of $1 + d = 1 + d_1 + d_2$ assets. The first asset, with price process $S^0$, serves as numéraire, and we assume that \mbox{$(S^0_t)_{0 \leq t \leq T} \equiv 1$}. The dynamics of the numéraire asset is not determined in equilibrium as we do not consider intertemporal consumption. In addition, we consider $d_1$ \textit{financial assets} with price processes $S^{(1)} = (S^1_t, \ldots, S^{d_1}_t)_{0 \leq t \leq T}$ and $d_2$ \textit{productive assets} (sometimes also referred to as \emph{real assets}) with price processes $S^{(2)} = (S^{d_1+ 1}_t, \ldots, S^{d_1+d_2}_t)_{0 \leq t \leq T}$. The risky assets are collectively expressed as $S := (S^{(1)}, S^{(2)})$. In the following, we likewise use the notation $x = (x^{(1)}, x^{(2)})$ for each $x \in \R^{d_1+d_2}$ with $x^{(i)} \in \R^{d_i}$.

\medskip{}

The price processes $S^{(1)}$ and $S^{(2)}$ are not given a priori but rather determined in equilibrium between $K$ agents trading in the market. 
\medskip{}

For the financial assets, which are in zero net supply, we assume that the initial value and volatility structure are exogenously given, whereas the drift is determined in equilibrium. More precisely, for $j \in \{1, \ldots, d_1\}$, we have
\begin{equation}
\label{eq:def:prim:fin:assets:decomp}
S^{j}_t = S_0^j + M^j_t + A^j_t, \quad 0 \leq t \leq T, 
\end{equation}
where $S_0^j \in \R$ and the local martingale part $M^j\in \cM^{2}_{0,\loc}$ are given a priori, whereas the predictable finite-variation process $A^j$ (which is null at time $0$) is to be determined in equilibrium.

The economic interpretation of the financial assets is as follows: The processes $M^1, \ldots, M^{d_1}$ may be regarded as modelling different risk factors in the economy, and the introduction of the assets $S^1, \ldots, S^{d_1}$ enables the market participants to trade their exposure to these risk factors, at appropriate prices set by the market, which are reflected in the dynamics of $A^1, \ldots, A^{d_1}$. For a further motivation for the structure of these assets, we refer to Zitkovic \cite{zitkovic:12} and Kardaras et al. \cite{kardaras:al:22} who have called this the “fast-and-slow“ model.

\medskip{}

For the productive assets, which are in unit net supply, we assume that they entitle the owner to a terminal dividend at time $T$, and otherwise the price process is determined in equilibrium. More precisely, for $j \in \{1, \ldots, d_2\}$ we have
\begin{equation}
\label{eq:def:prim:prod:assets:divid}
S^{d_1+j}_T = D^j,
\end{equation}
where $D^j \in L^2$ denotes the random terminal dividend, and the rest of the price process $(S^j)_{0 \leq t < T}$ is to be determined by the market in equilibrium. The assumption of a terminal dividend only for the productive assets is fairly standard in the literature on general equilibria, see e.g.~\cite{escauriaza:al:21, herdegen:al:21}.

\medskip{}

For mathematical reasons, we assume that each asset $S^j$ is a \emph{local $L^2$-semimartingale} for $j \in \{1, \ldots, d\}$. This means that there exists a localising sequence of stopping times $(\tau_n)_{n \in \N}$ such that each stopped process $S^{j,\tau_n}=(S^j_{\tau_n\wedge t})_{0\leq t\leq T}$ is an $L^2$-semimartingale, in the sense that
\begin{equation}
\label{def:local L2 semimart}
\sup\big\{\E[(S^{j,\tau_n}_\sigma)^2]: \sigma \textrm{ stopping time}\big\} < \infty;
\end{equation}
see Delbaen/Schachermayer \cite{delbaen:schachermayer:96} and \v Cern\' y/Kallsen \cite{cerny:kallsen:07} for details. We refer to this property by calling $(1, S)$ a \emph{local $L^2$-market}. Note that by \cite[Lemma A.2]{cerny:kallsen:07}, a stochastic process is a local $L^2$-semimartingale if and only if it is a special semimartingale whose local martingale part is locally square-integrable. In view of \eqref{eq:def:prim:fin:assets:decomp}, this is only a condition on the productive assets as it is automatically satisfied by the financial assets.

\subsection{Admissible strategies}

In order to describe trading in the market $(1, S)$, we need to specify which strategies are considered \emph{admissible}. To this end, we follow a two-step approach as in \v Cern\' y/Kallsen \cite{cerny:kallsen:07}, to which we also refer for further details.

\begin{definition}
Let $(1, S)$ be a local $L^2$-market. A \emph{simple integrand for $S$} is a process of the form $\vt=\sum_{i=1}^{m-1}\xi_i \I_{\rrbracket \sigma_i,\sigma_{i+1}\rrbracket }$, where $m \in \N$, \mbox{$0\leq\sigma_1\leq\cdots\leq\sigma_m$} are $[0, T]$-valued stopping times, and each $\xi_i$ is a bounded $\cF_{\sigma_i}$-measurable random vector in $\R^d$, such that each stopped process $S^{j, \sigma_m} = (S^j_{\sigma_m \wedge t})_{0 \leq t \leq T}$ is an $L^2$-semimartingale for $j=1, \ldots, d$. We denote by $\Theta_{\rm simple}(S)$ the linear space of all simple integrands for $S$. We also let $L(S)$ be the set of predictable $S$-integrable processes on $[0,T]$; see Jacod/Shiryaev \cite[III.6.17]{jacod:shiryaev:03}.
\end{definition}
\begin{definition}\label{def:adm:strategy:l:2:cerny:kallsen}
Let $(1, S)$ be a local $L^2$-market. Then $\vt\in L(S)$ is called \emph{\mbox{$L^2$-admissible} for $S$} if $\vt \sint S_T \in L^2$ and there exists a sequence $(\vt^n)_{n\in\N}$ in $\Theta_{\rm simple}(S)$ such that
\begin{itemize}
\item[\textbf{1)}]  $\vt^n\sint S_T\overset{L^2}{\longrightarrow} \vt\sint S_T$,
\item[\textbf{2)}]  $\vt^n\sint S_\tau\overset{\Pb}{\longrightarrow} \vt\sint S_\tau$ for all $[0,T]$-valued stopping times $\tau$.
\end{itemize}
Here, $\vt \sint S = (\vt \sint S_t)_{0 \leq t \leq T}$ denotes the stochastic integral $\vt \sint S_t = \int_0^t \vt_r dS_r$ for $t \in [0,T]$. The set of all $L^2$-admissible trading strategies is denoted by $\overline{\Theta}(S)$.
\end{definition}

\begin{remark}
Our definition of $L^2$-admissible strategies slightly differs from the definition given in \cite{cerny:kallsen:07}, because we stipulate 2) for all stopping times $\tau$ and not only for deterministic times $t\in [0,T]$. However, under \cite[Assumption 2.1]{cerny:kallsen:07}, i.e., if there exists an equivalent local martingale measure (ELMM) $\Q$ for $S$ with $\frac{\dd \Q}{\dd \Pb}\in L^2(\Pb)$, both definitions coincide. 
The reason for this change is that it allows us to use dynamic programming arguments even if there does not exist an ELMM $\Q$ for $S$ with $\frac{\dd \Q}{\dd \Pb}\in L^2(\Pb)$, as in Czichowsky/Schweizer \cite{czichowsky:schweizer:13}.
\end{remark}

We denote by $(e^j_t)_{0 \leq t \leq T} \equiv (0, \ldots, 0, 1, 0, \ldots, 0) \in \R^{d}$ the buy-and-hold strategy of the $j$-th risky asset, where $1$ is in the $j$-th position. In general, this strategy will not be $L^2$-admissible for $ S$ because we only assume $S^j$ to be a local $L^2$-semimartingale. However, if $S^j$ is an $L^2$-semimartingale (i.e., $\sup\{\E[(S^j_\sigma)^2]: \sigma \textrm{ stopping time}\} < \infty$ for each $j \in \{1, \ldots, d_1+d_2\}$), then automatically $e^j \in \Theta_{\rm simple}(S) \subseteq \overline \Theta(S)$ by \cite[Corollary 2.9]{cerny:kallsen:07}. 

Since different strategies may lead to the same stochastic integral, we naturally identify strategies $\vt \in \TB(S)$ via the following equivalence relation.

\begin{definition}
Let $(1, S)$ be a local $L^2$-market. Then $\vt,  \vt' \in \overline \Theta(S)$ are called \emph{$S$-equivalent} if $\vt \sint S$ and $\vt' \sint S $ are indistinguishable. In this case, we  write $\vt \eq{S} \vt'$; see Czichowsky/Schweizer \cite{czichowsky:schweizer:11} for more details on how to represent different equivalent classes via the so-called projection onto the predictable range.
\end{definition}

We assume that trading is frictionless and that market participants choose self-financing portfolios of the form $(\vt^0_t, \vt_t)_{0 \leq t \leq T}$, where $\vt^0$ is a predictable process, $\vt \in \overline{\Theta}(S)$ and the wealth process $(V_t(\vt^0,\vt))_{0 \leq t \leq T}$ satisfies the self-financing condition 
$$V_t(\vt^0,\vt) = \vt^0_t + \vt_t^\top S_{t} = \vt^0_0 +  \vt_0^\top S_0  +\vt \sint S_{t}, \quad \Pas \textrm{ for all } t \in [0, T].$$
Since we shall include the initial wealth of the agents into their endowments, a self-financing portfolio can be specified in terms of $\vt \in \overline{\Theta}(S)$ alone; henceforth, we drop the dependence on $\vartheta^0$. We do not impose any further conditions on the strategies of the agents: for example, they are allowed to borrow or short sell any asset.

\subsection{Agents and preferences}
\label{subsec:AaE}

We consider $K \in \N$ agents trading in the financial market $(1, S)$. We assume that each agent $k \in \{1, \ldots, K\}$ owns \emph{traded endowments} $\eta^{k,d_1+1}, \ldots,\eta^{k,d_1+d_2} \in \R$ in the productive assets $S^{d_1+1}, \ldots, S^{d_1+d_2}$ at time $0$, and is also entitled to receive a \emph{non-traded endowment} at time $T$, which consists of a random income $\Xi^{k,\nt} \in L^2$. Because  we assume zero interest rates and there are no liquidity constraints on the portfolios of the agents, it does not matter whether a fixed amount of cash is received via the traded or non-traded endowment. Thus by transferring it to the non-traded endowment, we may assume that each agent starts with zero cash. To simplify the notation, we formally set $\eta^{k,j} := 0$ for $j \in \{1, \ldots, d_1\}$.

We denote  by $\Xi^{k,\tr} :=  {\eta^{k, (2)}}^{\top}D^{(2)} \in L^2$ the value of the traded endowment of agent $k$ at time $T$.
The \emph{total endowment} of agent $k$ at time $T$ is then given by
\begin{equation}
\label{eq:def:tot:endowment:k}
\Xi^k = \Xi^{k,\tr}+ \Xi^{k,\nt}.
\end{equation}

Each agent $k \in \{1, \ldots, K\}$ interacts with the market by buying and selling assets according to an $L^2$-admissible strategy $\vt \in \overline{\Theta}(S)$, which includes the original endowment $\eta^{k,(2)}$ in the productive assets. 
Starting from their initial wealth ${\eta^{k, (2)}}^{\top}\! S^{(2)}_0$, Agent $k$ then can generate the wealth process ${\eta^{k, (2)}}^{\top}\! S^{(2)}_0 + \vt \sint S$ by trading with the strategy $\vt $ in a self-financing way. Since they additionally receive the non-traded endowment $\Xi^{k,\nt}$ at time $T$, their terminal wealth at time $T$ is given by 
\begin{equation}
\label{eq:wealth agent k:initial:version}
V^k_T(\vt) = {\eta^{k, (2)}}^{\top}\! S^{(2)}_0 +\vt\sint S_T + \Xi^{k,\nt}.
\end{equation}
Note that the traded endowment has the terminal value
$$\Xi^{k,\tr} =  {\eta^{k, (2)}}^{\top}D^{(2)} = {\eta^{k, (2)}}^{\top}\! S^{(2)}_T = {\eta^{k, (2)}}^{\top}\! S^{(2)}_0 + {\eta^{k}} \sint S_T,$$
since $\eta^{k} = (0, \eta^{k, (2)})$ is constant. Thus the terminal wealth of agent $k$ at time $T$ can be equivalently written as
\begin{align}
\label{eq:wealth agent k}
V^k_T(\vt) = \Xi^{k,\tr} - {\eta^{k}} \sint S_T + \vt\sint S_T + \Xi^{k,\nt} = (\vt - \eta^k) \sint S_T + \Xi^k
\end{align}
in terms of the total endowment defined in \eqref{eq:def:tot:endowment:k}.

From the right-hand side of \eqref{eq:wealth agent k}, we see that the total wealth at time $T$ is the sum of the total endowment and any gains or losses generated by the strategy $\vt - \eta^k$. This difference may be interpreted as a discretionary strategy that is employed by the agent in addition to the fixed endowment $\eta^{k}$. The right-hand side of \eqref{eq:wealth agent k:initial:version} gives an alternative interpretation. Instead of keeping the traded endowment, agent $k$ may immediately sell it for the price of ${{\eta}^{k,(2)}}^{\top}\! S^{(2)}_0$ and then trade with the strategy $\vt$; the non-traded endowment $\Xi^{k,{\rm n}}$ is then added to the wealth at time $T$. However, we note that the price $S^{(2)}_0$ is not known a priori, but rather determined in equilibrium. Thus, the right-hand side of \eqref{eq:wealth agent k} is more useful for solving the equilibrium problem, since the total endowment $\Xi^k$ is fixed by the primitives, so that only the stochastic integral term $ (\vt - \eta^k) \sint S_T$ depends on the (unknown) dynamics of $S$.

\medskip{}

Each agent $k \in \{1, \ldots, K\}$ has preferences over terminal wealth at time $T$ described by a functional $\cU_k: L^2 \to \R$. Agent $k$ seeks to maximise utility from terminal wealth at time $T$, i.e., to solve the problem
\begin{equation}
\label{eq:maximisation}
\cU_k\big((\vt - \eta^k) \sint S_T + \Xi^k \big) \longrightarrow \max_{\vt \in \overline{\Theta}(S)}!
\end{equation}
We consider two types of functionals $\cU_k$: a quadratic utility functional, which is the subject of Section \ref{sec:quadratic}, and a (general) mean--variance functional, which is considered in Section \ref{sec:MV}. 

The \emph{quadratic utility functional} for agent $k \in \{1, \ldots, K\}$ is given by
\begin{equation}
\label{eq:def:gen:quad:util}
\cU_k^{{\rm Q}}(V) = \E[U^{{\rm Q}}_k(V)] = \E[2\gamma_k V - V^2], \quad V \in L^2,
\end{equation}
where $U^{{\rm Q}}_k(x) := 2\gamma_k x - x^2$ for some risk-tolerance $\gamma_k > 0$. The parameter $\gamma_k$ is also sometimes  called the \emph{bliss point} of agent $k$, as it is the optimal wealth that the agent would like to attain at time $T$ in order to maximise their utility.

The \emph{(general) mean--variance functional} for agent $k \in \{1, \ldots, K\}$ is given by
\begin{equation}
\label{eq:MVutility}
\cU_k^{\MV}(V) = U_k^{\MV}(\E[V],\sqrt{\Var[V]}), \quad V \in L^2,
\end{equation}
where the \emph{mean-variance utility} $U_k^{\MV}: \R \times \R_+ \to \R$, is strictly quasiconcave, strictly increasing in the first variable $\mu$, strictly decreasing in the second variable $\sigma$, and there exist constants $\lambda_k^{\min},\lambda_k^{\max} \in (0,\infty)$ such that
\begin{equation}
\label{eq:mvbounds}
\frac{\sigma}{\lambda_k^{\max}} \leq - \frac{\partial_\sigma U_k^{\MV}(\mu,\sigma)}{\partial_\mu U_k^{\MV}(\mu,\sigma)} \leq \frac{\sigma}{\lambda_k^{\min}}, \qquad (\mu,\sigma) \in \R \times \R_+.
\end{equation}
The constants $\lambda_k^{\min}$ and $\lambda_k^{\max}$ essentially capture the breadth in the risk aversion of the agent; see Koch-Medina/Wenzelburger~\cite[Section 3.3]{koch-medina:wenzelburger:18} for an elaborate discussion on (sufficient conditions of) $(\mu,\sigma) \mapsto -\partial_\sigma U_k^{\MV}(\mu,\sigma)/\partial_\mu U_k^{\MV}(\mu,\sigma)$.

\begin{remark}
(a) Note that  \eqref{eq:def:gen:quad:util} differs slightly from the more standard definition of quadratic utility
\begin{equation*}
\tilde \cU_k^{{\rm Q}}(V) := \E[\tilde U^{{\rm Q}}_k(V)] = \E\bigg[V - \frac{V^2}{2\gamma_k}\bigg], \quad V \in L^2,
\end{equation*}
where $\tilde U^{{\rm Q}}_k(x) := x - \frac{1}{2 \gamma_k} x^2$. Of course, both formulations are equivalent.

(b) Main examples of mean-variance utility functions $U^{\MV}$ can be constructed by setting $U^{\MV}(\mu,\sigma) = f(\mu) - g(\sigma)$, $(\mu,\sigma) \in \R \times \R_+$, where $f: \R \to \R$ and $g: \R_+ \to \R$ are twice continuously differentiable functions with $f'(\mu) > 0$, $f''(\mu) < 0$, $g'(\sigma) > 0$ and $g''(\sigma) > 0$ for all $(\mu,\sigma) \in \R \times \R_+$, such that $f'$ is bounded and $g'$ is sandwiched by linear functions passing through the origin. A special case is that of \emph{linear mean-variance utility functional}, i.e.,
\begin{equation}
\label{eq:def:lin:mv}
\cU^{{\rm MV}}(V) = U^{{\rm MV}}(\E[V],\sqrt{\Var[V]}) = \E[V] - \frac{\Var[V]}{2\lambda}, \quad V \in L^2,
\end{equation}
where $U^{{\rm MV}}(\mu,\sigma) = \mu - \frac{\sigma^2}{2\lambda}$ for some risk tolerance $\lambda_k>0$. Note that the bounds~\eqref{eq:mvbounds} are satisfied in this case with $\lambda^{\min} = \lambda^{\max} = \lambda$.

Despite their apparent similarity, the maximisation problems induced by the two preference functionals \eqref{eq:def:gen:quad:util} and \eqref{eq:def:lin:mv} are not equivalent. Indeed, for $\lambda_k = \gamma_k$, the quadratic utility functional can be written as
\[
\begin{aligned}
\tilde \cU_k^{{\rm Q}}(V)
&= \E[\tilde U^{{\rm Q}}_k(V)]
 = \E\bigg[V - \frac{V^2}{2 \gamma_k}\bigg] \\
&= \E[V] - \frac{\Var[V]}{2 \gamma_k} - \frac{\E[V]^2}{2 \gamma_k} \\
&= \cU_k^{{\rm MV}}(V) - \frac{1}{2\gamma_k}\E[V]^2
\end{aligned}
\]
which is not of the form \eqref{eq:def:lin:mv} due to the presence of the additional term $-\frac{1}{2\gamma_k} \E[V]^2$.
\end{remark}

\subsection{Equilibrium}
We can now formulate the key notion of an equilibrium market, which we adapt from the classical concept of a Radner equilibrium. We take the primitives $M^{(1)}, D^{(2)}$, $\eta^{k}$, $\Xi^{k,\nt}$ and $\cU_k$ defined in Sections \ref{subsec:fin:mark:primitives} and \ref{subsec:AaE} as given, where  $\cU_k$ is either $\cU^{\rm Q}_k$ or $\cU_k^{{\rm MV}}$.
\begin{definition}\label{def:equilibrium}
A local $L^2$-market $(1, S^{(1)}, S^{(2)})$ is called an \emph{equilibrium market} if there exists strategies $\hat \vt^1, \ldots, \hat \vt^K \in \overline{\Theta}(S)$ such that
\begin{enumerate}[label={\textbf{\arabic*)}}]
\item For each agent $k \in \{1, \ldots, K\}$, $\hat \vt^k$ is the unique solution to the maximisation problem~\eqref{eq:maximisation}.
\item The market clears, i.e., for $t \in [0, T]$,
\begin{equation}\label{eq:market clearing}
\sum_{k=1}^K \hat \vt^{k,j}_t \eq{S} \bar \eta^j :=
\begin{cases}
0, &\text{if } j \in \{1, \ldots, d_1\}, \\
\sum_{k=1}^K\eta^{k,j}, &\text{if } j \in \{d_1 + 1, \ldots, d_1 + d_2\}.
\end{cases}
\end{equation}
\item $e^j \in \overline{\Theta}(S)$ for $j \in \{d_1+1, \ldots, d\}$, i.e., the buy-and-hold strategies of the productive assets are $L^2$-admissible.
\end{enumerate}
	
If $(1, S^{(1)}, S^{(2)})$ is an equilibrium market with respect to the mean--variance functionals defined in \eqref{eq:def:gen:quad:util} or \eqref{eq:def:lin:mv} for some parameters $(\gamma_k)_{k=1}^K$ or $(\lambda_k)_{k=1}^K$, we say that it is a \emph{quadratic equilibrium market} or \emph{mean--variance equilibrium market}, respectively.
\end{definition}

\begin{remark}
Definition \ref{def:equilibrium} is standard apart from two points: In 1) we require the uniqueness of the individually optimal strategies and in 3) we require that each buy-and-hold strategy $\eta^k$ is $L^2$-admissible. Both are natural requirements in the setting of quadratic/linear mean-variance preferences. While the Law of One Price is an intrinsic property of CAPM equilibria, uniqueness of optimal strategies is a condition on the agent's behaviour that is not necessarily guaranteed due to the agent's indifference to invest in one asset or another that leads to the same terminal payoff.

Regarding 1), note that for standard utility functions, existence of an optimiser is equivalent to existence of a unique optimiser because standard utility functions are strictly increasing and strictly concave. Since a quadratic utility is decreasing after the bliss point, we need to make uniqueness of the optimiser part of the equilibrium condition.

Regarding 3), note that the agents should be allowed to simply hold their respective traded endowments. This condition is always satisfied for equilibrium markets with standard utility functions.
Mathematically, this condition also ensures that $\hat \vt^k \in \TB(S)$ if and only if $\hat \vt^k - \eta^k \in \TB(S)$, which will allow us to find the optimal strategies $\hat \vt^k$ by first solving for $\hat \vt^k - \eta^k$; recall also the discussion below \eqref{eq:wealth agent k}.
\end{remark}

\begin{remark}
Another important point to stress is that the agent's mean-variance preferences are evaluated at time zero, and our corresponding equilibria are understood in a \emph{precommitment sense}: each agent chooses a dynamic trading strategy to maximise their terminal wealth evaluated at time zero. It is well known that mean-variance preferences are generally dynamically time inconsistent, and so an optimal strategy chosen at time zero need not remain optimal if the agent reoptimises the conditional criterion at an intermediate time $t\in(0,T)$ over the remaining time horizon $(t,T]$. In other words, our equilibrium concept requires individual optimality at time zero and market clearing, but does not impose time consistency. Alternative formulations based on time-consistent equilibrium strategies have been developed among others by Basak / Chabakauri~\cite{basak:chabakauri:10}, Czichowsky~\cite{czichowsky13} and Bj\"{o}rk et al.~\cite{bjorketal:17}.
\end{remark}

In the remainder of the paper we seek to find equilibrium markets for quadratic and mean--variance preferences. More precisely, we look for conditions on the primitives that ensure the existence and uniqueness of a corresponding equilibrium market, and we seek to characterise that market. We shall study quadratic equilibria in Section \ref{sec:quadratic}. This in turn will allow us to obtain results on mean--variance equilibria in Section \ref{sec:MV}.

\section{Quadratic preferences}
\label{sec:quadratic}

In this section, we consider the situation of quadratic utility, i.e., each agent $k$ solves the problem
\begin{align}
\label{eq:def:gen:quad:util:indiv:prob}
\cU^{{\rm Q}}_k\big(V^k_T(\vt)\big) = \E\big[2\gamma_k\big((\vt - \eta^k) \sint S_T + \Xi^k\big) - \big((\vt - \eta^k) \sint S_T + \Xi^k \big)^2\big] \longrightarrow \max_{\vt \in \overline{\Theta}(S)}!
\end{align}
As customary in the equilibrium literature, we first study the individual problem \eqref{eq:def:gen:quad:util:indiv:prob} for a given local $L^2$-market $(1,S)$ and then solve for equilibrium.

\subsection{Individual optimality}
\label{sec:ind:opt}
Throughout this section, we fix a local $L^2$-market $(1,S)$ that satisfies \eqref{eq:def:prim:fin:assets:decomp} and \eqref{eq:def:prim:prod:assets:divid}.

In order to study the maximisation problem \eqref{eq:def:gen:quad:util:indiv:prob}, we note that it is closely linked to the so-called mean--variance hedging problems for a payoff $H\in L^2$: 
\begin{itemize}
\item The \emph{mean--variance hedging}
(MVH) problem is given by
\begin{equation}
\label{eq:MVH}
\E\big[(\vt\sint S_T - H)^2\big] \longrightarrow \min_{\vt \in \overline{\Theta}(S)}!
\end{equation}
\item The \emph{extended mean--variance hedging} (exMVH) problem is given by
\begin{equation}
\label{eq:MVHex}
\E\big[(c + \vt\sint S_T - H)^2\big] \longrightarrow  \min_{(c,\vt) \in \R\times \overline{\Theta}(S)}!
\end{equation}
\end{itemize}
We say that \eqref{eq:MVH} has a unique solution if $\vt^1 \eq{S} \vt^2$ for any two solutions $\vt^1, \vt^2 \in \TB(S)$.
Similarly, we say that \eqref{eq:MVHex} has a unique solution if $c_1 = c_2$ and $\vt^1 \eq{S} \vt^2$ for any two solutions $(c_1,\vt^1), (c_2,\vt^2) \in \R \times \TB(S)$. A sufficient condition for the existence and uniqueness of solutions to \eqref{eq:MVH} and \eqref{eq:MVHex} is given in Proposition \ref{prop:LOP}. We refer to Schweizer \cite{Schweizer:10} for a recent overview of mean--variance hedging.

Now, denote the minimal value of~\eqref{eq:MVHex} by $\varepsilon(H,S)$. Moreover, we denote the corresponding solutions to~\eqref{eq:MVH}-\eqref{eq:MVHex} by $\vartheta^{\MVH}(H,S)$ and $(c(H,S),\vartheta^{\ex}(H,S))$, respectively. We shall drop the dependence on $S$ for simplicity when searching for individually optimal strategies since the market $(1,S)$ is fixed in this section.

The following result shows that the quadratic utility problem~\eqref{eq:def:gen:quad:util:indiv:prob} is equivalent to a MVH problem~\eqref{eq:MVH}. We recall that all proofs are given in Appendix~\ref{sec:proofs}.

\begin{lemma}
\label{lem:ind:opt}
Let  $(1,S)$ be a local $L^2$-market and assume that $\eta^k \in \overline{\Theta}(S)$. Then the following are equivalent:
\begin{enumerate}
\item The quadratic utility problem  \eqref{eq:def:gen:quad:util:indiv:prob} has a unique solution $\hat \vt^k \in \overline{\Theta}(S)$.
\item  For $H^k := \gamma_k - \Xi^k$, the MVH problem
\begin{equation}
\label{eq:lem:ind:mod}
\E[(\vt\sint S_T -  H^k )^2] \to \min_{\vt \in \overline{\Theta}(S)}!
\end{equation}
has a unique solution $\vt^{\rm MVH}(H^k) \in \overline{\Theta}(S)$.
\end{enumerate}
If either statement holds, then $\hat \vt^k \eq{S} \eta^k + \vt^{\rm MVH}(H^k)$.
\end{lemma}
In the above result, $H^k = \gamma_k - \Xi^k$ may be interpreted as the additional wealth that agent $k$ would like to obtain in order to reach the bliss point $\gamma_k$.

We now use part (b) of Lemma \ref{lem:ind:opt} to deduce a more explicit decomposition for the optimal strategy $\hat \vt^k$ of agent $k$ into a hedging and pure investment part, where the \emph{hedging problem} of agent $k$ is the extended MVH problem
\begin{equation}
\label{eq:hedging:ind agent}
\E\big[\big(c + \vt\sint S_T -  \Xi^k \big)^2\big] \to \min_{c \in \R, \vt \in \overline{\Theta}(S)}!
\end{equation}
and the \emph{pure investment problem} is the MVH problem
\begin{equation}
\label{eq:pure:inv:sec:quad}
\E[(\vt\sint S_T - 1)^2] \to \min_{\vt \in \overline{\Theta}(S)}!
\end{equation}
For future reference, we set for each $k \in \{1,...,K\}$,
\begin{align}
\label{eq:ell}
\ell(S) &:= \E[(\vartheta^{\MVH}(1) \bullet S_T - 1)^2], \\
\label{eq:exconst}
c_k(S) &:= c(\Xi^k,S), \\
\label{eq:varepsilon2}
\varepsilon_k(S) &:= \varepsilon(\Xi^k,S) = \E[(c_k(S) + \vartheta^{\ex}(\Xi^k,S) \bullet S_T - \Xi^k)^2],
\end{align}
and we henceforth drop the dependence on $S$ in this section. We assume that $\ell = E[(\vt^{{\rm MVH}}(1)\sint S_T - 1)^2]>0$, or equivalently $\vt^{{\rm MVH}}(1)\sint S_T \not \equiv 1$, which can be seen as a weak variant of the condition of `the law of one price' in \v Cerny/Czichowsky~\cite{cerny:czichowsky:22}, that is, a weak no-free-lunch condition on $S$.  Under that assumption,  we obtain the following decomposition.

\begin{proposition}
\label{prop:ind:opt:decomp}
Let $(1,S)$ be a local $L^2$-market with $\eta^k \in \overline{\Theta}(S)$. Suppose that there is a unique solution $\vt^{{\rm MVH}}(1)$ to the pure investment problem \eqref{eq:pure:inv:sec:quad} and $\ell >0$. Then there is a unique solution $\vt^{\rm MVH}(H^k) \in \overline{\Theta}(S)$ to the MVH problem \eqref{eq:lem:ind:mod} if and only if there is a unique solution $(c_k, \vt^{\rm ex}(\Xi^k))$ to the exMVH problem \eqref{eq:MVHex} with $H=\Xi^k$. In this case $\hat \vt^k$ can be decomposed as
\begin{equation}
\label{eq:mvh:prob:hk:sol:decomp}
\hat \vt^k \eq{S} \eta^k - \vt^{\rm ex}(\Xi^k) + (\gamma_k-c_k) \vt^{{\rm MVH}}(1).
\end{equation}

\end{proposition}
The above results shows that $\hat \vt^k$ can be decomposed into the traded endowment $\eta^k$, a hedging component for the total endowment $\Xi^k$ and an investment component proportional to  $\vt^{{\rm MVH}}(1)$.

\subsection{The representative agent}
\label{sec:representative:agent}

We now return to the question of finding a price process $(1,S)$ that leads to an equilibrium market in the sense of Definition \ref{def:equilibrium}. Even though the equilibrium market will generically be incomplete, we show that due to the quadratic structure of the utility functions, a classical idea from equilibrium theory in complete markets extends to the incomplete setting -- the use of a \emph{representative agent}. The latter is a fictitious agent who aggregates the endowment and the utility of all agents, i.e., it owns both $\bar \eta = \sum_{k=1}^K \eta^k$ units of the assets as well as the sum of the non-traded endowments of the agents. Equivalently, the representative agent receives the total endowment $\bar \Xi = \sum_{k=1}^K  \Xi^k$. By the same argument as in \eqref{eq:wealth agent k}, the representative agent can attain the terminal wealth $(\vt - \bar \eta) \sint S_T + \bar \Xi$ by trading with a strategy $\vt \in \TB(S)$. Note that for other utility functions, the representative agent approach breaks down in the incomplete setting.

The utility function of the representative agent is defined by
\begin{equation*}
\bar{U}^Q_\lambda(x)=\sup \bigg\{\sum_{k=1}^K\lambda_k U^Q_k(x_k) : x_1, \ldots, x_K\in \R^d,  \sum_{k =1}^K x_k =x \bigg\},
\end{equation*}
where $\lambda = (\lambda_1, \ldots, \lambda_K) \in \R^K$ is a fixed set of \emph{Negishi weights}. Usually, the convention in the literature is that the Negishi weights are positive and sum up to $1$. We denote by $\bar \gamma := \sum_{k=1}^K \gamma_k$ the aggregate risk tolerance and make the ansatz $\lambda^*_k := \frac{\bar{\gamma}}{\gamma_k}$. Moreover, we set $\bar{U}^Q := \bar{U}^Q_{\lambda^*}$ for simplicity. It is not difficult to check that
\begin{align*}
\bar{U}^Q(x) &=\sup \bigg\{\sum_{k=1}^K \lambda_k (2\gamma_k x_k - x_k^2) : x_1, \ldots, x_K\in \R,  \sum_{k =1}^K x_k =x \bigg\} = 2 \bar \gamma x - x^2,
\end{align*}
so that the utility function of the representative agent is of the same form as the utility function of the individual agents. The representative agent then solves the maximisation problem
\begin{equation}
\label{max:rep:agent}
\E\big[\bar U^Q\big((\vt - \bar \eta)\sint S_T + \bar \Xi\big)\big] \to \max_{\vt \in \overline{\Theta}(S)}!
\end{equation}
From a mathematical perspective, \eqref{max:rep:agent} has the same structure as the individual maximisation problem~\eqref{eq:def:gen:quad:util:indiv:prob}. Thus we get an analogue of Lemma~\ref{lem:ind:opt} for the representative agent. To keep notation concise, for any $H \in L^2(\cF_T)$, we define the martingale generated by $H$ to be $Z_t^H = \E[H \,|\, \cF_t]$, $t \in [0,T]$, and setting $H(\overline{\gamma}) := \overline{\gamma} - \overline{\Xi}$, and write for convenience $Z^{H}(\overline{\gamma}) := Z^{H(\overline{\gamma})}$ and $Z^{HD}(\overline{\gamma}) := Z^{H(\overline{\gamma})D}$. Since $\overline{\gamma}$ is fixed in this section, we write for simplicity $\overline{H} := H(\overline{\gamma})$.
\begin{lemma}
\label{lem:rep:agent}
Let $(1,S)$ be a local $L^2$-market with $\bar \eta \in \overline{\Theta}(S)$.  Then the following are equivalent:
\begin{enumerate}
\item The optimisation problem  \eqref{max:rep:agent} has a unique solution $\bar \vt \in \overline{\Theta}(S)$.
\item For $\bar H = \bar \gamma - \bar \Xi$, the MVH problem
\begin{equation}
\label{eq:lem:rep:agent:mod}
\E[(\vt\sint S_T -  \bar H )^2] \to \min_{\vt \in \overline{\Theta}(S)}!
\end{equation}
has a unique solution $\vt^{\rm MVH}(\bar H) \in \overline{\Theta}(S)$.
\end{enumerate}
If either statement holds, then $\bar\vt \eq{S} \bar \eta + \vt^{\rm MVH}(\bar H)$.
\end{lemma}
In the above result, $\bar H = \bar \gamma - \bar \Xi = \sum_{k =1}^K H^k$ may be interpreted as the aggregate additional wealth that the agents would (collectively) like to obtain in order to reach the aggregate bliss point $\bar \gamma$.

Similarly to Lemma \ref{prop:ind:opt:decomp}, we can also decompose the optimal strategy $\bar \vt$ of the representative agent into a hedging and a pure investment part, where the \emph{hedging problem of the representative agent} is the extended MVH problem
\begin{equation}
\label{eq:hedging:rep agent}
\E\big[\big(c + \vt\sint S_T -  \bar \Xi \big)^2\big] \to \min_{c \in \R, \vt \in \overline{\Theta}(S)}!
\end{equation}

\begin{proposition}
\label{prop:hedge:rep agent}
Let $(1,S)$ be a local $L^2$-market such that $\bar \eta \in \overline{\Theta}(S)$. Suppose that there is a unique solution $\vt^{{\rm MVH}}(1)$ to the pure investment problem \eqref{eq:pure:inv:sec:quad} and $\ell>0$. Then the optimisation problem \eqref{max:rep:agent} of the representative agent has a unique solution $\bar \vt \in \overline{\Theta}(S)$  if and only if the exMVH problem \eqref{eq:MVHex} for $\bar \Xi$ has a unique solution $(c(\bar \Xi), \vt^{\mathrm{ex}}(\bar \Xi))\in  \R \times \overline{\Theta}(S)$. In this case $\bar \vt$ can then be decomposed as
\begin{equation}
\label{eq:prop:hedge:rep agent}
\bar \vt \eq{S} \bar \eta + \left(\bar \gamma - c(\bar \Xi)\right) \vt^{\rm MVH}(1) - \vt^{\mathrm{ex}}(\bar \Xi).
\end{equation}
\end{proposition}

The following result shows that the optimal strategy for the representative agent is given by the sum of the strategies of the individual agents. This result gives  a characterisation of the aggregate demand for the risky assets, which is the key to finding an equilibrium market.
\begin{lemma}
\label{lem:rep:agent:rel to ind}
Let  $(1,S)$ be a local $L^2$-market. Assume that $\eta^1, \ldots, \eta^K \in \overline{\Theta}(S)$ and for each agent $k \in \{1, \ldots, K\}$, the individual optimisation problem \eqref{eq:def:gen:quad:util:indiv:prob} has a unique solution $\hat \vt^k$. Then the optimisation problem \eqref{max:rep:agent} of the representative agent has a unique solution $\bar \vt$ satisfying
\begin{equation}
\label{eq:quad:equil:rep:agent:strat:sum:indiv}
\bar \vt \eq{S} \sum_{k = 1}^K \hat \vt^k.
\end{equation}
\end{lemma}

\subsection{Existence and uniqueness of equilibria}
\label{sec:ex:uniqueness:equilibria:chap:iii}

We can now prove our main result on the existence and uniqueness of quadratic equilibrium markets. In a first step, we show that a necessary condition for the existence of a quadratic equilibrium is that the martingale generated by $\ol H = \ol \gamma -\ol \Xi$ yields a (generalised) local martingale measure for $(1, S)$.

\begin{lemma}
\label{lmm:quad:eq:mark:zs:is:loc:mart}
If $(1,S)= (1, S^{(1)}, S^{(2)})$ is a quadratic equilibrium market, then for each $j \in \{1, \ldots, d_1 + d_2\}$, the process $(Z_t^{\overline{H}} S^j_t)_{0 \leq t \leq T}$ is a local $\Pb$-martingale.
\end{lemma}

In order to obtain also a sufficient condition for the existence of a quadratic equilibrium, we now consider the case that the total endowment $\ol \Xi$ lies below the aggregate bliss point $\ol \gamma$ of the representative agent, i.e., we are in the region, where the utility function of the representative agent is strictly increasing; from a purely mathematical perspective, this assumption can be relaxed, see Theorem \ref{thm:equil}.

For the following theorem, we need the Galtchouk--Kunita--Watanabe decomposition of the square-integrable martingale $Z^{\overline{H}}$ with respect to $M^{(1)}$, the local martingale part $M^{(1)}$ of the financial asset  given in \eqref{eq:def:prim:fin:assets:decomp}. We have
\begin{equation}
\label{eq:gkw:decomp:bar:g:m:1}
Z_t^{\overline{H}} = Z_0^{\overline{H}} + \xi^{\overline{H},(1)} \sint M^{(1)}_t + M^{\overline{H}}_t, \quad 0 \leq t \leq T,
\end{equation}
where \mbox{$\xi^{\overline{H},(1)} \in L^2(M^{(1)})$} and $M^{\overline{H}}$ is a square-integrable $\Pb$-martingale strongly orthogonal to $M^{(1)}$. By Remark \ref{rmk:gkw:decomp:bar:g:m:1:detail:pred:quad:var}, we may choose $\xi^{\overline{H},(1)}$ in \eqref{eq:gkw:decomp:bar:g:m:1} in such a way that $\xi^{\overline{H},i} \in L^2(M^i)$ for each $i \in \{1, \ldots, d_1\}$.

We can now formulate the main existence and uniqueness theorem for quadratic equilibria. It shows in particular that quadratic equilibria are arbitrage-free.

\begin{theorem}
\label{thm:equil:new}
Suppose that $\ol \Xi < \ol \gamma$ $\Pas$ If there exists a quadratic equilibrium $(1, S^{(1)}, S^{(2)})$, it is unique and explicitly given by
\begin{align}
\label{eq:thm:equil:new:S1}
S^j_t &= S_0^j + M^j_t - \int_0^t  \frac{\dd \langle Z^{\overline{H}}, M^j \rangle_s}{ Z_{s-}^{\overline{H}}}= S_0^j + M^j_t -  \sum_{i = 1}^{d_1} \int_0^t  \frac{\xi^{\overline{H},i}_s}{Z_{s-}^{\overline{H}}}\dd \langle M^i, M^j \rangle_s, &&  j \in \{1, \ldots, d_1\}, \\
 \label{eq:thm:equil:new:S2}
S^{d_1+j}_t &= \frac{\E[\overline{H} D^j \mid \cF_t]}{Z_t^{\overline{H}}} = \frac{Z_t^{\overline{H}D,j}}{Z_t^{\overline{H}}} = Z_t^{D} - \frac{\Cov[\overline{\Xi},D \,|\, \cF_t]}{\overline{\gamma} - Z_t^{\overline{\Xi}}}, && j \in \{1, \ldots, d_2\}.
\end{align}
Moreover, in this case the market $(1, S) = (1, S^{(1)}, S^{(2)})$ admits an equivalent local martingale measure $\ol \Q \approx \Pb$ with a square-integrable density $\frac{\dd \ol \Q}{\dd \Pb} := \frac{\overline{H}}{\EX{\overline{H}}}$.

Conversely, the market $(1, S) = (1, S^{(1)}, S^{(2)})$ defined by \eqref{eq:thm:equil:new:S1} and \eqref{eq:thm:equil:new:S2} is a quadratic equilibrium if and only if $S^{(2)}$ is a local $L^2$-semimartingale and for each \mbox{$j \in \{d_1+1, \ldots, d\}$}, we have $e^j \in \overline{\Theta}(S)$.
\end{theorem}

In general, it is rather difficult to check if $S^{(2)}$ given by \eqref{eq:thm:equil:new:S2} is a local $L^2$-semimartingale and $e^j \in \overline{\Theta}(S)$ for $j \in \{d_1+1, \ldots, d\}$. A sufficient (but not necessary) condition for this is if $S^{(2)}$ is an $L^2$-semimartingale because then \mbox{$e^j \in \Theta_{\rm simple}(S) \subseteq \overline{\Theta}(S)$} for $j \in \{d_1+1, \ldots, d\}$. The following result gives some conditions on the primitives that ensure the latter. The first two sets of conditions are particularly easy to check; they are special cases of the more elaborate third condition.

\begin{lemma}
\label{lem:quad equil:suf cond}
Suppose that $\ol \Xi < \ol \gamma$ $\Pas$. Assume that one of the following conditions holds:
\begin{enumerate}
\item $D^j\in L^\infty(\Pb)$ for $j \in \{1, \ldots, d_2\}$.
\item $\bar H, \bar H^{-1} \in L^\infty(\Pb)$.
\item $\bar H \in L^{1-p_1}(\Pb) \cap L^{p_2}(\Pb)$ and $D^j \in L^{2 q_1 q_2}(\Pb)$ for all $j \in \{1, \ldots, d_2\}$, where $p_1, p_2  \in [1, \infty]$ and $1/p_i + 1/q_i = 1$ for $i \in \{1, 2\}$. Note that we can always choose $p_2 \geq 2$ as $\bar H$ is square-integrable by assumption.
\end{enumerate}
Then the process $S^{(2)}$ given by \eqref{eq:thm:equil:new:S2}  is an $L^2$-semimartingale.
\end{lemma}

The following corollary is particularly useful in applications. It is a direct consequence of Lemma \ref{lem:quad equil:suf cond}(ii). It will also be important in our study of linear mean-variance preferences in the next section.
\begin{corollary}
\label{cor:quad:equil}
Assume that the total endowment $\ol \Xi$ is nonnegative and satisfies $
\esssup \ol \Xi < \ol \gamma$. Then there exists a unique quadratic equilibrium $(1, S^{(1)}, S^{(2)})$, which is explicitly given by
\eqref{eq:thm:equil:new:S1} and \eqref{eq:thm:equil:new:S2}. Moreover, $S^{(2)}$ is an $L^2$-semimartingale, and the market $(1, S^{(1)}, S^{(2)})$ an equivalent local martingale measure $\ol \Q \approx \Pb$ with a bounded density $\frac{\dd \ol \Q}{\dd \Pb} = \frac{\overline{H}}{\EX{\overline{H}}}$.
\end{corollary}

In the following example we illustrate an application of the quadratic equilibrium in an incomplete market that is formulated as a perturbation of the complete market setting.

\begin{example}
Fix $\varepsilon \geq 0$ and let $W = (W^1,W^2) = ((W_t^1,W_t^2))_{t\in[0,T]}$ be a two-dimensional Brownian motion. Suppose that the market filtration $\F^{\varepsilon} = (\cF_t^\varepsilon)_{t\in[0,T]}$ is generated and completed by $(W^1,\varepsilon W^2)$, and that there is only one productive asset, i.e., $d_1 = 0$ and $d_2 = 1$.

If $\varepsilon > 0$, the market filtration $\F^{\varepsilon}$ can be equivalently generated by $(W^1,W^2)$, whereas if $\varepsilon = 0$, this recovers the complete market case where the market filtration is generated by $W^1$ alone. Thus, this seems to be the right convention for including both cases at once.

Consider the Ornstein--Uhlenbeck process $X^\varepsilon$ given by
\[
dX_t^\varepsilon = -\kappa X_t^\varepsilon\,dt + \varepsilon\nu\,dW_t^2, \qquad X_0^\varepsilon=x_0,
\]
where $\kappa>0, \nu>0, x_0 \neq 0$, which represents the stochastic volatility factor. We also introduce the integrated stochastic-volatility factor
\[
I_t^\varepsilon := \int_0^t (X_s^\varepsilon)^2\,ds, \qquad t\in[0,T].
\]
Observe that, if $\varepsilon=0$, then $X_t^0=x_0e^{-\kappa t}$, which is deterministic, whereas if $\varepsilon>0$, $X^{\varepsilon}$ is genuinely stochastic through the additional Brownian motion $W^2$.

Fix constants $d_0,d_1,h_0,h_1,q_D,r_D,q_{\overline{H}},r_{\overline{H}}>0$. We choose the terminal dividend of the productive asset $D^\varepsilon$ and the shifted aggregate endowment $\overline H^\varepsilon$ as
\[
D^\varepsilon := d_0 + d_1\Phi(q_D W_T^1) \exp\bigl(-r_DI_T^\varepsilon\bigr) \quad \text{and} \quad \overline H^\varepsilon := h_0 + h_1\Gamma_T \exp\bigl(-r_{\overline{H}}I_T^\varepsilon\bigr),
\]
where $\Phi$ denotes the standard normal distribution function and $\Gamma := \cE(q_{\overline{H}} W^1)$. Note that $\overline{H}^{\varepsilon} \in L^2(\cF_T^{\varepsilon})$ is strictly
positive and $D^\varepsilon\in L^\infty(\cF_T^\varepsilon)$, satisfying in particular Lemma~\ref{lem:quad equil:suf cond}(i).

For $c\geq 0$, $t\in[0,T]$, and $x\in\R$, define
\[
F_c^\varepsilon(t,x) := \E^x\left[ \exp\left( -c\int_0^t (X_s^\varepsilon)^2\,ds \right) \right],
\]
where $\E^x$ denotes expectation conditioned on $X_0^\varepsilon = x$.

The function $F_c^\varepsilon$ is explicitly given by
\[
F_c^\varepsilon(t,x) = \frac{ e^{\kappa t/2} }{ \sqrt{ \cosh(\gamma_c^\varepsilon t) + \frac{\kappa}{\gamma_c^\varepsilon} \sinh(\gamma_c^\varepsilon t) } } \exp\left( -p_c^\varepsilon(t)x^2 \right),
\]
where $\gamma_c^\varepsilon := \sqrt{\kappa^2+2\varepsilon^2\nu^2c}$ and $p_c^\varepsilon(t) := \frac{c\sinh(\gamma_c^\varepsilon t)}{\gamma_c^\varepsilon
\cosh(\gamma_c^\varepsilon t) + \kappa\sinh(\gamma_c^\varepsilon t)}$.

 For $t \in [0,T]$ and $x\in \R$, define
\[
G_0(t,x) := \Phi\left( \frac{q_D x} {\sqrt{1+q_D^2t}} \right)  \quad \text{and} \quad G_{\overline{H}}(t,x) := \Phi\left( \frac{q_D(x+ q_{\overline{H}} t) }{ \sqrt{1+q_D^2t} }\right).
\]
The first function arises from the Gaussian conditional-expectation identity
\[
\E\left[ \Phi(q_D W_T^1) \,\middle|\, \cF_t^\varepsilon \right] = G_0(T-t,W_t^1),
\]
while the exponential tilt of $\Gamma_T$ gives
\[
\E\left[ \Gamma_T\Phi(q_D W_T^1) \,\middle|\,\cF_t^\varepsilon \right] = \Gamma_t G_{\overline{H}}(T-t,W_t^1).
\]
Since $q_{\overline{H}}>0$, we observe that $G_{\overline{H}}(\,\cdot\,,w)>G_0(\,\cdot\,,w)$.

\begin{sloppypar}
Consider now the generated martingales $Z_t^{\overline H,\varepsilon} := \E\left[ \overline H^\varepsilon \,\middle|\, \cF_t^\varepsilon \right]$ and $Z_t^{\overline H D,\varepsilon} := \E\left[\overline H^\varepsilon D^\varepsilon \,\middle|\,\cF_t^\varepsilon \right]$, and for $i\in \{\overline{H},D\}$, set $A_t^{i,\varepsilon} := e^{-r_i I_t^\varepsilon} F_{r_i}^\varepsilon \left( T-t,X_t^\varepsilon \right)$ as well as $\overline{A}_t^{\varepsilon} := e^{-(r_D + r_{\overline{H}}) I_t^\varepsilon} F_{r_D+r_{\overline{H}}}^\varepsilon \left( T-t,X_t^\varepsilon \right)$, $t\in[0,T]$. By the Markov property of $X^\varepsilon$, we have
\end{sloppypar}
\[
\E\left[ e^{-r_{\overline{H}} I_T^\varepsilon} \,\middle|\, \cF_t^\varepsilon \right]
= A_t^{\overline{H},\varepsilon}, \qquad t \in [0,T],
\]
which in turn implies together with the independence of $W^1$ and $W^2$ that
\[
Z_t^{\overline H,\varepsilon} = h_0 + h_1\Gamma_t A_t^{\overline{H},\varepsilon}, \qquad t \in [0,T].
\]
In particular, it follows that $Z_t^{\overline H,\varepsilon} \geq h_0 >0$, so the denominator in the equilibrium pricing formula is uniformly bounded away from zero. Again, the independence of $W^1$ and $W^2$ gives
\[
\E\left[ D^\varepsilon \,\middle|\, \cF_t^\varepsilon \right] = d_0 + d_1A_t^{D,\varepsilon} G_0(T-t,W_t^1).
\]
Furthermore, after expanding, we have
\[
\overline H^\varepsilon D^\varepsilon = h_0d_0 + h_0d_1\Phi(q_D W_T^1)e^{-r_DI_T^\varepsilon} + h_1d_0\Gamma_T e^{-r_{\overline{H}}I_T^\varepsilon} + h_1d_1\Gamma_T\Phi(q_D W_T^1)e^{-(r_{\overline{H}}+r_D)I_T^\varepsilon},
\]
which then yields
\begin{equation*}
Z_t^{\overline H D,\varepsilon}
= h_0\left( d_0+d_1A_{t}^{D,\varepsilon} G_0(T-t,W_t^1) \right) +
h_1\Gamma_t\left( d_0A_{t}^{\overline{H},\varepsilon} + d_1\overline{A}_t^\varepsilon G_{\overline{H}}(T-t,W_t^1) \right).
\end{equation*}
Consequently, after certain simplifications, the equilibrium asset price $S^\varepsilon$ is described by
\[
S_t^\varepsilon = \frac{ Z_t^{\overline H D,\varepsilon}}{ Z_t^{\overline H,\varepsilon}} = d_0 + d_1 \frac{h_0A_{t}^{D,\varepsilon} G_0(T-t,W_t^1) + h_1\Gamma_t \overline{A}_t^\varepsilon G_{\overline{H}}(T-t,W_t^1) }{h_0+h_1\Gamma_tA_{t}^{\overline{H},\varepsilon}}.
\]
It is easy to check that the terminal condition $S_T^\varepsilon = D^\varepsilon$ is satisfied.

To illustrate the effect of the shifted aggregate endowment more explicitly,
define the relative risky-endowment weight
\[
w_t^\varepsilon := \frac{ h_1\Gamma_tA_{t}^{\overline{H},\varepsilon} }{ h_0+h_1\Gamma_t A_{t}^{\overline{H},\varepsilon}},
\]
and the stochastic-volatility interaction factor
\[
J_t^\varepsilon := \frac{\overline{A}_t^\varepsilon}{A_{t}^{\overline{H},\varepsilon} A_{t}^{D,\varepsilon}} = \frac{F_{r_{\overline{H}}+r_D}^{\varepsilon}(T-t,X_t^{\varepsilon})}{F_{r_{\overline{H}}}^{\varepsilon}(T-t,X_t^{\varepsilon}) F_{r_D}^{\varepsilon}(T-t,X_t^{\varepsilon})}.
\]
Then, we have $0<w_t^\varepsilon<1$, and the equilibrium stock price can equivalently be written as
\[
S_t^\varepsilon = d_0 + d_1A_{t}^{D,\varepsilon} \left( (1-w_t^\varepsilon)G_0(T-t,W_t^1) + w_t^\varepsilon J_t^\varepsilon G_{\overline{H}}(T-t,W_t^1) \right).
\]

If $\varepsilon=0$, the future integrated variance  $\int_t^T(X_s^0)^2\,ds$ is deterministic. Consequently, we get
\[
F_{r_{\overline{H}}+r_D}^0
=
F_{r_{\overline{H}}}^0F_{r_D}^0,
\]
and therefore $J_t^0=1$. Then, it follows that
\[
S_t^0 = d_0 + d_1A_{t}^{D,0} \left( (1-w_t^0)G_0(T-t,W_t^1) + w_t^0 G_{\overline{H}}(T-t,W_t^1) \right).
\]
Thus, in the complete-market benchmark, the stochastic-volatility interaction factor disappears. The shifted aggregate endowment may nevertheless continue to affect the equilibrium price through the direct exponential tilt represented by $G_{\overline{H}}$ and the weight $w^0$. This is a complete-market risk-sharing effect and should be distinguished from the additional incomplete-market effect generated by unspanned stochastic volatility.

Now, if $\varepsilon > 0$, the future integrated variance $\int_t^T(X_s^\varepsilon)^2\,ds$ is nonconstant. Since both functions $x\longmapsto e^{-r_{\overline{H}}x}$ and $x\longmapsto e^{-r_Dx}$ are strictly decreasing, they are positively correlated, so that $J_t^\varepsilon > 1$. Then, it follows that
\[
S_t^\varepsilon > \widetilde{S}_t^{\varepsilon} :=  d_0 + d_1A_{t}^{D,\varepsilon} \left( (1-w_t^\varepsilon)G_0(T-t,W_t^1) + w_t^\varepsilon G_{\overline{H}}(T-t,W_t^1) \right).
\]
where $\widetilde{S}^{\varepsilon}$ denotes the price obtained by suppressing the stochastic-volatility interaction $J^\varepsilon$ while keeping the direct endowment tilt unchanged. In other words, in the incomplete-market model, the common dependence of the dividend and the shifted aggregate endowment on the unspanned stochastic-volatility factor generates an additional positive equilibrium-price adjustment.
\end{example}

We conclude this section by giving an example of a setup where an equilibrium market (in the sense of Definition \ref{def:equilibrium}) fails to exist due to integrability issues in the candidate for $S^{(2)}$. We consider the  simplest case of a market with no financial and one productive asset, i.e., $d_1 = 0$ and $d_2 = 1$. The setup is based on the counterexample in \v Cern\' y/Kallsen \cite{cerny:kallsen:08}, which in turn is inspired by the well-known counterexample of Delbaen/Schachermayer \cite{delbaen:schachermayer:98}.

\begin{example}
\label{ex:counter:ex}
After rescaling the time interval $[0,\infty]$ to $[0,T]$, there exists by \cite[Lemma 2.2]{cerny:kallsen:08} a filtered probability space $(\Omega, \cF, \FF = (\cF_t)_{0 \leq t \leq T}, \Pb)$ supporting two probability measures $\bar \Q, \Q'$ and a continuous process $(X_t)_{0 \leq t \leq T}$ null at $0$ with the following properties:

\begin{itemize}
\item[\textbf{1)}] The measures $\bar \Q, \Q'$ are equivalent to $\Pb$, with $\frac{\dd \bar \Q}{\dd \Pb}, \frac{\dd \Q'}{\dd \Pb} \in L^2(\Pb)$.
\item[\textbf{2)}] The process $X$ is a uniformly integrable martingale under $\bar \Q$, and a strict local martingale under $\Q'$. Moreover, $X_T \in L^2(\Pb)$.
\end{itemize}

Fix now some $\bar \gamma > 0$, and suppose that $d_1 = 0$, $d_2 = 1$, $D^1 := X_T$ and $\bar \Xi := \bar \gamma - \frac{\dd \bar \Q}{\dd \Pb}$, so that $\bar H = \frac{\dd \bar \Q}{\dd \Pb} > 0$ $\as{\Pb}$ Then it follows from Theorem \ref{thm:equil} that if a quadratic equilibrium market exists, it must satisfy $S^1_t = \E^{\bar \Q}[X_T \mid \cF_t] = X_t$. However, since $X = 1 \sint X$ is not a $\Q'$-martingale, the strategy $e^1 \equiv 1$ is not admissible by \v Cern\' y/Kallsen \cite[Corollary 2.5]{cerny:kallsen:07}. Therefore, a quadratic equilibrium does not exist in this setup.
\end{example}

\section{Mean-variance preferences}
\label{sec:MV}

In the remainder of this paper, we assume that the filtration $\F$ is continuous, i.e., all $\F$-local martingales are continuous, and complete. For instance, this holds when the filtration is generated by a Brownian motion since the continuity of $\F$ is ensured by the It\^{o}'s representation theorem. Moreover, we consider the economically relevant setting of the terminal endowment $\Xi^k$ for each agent $k$ to be bounded above, and below away from $0$, i.e., $\essinf \Xi^k > 0$. The boundedness assumption on $\Xi^k$ is also in good agreement with the forthcoming quadratic BSDE formulation.

Let us recall the situation of mean-variance preferences, where each agent $k \in \{1,...,K\}$ solves the problem 
\begin{equation}
\label{eq:MVoptprob}
\cU_k^{\MV}(V_T^k(\vartheta)) = U_k^{\MV}\big(\E[V_T^k(\vartheta)],\sqrt{\Var[V_T^k(\vartheta)]} \big) \to \max_{\vartheta \in \overline{\Theta}(S)}!
\end{equation}
where $V_T^k(\vartheta) = (\vartheta - \eta^k) \bullet S_T + \Xi^k$, and the mean-variance utility is defined after~\eqref{eq:MVutility}.

\subsection{Individual optimality}
\label{sec:indopt}

Throughout this section, we fix a local $L^2$-market $(1,S)$ that satisfies~\eqref{eq:def:prim:fin:assets:decomp} and~\eqref{eq:def:prim:prod:assets:divid} and assume that $e^j \in \overline{\Theta}(S)$ for each $j \in \{d_1 + 1,...,d\}$.

The mean-variance optimisation problem~\eqref{eq:MVoptprob} is connected to the quadratic utility problem~\eqref{eq:def:gen:quad:util:indiv:prob} as well as the mean-variance hedging problems~\eqref{eq:MVH}-\eqref{eq:MVHex}. A sufficient condition for the existence and uniqueness (up to $S$-equivalence) of solutions to~\eqref{eq:MVH}-\eqref{eq:MVHex} is given in Proposition~\ref{prop:LOP}, which we state in the following.

\begin{assumption}
\label{ass:NFL}
There exists an equivalent local martingale measure (ELMM) $\Q \approx \Pb$ for $S$ with square-integrable density $\frac{\dd \Q}{\dd \Pb} \in L^2(\Pb)$.
\end{assumption}

The connection between~\eqref{eq:MVoptprob} with~\eqref{eq:def:gen:quad:util:indiv:prob} and~\eqref{eq:MVH},~\eqref{eq:MVHex} arises through the classical notions of mean-variance efficient strategies and the mean-variance efficient frontier.

\begin{definition}
\label{def:MVefficient}
A strategy $\vartheta \in \overline{\Theta}(S)$ is \emph{mean-variance efficient} with respect to $H \in L^2$ if there does not exist any other strategy $\vartheta' \in \overline{\Theta}(S)$ such that
\begin{align}
\label{eq:efficiency1}
\E[\vartheta' \bullet S_T + H] &\geq \E[\vartheta \bullet S_T + H], \\
\label{eq:efficiency2}
\Var[\vartheta' \bullet S_T + H] &\leq \Var[\vartheta \bullet S_T + H],
\end{align}
where one of the inequalities is strict.
\end{definition}

We say that $\vartheta \in \overline{\Theta}(S)$ is \emph{mean-variance efficient for agent $k$} if $\vartheta - \eta^k$ is mean-variance efficient with respect to $H = \Xi^k$. Recalling the notation $V_T^k(\vartheta) = (\vartheta - \eta^k) \bullet S_T + \Xi^k$ from~\eqref{eq:wealth agent k}, this is equivalent to the nonexistence of another strategy $\vartheta' \in \overline{\Theta}(S)$ such that both $\E[V_T^k(\vartheta')] \geq \E[V_T^k(\vartheta)]$ and $\Var[V_T^k(\vartheta')] \leq \Var[V_T^k(\vartheta)]$, with one of the inequalities being strict. The \emph{mean-variance efficient frontier} for agent $k$ is defined as the set
\begin{equation}
\label{eq:efficientfrontier}
\cE^k := \{(\E[V_T^k(\vartheta)],\sqrt{\Var[V_T^k(\vartheta)]}): \vartheta \in \overline{\Theta}(S) \,\, \text{is mean-variance efficient for agent} \,\, k \}.
\end{equation}

We begin our analysis by showing that any solution to the mean-variance problem~\eqref{eq:MVoptprob} must be mean-variance efficient.

\begin{lemma}
\label{lem:MVefficiency}
Suppose the market $(1,S)$ satisfies Assumption~\ref{ass:NFL}. Any solution $\hat{\vartheta}^k \in \overline{\Theta}(S)$ to~\eqref{eq:MVoptprob} must be mean-variance efficient for agent $k$.
\end{lemma}

\begin{proof}
Seeking a contradiction, suppose that $\hat{\vartheta}^k$ is not mean-variance efficient for agent $k$. Then, there exists some $\vartheta' \in \overline{\Theta}(S)$ satisfying
\begin{align*}
\E[\vartheta' \bullet S_T + \Xi^k] &\geq \E[(\hat{\vartheta}^k - \eta^k) \bullet S_T + \Xi^k] = \E[V_T^k(\hat{\vartheta}^k)], \\
\Var[\vartheta' \bullet S_T + \Xi^k] &\leq \Var[(\hat{\vartheta}^k - \eta^k) \bullet S_T + \Xi^k] = \Var[V_T^k(\hat{\vartheta}^k)],
\end{align*}
where one of the inequalities is strict. Since $U_k^{\MV}(\mu,\sigma)$ is strictly increasing in $\mu \in \R$ and strictly decreasing in $\sigma \in \R_+$, we get
\[
U_k^{\MV}\Big(\E[V_T^k(\hat{\vartheta}^k)],\sqrt{\Var[V_T^k(\hat{\vartheta}^k)]}\Big) < U_k^{\MV}\Big(\E[V_T^k(\widetilde{\vartheta})], \sqrt{\Var[V_T^k(\widetilde{\vartheta})]}\Big),
\]
where $\widetilde{\vartheta} := \vartheta' + \eta^k \in \overline{\Theta}(S)$, contradicting the optimality of $\hat{\vartheta}^k$ for~\eqref{eq:MVoptprob}.
\end{proof}

We now find the set of mean-variance efficient strategies by relating them to the solutions of the quadratic utility problem~\eqref{eq:def:gen:quad:util:indiv:prob} and the MVH problem for the payoff $H^k(\gamma_k) := \gamma_k - \Xi^k$. The proofs of this and other subsequent results can be found in Appendix \ref{app:auxiliary}. We also note that under Assumption~\ref{ass:NFL} and the continuity of the filtration, the market $(1,S)$ is continuous and admits an ELMM with square-integrable density. Thus, by Delbaen/Schachermayer~\cite[Theorem 1.3]{DS1997}, the \emph{variance-optimal martingale measure} (VOMM) $\Q^*$ exists and is equivalent to $\Pb$ (as opposed to generally being a signed measure). Therefore, together with Corollary~\ref{cor:UOVP link UOGP}, this gives
\begin{equation}
\label{eq:ckpositive}
c_k = \E^{\Q^*}[\Xi^k] \geq \essinf \Xi^k > 0,
\end{equation}
which would imply in the subsequent results that $\gamma_k > c_k$ is positive.

\begin{proposition}
\label{prop:efficientfrontier}
Suppose the market $(1,S)$ satisfies Assumption~\ref{ass:NFL} and $S$ is not a local $\Pb$-martingale. For $\vartheta \in \overline{\Theta}(S)$, the following statements are equivalent:
\begin{enumerate}[(a)]
\item $\vartheta$ is mean-variance efficient for agent $k$.
\item There exists a unique $\gamma_k > c_k > 0$ such that $\vartheta =_S \vartheta^k(\gamma_k)$, where
\[
\vartheta^k(\gamma_k) := (\gamma_k - c_k)\vartheta^{\MVH}(1) + \eta^k - \vartheta^{\ex}(\Xi^k).
\]
\item There exists a unique $\gamma_k > c_k > 0$ such that $\vartheta - \eta^k$ is the unique solution to the MVH problem~\eqref{eq:MVH} for the payoff $H^k(\gamma_k) := \gamma_k - \Xi^k$.
\item There exists a unique $\gamma_k > c_k > 0$ such that $\vartheta$ is the unique solution to the quadratic utility problem~\eqref{eq:def:gen:quad:util:indiv:prob} with risk tolerance $\gamma_k$.
\end{enumerate}
The constant $\gamma_k > c_k > 0$ can be chosen identically in statements (b)-(d).
\end{proposition}

With the help of Proposition~\ref{prop:efficientfrontier}, we obtain an explicit formula for the mean-variance efficient frontier $\cE^k$ in terms of the triplet $(\ell,c_k,\varepsilon_k) \in (0,1] \times \R_+ \times \R_+$.

\begin{corollary}
\label{cor:efficientfrontier}
Suppose the market $(1,S)$ satisfies Assumption~\ref{ass:NFL}. The mean-variance efficient frontier for agent $k$ is given by
\begin{equation}
\label{eq:efficientfrontierexpl}
\cE^k = \{(\mu_k(\gamma_k),\sigma_k(\gamma_k) = (\ell\,c_k + (1-\ell)\gamma_k, \sqrt{\varepsilon_k + \ell(1-\ell)(\gamma_k-c_k)^2}): \gamma_k \geq c_k \}.
\end{equation}
\end{corollary}

In the following result, we prove the stronger statement that a strategy maximises~\eqref{eq:MVoptprob} over $\overline{\Theta}(S)$ if and only if it maximises~\eqref{eq:MVoptprob} over the set of mean-variance efficient strategies for agent $k$. In other words, we discard all other strategies not only as candidate solutions to~\eqref{eq:MVoptprob}, but also as competitors to the candidate solution.

\begin{corollary}
\label{cor:MVoptprobequiv}
Suppose the market $(1,S)$ satisfies Assumption~\ref{ass:NFL}. A strategy $\hat{\vartheta}^k \in \overline{\Theta}(S)$ is a solution to the maximisation problem~\eqref{eq:MVoptprob} if and only if $\hat{\vartheta}^k =_S \vartheta^k(\hat{\gamma}_k)$, where $\hat{\gamma}_k \geq c_k$ is a maximiser for the problem
\begin{equation}
\label{eq:MVoptprobeff}
U_k^{\MV}(\mu_k(\gamma_k),\sigma_k(\gamma_k)) \to \max_{\gamma_k\geq c_k}!,
\end{equation}
and $\mu_k$ and $\sigma_k$ are given in~\eqref{eq:efficientfrontierexpl}.
\end{corollary}

\begin{remark}
\label{rem:degenerate}
We note that $\ell = 1$ leads to a degenerate case, since the functions $\mu_k$ and $\sigma_k$ are constant and every $\gamma_k \geq c_k$ is a solution to~\eqref{eq:MVoptprobeff}. Since we also have (by uniqueness) in this case that $\vartheta^{\MVH}(1) = 0$, it follows by Proposition~\ref{prop:efficientfrontier} that the only mean-variance efficient strategy for agent $k$ is $\hat{\vartheta}^k = \eta^k - \vartheta^{\ex}(\Xi^k)$, and hence this is the unique solution to~\eqref{eq:MVoptprob}.
\end{remark}

We proceed to study an abstract version of~\eqref{eq:MVoptprobeff} as follows. Given constants $(\ell,c,\varepsilon) \in (0,1) \times \R_+ \times \R_+$ and a mean-variance utility function $U^{\MV}$, we consider the optimisation problem
\begin{equation}
\label{eq:MVoptprobabstract}
U^{\MV}(\ell\,c + (1-\ell)\gamma, \sqrt{\varepsilon + \ell(1-\ell)(\gamma-c)^2}) \to \max_{\gamma\geq c}!
\end{equation}
Aside from the generic choice of $U^{\MV}$, the main difference between~\eqref{eq:MVoptprobabstract} and~\eqref{eq:MVoptprobeff} is that we consider arbitrary constants $(\ell,c,\varepsilon) \in (0,1) \times \R_+ \times \R_+$, as opposed to the constants $(\ell,c_k,\varepsilon_k)$ defined in~\eqref{eq:ell}-\eqref{eq:varepsilon2} for the given price process $S$ and payoff $\Xi^k$. This allows us to study not only the existence and uniqueness of a solution $\hat{\gamma}_{U^{\MV}}$, but also whether $\hat{\gamma}_{U^{\MV}}$ depends continuously on $(\ell,c,\varepsilon)$. Although the latter is not necessary for solving the individual optimisation problems with respect to $S$, this property will play a crucial role when we later study the equilibrium problem, since $S$ is not known a priori, and hence neither are the parameters $(\ell,c_k,\varepsilon_k)$.

\begin{proposition}
\label{prop:haty}
There exists a continuous map $\hat{\gamma}_{U^{\MV}}: (0,1) \times \R_+ \times \R_+ \to \R_+$ such that $\hat{\gamma}_{U^{\MV}}(\ell,c,\varepsilon)$ is the unique solution to~\eqref{eq:MVoptprobabstract}, for every $(\ell,c,\varepsilon) \in (0,1) \times \R_+ \times \R_+$. Moreover, the map $\hat{\gamma}_{U^{\MV}}$ satisfies the bounds
\[
\frac{\lambda^{\min}}{\ell} \leq \hat{\gamma}_{U^{\MV}}(\ell,c,\varepsilon) - c \leq \frac{\lambda^{\max}}{\ell}, \qquad (\ell,c,\varepsilon) \in (0,1) \times \R_+ \times \R_+.
\]
\end{proposition}

We can now combine Propositions~\ref{prop:efficientfrontier} and~\ref{prop:haty}, Corollary~\ref{cor:efficientfrontier} and Remark~\ref{rem:degenerate}, to yield existence and uniqueness of solutions to~\eqref{eq:MVoptprob}.

\begin{theorem}
\label{thm:MVsol}
Suppose the market $(1,S)$ satisfies Assumption~\ref{ass:NFL} and $S$ is not a local $\Pb$-martingale. Then, for each $k \in \{1,...,K\}$, the unique solution to~\eqref{eq:MVoptprob} is given by
\[
\hat{\vartheta}^k := \vartheta^k(\hat{\gamma}_{U_k^{\MV}}(\ell,c_k,\varepsilon_k)) = (\hat{\gamma}_{U_k^{\MV}}(\ell,c_k,\varepsilon_k) - c_k)\vartheta^{\MVH}(1) + \eta^k - \vartheta^{\ex}(\Xi^k),
\]
where the continuous map $\hat{\gamma}_{U_k^{\MV}}: (0,1) \times \R_+ \times \R_+ \to \R_+$ is given by Proposition~\ref{prop:haty}, satisfying the bounds
\[
\frac{\lambda_k^{\min}}{\ell} \leq \hat{\gamma}_{U_k^{\MV}}(\ell,c_k,\varepsilon_k) - c_k \leq \frac{\lambda_k^{\max}}{\ell}.
\]
Moreover, $\hat{\vartheta}^k$ is the unique solution to the quadratic utility problem~\eqref{eq:def:gen:quad:util:indiv:prob} with risk tolerance $\hat{\gamma}_{U_k^{\MV}}(\ell,c_k,\varepsilon_k)$, and $\hat{\vartheta}^k - \eta^k$ is the unique solution to the MVH problem with payoff $H^k(\hat{\gamma}_{U_k^{\MV}}(\ell,c_k,\varepsilon_k)) = \hat{\gamma}_{U_k^{\MV}}(\ell,c_k,\varepsilon_k) - \Xi^k$.
\end{theorem}

\begin{proof}
Fix $k \in \{1,...,K\}$. By Proposition~\ref{prop:haty}, there exists a continuous map $\hat{\gamma}_{U_k^{\MV}}: (0,1) \times \R_+ \times \R_+ \to \R_+$ such that $\hat{\gamma}_{U_k^{\MV}}(\ell,c_k,\varepsilon_k)$ is the unique solution to the problem~\eqref{eq:MVoptprobabstract} and satisfies the bounds $\hat{\gamma}_{U_k^{\MV}}(\ell,c_k,\varepsilon_k) \in [c_k + \lambda_k^{\min}/\ell, c_k + \lambda_k^{\max}/\ell]$. Thus, by Corollary~\ref{cor:MVoptprobequiv}, the strategy $\hat{\vartheta}^k \in \overline{\Theta}(S)$ defined by
\begin{equation}
\label{eq:uniqueoptstrat}
\hat{\vartheta}^k := \vartheta^k(\hat{\gamma}_{U_k^{\MV}}(\ell,c_k,\varepsilon_k)) = (\hat{\gamma}_{U_k^{\MV}}(\ell,c_k,\varepsilon_k) - c_k)\vartheta^{\MVH}(1) + \eta^k - \vartheta^{\ex}(\Xi^k),
\end{equation}
is the unique solution to~\eqref{eq:MVoptprob}. The last two statements follow by the implications (b) $\Rightarrow$ (c), (d), respectively, of Proposition~\ref{prop:efficientfrontier}.
\end{proof}

\subsection{Connections to quadratic equilibria}
\label{sec:connections}

Having characterised the solution to the individual optimisation problem~\eqref{eq:MVoptprob} in Theorem~\ref{thm:MVsol}, we now turn to finding a mean-variance equilibrium market $(1,S)$ in the sense of Definition~\ref{def:equilibrium}. The first step is to show that any mean-variance equilibrium $(1,S)$ satisfying Assumption~\ref{ass:NFL} must also be a quadratic equilibrium with respect to some risk tolerances $(\gamma_k)_{k=1}^{K}$. 

To ensure that quadratic equilibria satisfy Assumption~\ref{ass:NFL}, notice that
\begin{equation}
\label{eq:esssup}
0 < \esssup \overline{\Xi} < \infty,
\end{equation}
which follows by the fact that each $\Xi^k$ is bounded above and below away from $0$. 


\begin{lemma}
\label{thm:Qequilibria}
For $\overline{\gamma} > \esssup \overline{\Xi}$, the market $(1,S(\overline{\gamma}))$ given by
\begin{align}
\label{eq:financialeq}
S_t^j(\overline{\gamma}) &= S_0^j + M_t^j - \int_0^t \frac{\dd \langle Z^H(\overline{\gamma}),M^j \rangle_s}{Z_{s}^H(\overline{\gamma})}, &j \in \{1,...,d_1\}, \\
\label{eq:productiveeq}
S_t^{d_1+j}(\overline{\gamma}) &= \frac{\E[H(\overline{\gamma})D^j \,|\, \cF_t]}{Z_t^H(\overline{\gamma})} = \frac{Z_t^{HD,j}(\overline{\gamma})}{Z_t^H(\overline{\gamma})} = Z_t^{D} - \frac{\Cov[\overline{\Xi},D \,|\, \cF_t]}{\overline{\gamma} - Z_t^{\overline{\Xi}}}, &j \in \{1,...,d_2\},
\end{align}
comprises the unique quadratic equilibrium for any choice of risk tolerances $(\gamma_k)_{k=1}^K$ such that $\sum_{k=1}^K \gamma_k = \overline{\gamma}$. Moreover, $S^{(2)}(\overline{\gamma})$ is an $L^2$-semimartingale and $S(\overline{\gamma})$ admits an ELMM $\Q(\overline{\gamma})$ with bounded density $\frac{\dd \Q(\overline{\gamma})}{\dd \Pb} := \frac{H(\overline{\gamma})}{\E[H(\overline{\gamma})]}$.
\end{lemma}

\begin{lemma}
\label{lem:eqconnection}
Suppose that $(1,S)$ is a mean-variance equilibrium market satisfying Assumption~\ref{ass:NFL}. Then, there exist $(\gamma_k)_{k=1}^K \in \R^K$ so that $S = S(\overline{\gamma})$, where $\overline{\gamma} = \sum_{k=1}^{K} \gamma_k$ and $(1,S(\overline{\gamma}))$ is the unique quadratic equilibrium defined by~\eqref{eq:financialeq}-\eqref{eq:productiveeq}.
\end{lemma}

\begin{proof}[Proof of Lemma~\ref{lem:eqconnection}]
Since $(1,S)$ is a mean-variance equilibrium in the sense of Definition~\ref{def:equilibrium}, it satisfies in particular condition (1) under mean-variance utilities, i.e., there exists for each agent $k \in \{1,...,K\}$ a unique solution $\hat{\vartheta}^k$ to~\eqref{eq:MVoptprob}. Theorem~\ref{thm:MVsol} yields $\hat{\vartheta}^k =_S \vartheta^k(\hat{\gamma}_{U_k^{\MV}}(\ell,c_k,\varepsilon_k))$, where $\vartheta^k(\hat{\gamma}_k)$ is given by Proposition~\ref{prop:efficientfrontier}(b) and $\hat{\gamma}_{U_k^{\MV}}(\ell,c_k,\varepsilon_k)$ is given by Proposition~\ref{prop:haty}. Thus, by the equivalent properties (b) and (d) of Proposition~\ref{prop:efficientfrontier}, $\vartheta^k(\hat{\gamma}_{U_k^{\MV}}(\ell,c_k,\varepsilon_k))$ is also the unique solution to the quadratic utility problem~\eqref{eq:def:gen:quad:util:indiv:prob} with risk tolerance $\hat{\gamma}_{U_k^{\MV}}(\ell,c_k,\varepsilon_k)$. This gives condition (1) of Definition~\ref{def:equilibrium} under such quadratic utilities. Therefore, $(1,S)$ is also a quadratic equilibrium with risk tolerances $(\hat{\gamma}_{U_k^{\MV}}(\ell,c_k,\varepsilon_k))_{k=1}^K$, when noting that conditions (2) and (3) of Definition~\ref{def:equilibrium} are inherited from $(1,S)$ being a mean-variance equilibrium to start with. The result then follows by Lemma~\ref{thm:Qequilibria}.
\end{proof}

Lemma~\ref{lem:eqconnection} points us to look for a mean-variance equilibrium within the class of quadratic equilibria for some choice of risk tolerances. Now, for fixed $\overline{\gamma} > \esssup \overline{\Xi}$, we define $\Psi(\overline{\gamma}) := \sum_{k=1}^K \hat{\gamma}_{U_k^{\MV}}(\ell,c_k,\varepsilon_k)(S(\overline{\gamma}))$, where we recall $\hat{\gamma}_{U_k^{\MV}}(\ell,c_k,\varepsilon_k)(S(\overline{\gamma}))$ as defined in Theorem~\ref{thm:MVsol}. In light of Theorems~\ref{thm:MVsol} and~\ref{thm:Qequilibria}, $\Psi(\bar \gamma)$ can be seen as an aggregate risk tolerance  ``implied'' by the strategies of agents with mean--variance preferences under the price process $S(\bar \gamma)$. Indeed, as we show later in Theorem~\ref{thm:fixedpointeq}, the search for a (general) mean-variance equilibrium of the form $S(\bar \gamma)$ can be posed as a fixed point problem for the map $\Psi$. The following result gives useful bounds on $\Psi$ that will help us narrow the search for a fixed point to an interval 
$[\overline{\gamma}_{\min},\overline{\gamma}_{\max}]$.

\begin{corollary}
\label{cor:convexcomb}
For $\overline{\gamma} > \esssup \overline{\Xi}$, the convex combination $\ell(S(\overline{\gamma}))\Psi(\overline{\gamma}) + (1 - \ell(S(\overline{\gamma})))\overline{\gamma}$ satisfies the bounds
\begin{equation}
\label{eq:convexcomb}
\overline{\gamma}_{\min} \leq \ell(S(\overline{\gamma}))\Psi(\overline{\gamma}) + (1 - \ell(S(\overline{\gamma})))\overline{\gamma} \leq \overline{\gamma}_{\max},
\end{equation}
where
\begin{equation}
\label{eq:gammabounds}
\overline{\gamma}_{\min} = \E[\overline{\Xi}] + \sum_{k=1}^{K} \lambda_k^{\min} \quad \text{and} \quad \overline{\gamma}_{\max} = \E[\overline{\Xi}] + \sum_{k=1}^{K} \lambda_k^{\max}.
\end{equation}
\end{corollary}

In the special case of linear mean-variance preferences, we have $\lambda^{\min} = \lambda^{\max} = \lambda$, and therefore $\overline{\gamma} = \E[\overline{\Xi}] + \sum_{k=1}^{K} \lambda_k$ is the (explicit) aggregate risk tolerance that secures a linear mean-variance equilibrium.
Henceforth, we consider aggregate risk tolerances $\overline{\gamma}$ in the compact interval $[\overline{\gamma}_{\min},\overline{\gamma}_{\max}]$, and impose for consistency the following assumption.

\begin{assumption}
\label{ass:gammalowerbd}
We have $\overline{\gamma}_{\min} = \E[\overline{\Xi}] + \sum_{k=1}^{K} \lambda_k^{\min} > \esssup \overline{\Xi}$.
\end{assumption}

In the remainder of this section, we study the dynamic dependence of $S(\overline{\gamma})$ on $\overline{\gamma}$. To this end, it is important to exclude a special case. Recall from~\eqref{eq:def:prim:fin:assets:decomp} that the local martingale parts of the financial assets are denoted by $M^j$ for $j \in \{1,...,d_1\}$ and define the (square-integrable) martingales $M^{d_1+j} := Z_t^{D,j} = \E[D^{j} \,|\, \cF_t]$, $t \in [0,T]$, for $j \in \{1,...,d_2\}$. The following assumption essentially ensures that $S(\overline{\gamma})$ is not a local $\Pb$-martingale for any $\overline{\gamma} \in [\overline{\gamma}_{\min}, \overline{\gamma}_{\max}]$.

\begin{assumption}
\label{ass:orthogonal}
There is $j \in \{1,...,d\}$ such that the local martingales $M^j$ and $Z^H(\overline{\gamma}_{\min})$ are not strongly orthogonal.
\end{assumption}

The following result shows that under the above assumption the map $\overline{\gamma} \mapsto S(\overline{\gamma})$ is injective.
	
\begin{lemma}
\label{lmm:distinct:proc:bar:gamma}
\begin{sloppypar}
Suppose that Assumption~\ref{ass:orthogonal} holds. Then for any $\overline{\gamma} > \esssup\Xi$ and $\gamma_1, \ldots, \gamma_K \in \R$, $(1,S(\overline{\gamma}))$ is a quadratic equilibrium with respect to the risk tolerances $(\gamma_k)_{k=1}^K$ if and only if $\overline{\gamma} = \sum_{k=1}^K  \gamma_k$. In particular, for $\overline{\gamma}' > \esssup\overline{\Xi}$, we have $S(\overline{\gamma}) = S(\overline{\gamma}')$ if and only if $\overline{\gamma} = \overline{\gamma}'$.
\end{sloppypar}
\end{lemma}

To study the map $\overline{\gamma} \mapsto S(\overline{\gamma})$, we first write the local martingale parts of~\eqref{eq:financialeq} and~\eqref{eq:productiveeq} with respect to local martingales that do not depend on $\overline{\gamma}$. To this end, we write the reference processes $\hat{Z}^{H} := Z^{H}(\overline{\gamma}_{\min})$ and $\hat{Z}^{HD} := Z^{HD}(\overline{\gamma}_{\min})$, which are $1$- and $d_2$-dimensional, respectively.

Noting for every $\overline{\gamma} \in [\overline{\gamma}_{\min},\overline{\gamma}_{\max}]$ the relationships $Z^H(\overline{\gamma}) = \overline{\gamma} - \overline{\gamma}_{\min} + \hat{Z}^H$ and $
Z^{HD}(\overline{\gamma}) = (\overline{\gamma} - \overline{\gamma}_{\min})Z^D + \hat{Z}^{HD}$, a simple manipulation of the finite variation and local martingale parts of $S^{(1)}(\overline{\gamma})$ and $S^{(2)}(\overline{\gamma})$ in~\eqref{eq:financialeq} and~\eqref{eq:productiveeq}, respectively, leads us to the following decomposition.

\begin{lemma}
\label{lem:Qequilibrium2}
Let $\overline{\gamma} \in [\overline{\gamma}_{\min},\overline{\gamma}_{\max}]$. Then, the market $(1,S(\overline{\gamma}))$ admits the decomposition
\begin{equation}
\label{eq:xidecomp1}
S_t(\overline{\gamma}) = S_0(\overline{\gamma}) + M_t(\overline{\gamma}) + A_t(\overline{\gamma}), \qquad  t \in [0,T], 
\end{equation}
where the local martingale part $M(\overline{\gamma}) \in \cM_{\loc}^2(\R^d)$ is given by $M^{(1)}(\overline{\gamma}) \equiv M^{(1)}$ and
\begin{equation}
\label{eq:LMeq}
M_t^{(2)}(\overline{\gamma}) = \int_0^t \Big(\frac{(\overline{\gamma} - \overline{\gamma}_{\min})\dd Z_s^D + \dd \hat{Z}_s^{HD}}{\overline{\gamma} - \overline{\gamma}_{\min} + \hat{Z}_s^H}
- \frac{((\overline{\gamma} - \overline{\gamma}_{\min})Z_s^D + \hat{Z}_s^{HD})\,\dd \hat{Z}_s^H}{(\overline{\gamma} - \overline{\gamma}_{\min} + \hat{Z}_s^H)^2}\Big), \qquad t \in [0,T],
\end{equation}
$S_0^{(1)}(\overline{\gamma}) = S_0^{(1)}$, and $S_0^{(2)}(\overline{\gamma})$ and the finite variation part $A(\overline{\gamma})$ are given by
\begin{equation}
\label{eq:FVeq}
S_0^{(2)}(\overline{\gamma}) = \frac{(\overline{\gamma} - \overline{\gamma}_{\min})Z_0^D + \hat{Z}_0^{HD}}{\overline{\gamma} - \overline{\gamma}_{\min} + \hat{Z}_0^{H}} \quad \text{and} \quad A_t(\overline{\gamma}) = - \int_0^t \frac{\dd \langle M(\overline{\gamma}), \hat{Z}^H \rangle_s}{\overline{\gamma} - \overline{\gamma}_{\min} + \hat{Z}_s^H}, \qquad t \in [0,T].
\end{equation}
\end{lemma}

Define the $\R^{d_3}$-valued local martingale $N := (M^{(1)},Z^{D},\hat{Z}^{HD},\hat{Z}^{H})$, where $d_3 := d_1 + 2d_2 + 1$. Using results of Delbaen et al.~\cite{delbaen:al:97}, we can explicitly anatomise the local martingale part $M(\overline{\gamma})$ and the finite variation part $A(\overline{\gamma})$ in~\eqref{eq:FVeq} by stochastic integral representations with respect to the corresponding integrators $N$ and $\langle N \rangle$, which do not depend on $\overline{\gamma}$.

\begin{lemma}
\label{lem:xilambdadecomp}
Let $\overline{\gamma} \in [\overline{\gamma}_{\min},\overline{\gamma}_{\max}]$.
\begin{enumerate}
\item There exists a unique (up to $N$-equivalence) process $\xi(\overline{\gamma}) \in L_{\loc}^2(N; \R^{d_3 \times d})$ such that $M(\overline{\gamma}) = \xi(\overline{\gamma}) \bullet N$. Moreover, the family $\{\xi(\overline{\gamma}): \overline{\gamma} \in [\overline{\gamma}_{\min},\overline{\gamma}_{\max}]\}$ can be chosen such that
\begin{equation}
\label{eq:xifullrank}
\Pb[\xi_t(\overline{\gamma}) \,\, \text{has full rank for all} \,\, t \in [0,T] \,\, \text{and} \,\, \overline{\gamma} \in [\overline{\gamma}_{\min},\overline{\gamma}_{\max}]] = 1.
\end{equation}
\item There exists a unique (up to $M(\overline{\gamma})$-equivalence) process $\lambda(\overline{\gamma}) \in L^2(M(\overline{\gamma}); \R^d)$ such that 
\begin{align}
\label{eq:lambdadecomp1}
S_t(\overline{\gamma}) &= S_0(\overline{\gamma}) + \int_0^t \big(\dd\langle M(\overline{\gamma}) \rangle_s\lambda_s(\overline{\gamma}) + \dd M(\overline{\gamma})_s \big) \\
\label{eq:lambdadecomp2}
&= S_0(\overline{\gamma}) + \int_0^t \big(\xi_s^\top(\overline{\gamma})\dd \langle N \rangle_s \xi_s(\overline{\gamma}) 
\lambda_s(\overline{\gamma}) + \xi_s(\overline{\gamma}) \dd N_s \big), \qquad t \in [0,T].
\end{align}
We call~\eqref{eq:lambdadecomp1} the structure-condition (SC) decomposition of $S(\overline{\gamma})$; see~\cite{delbaen:al:97}.
\end{enumerate}
\end{lemma}

\subsection{Existence of equilibria}
\label{sec:fixedpoint}

Having dissected the dependence of $S(\overline{\gamma})$ on $\overline{\gamma}$, we will now look for an equilibrium by finding a fixed point $\overline{\gamma}$ of the map $\overline{\gamma} \mapsto \Psi(\overline{\gamma})$ in the range $[\overline{\gamma}_{\min},\overline{\gamma}_{\max}]$. We start by introducing some notation. Denote by $\cS \subseteq L_{\loc}^2(\R^{d_3 \times d}) \times \cP(\R^{d})$ the set of pairs $(\xi,\lambda)$ such that $\xi$ has full rank $\Pb \ptimes I^N$-a.e. (with $I_t^N := \trr \langle N \rangle_t$, $t \in [0,T]$), $\lambda \in L_{\loc}^2(\xi \bullet N)$ and for some $S_0 \in \R^d$, the process $(S_t^{(\xi,\lambda)})_{t\in[0,T]}$ defined as in~\eqref{eq:lambdadecomp2} by the SC decomposition 
\begin{equation}
\label{eq:SC}
S_t^{(\xi,\lambda)} = S_0 + \int_0^t (\xi_s^\top\dd \langle N \rangle_s \xi_s \lambda_s + \xi_s \dd N_s), \qquad t \in [0,T],
\end{equation}
admits an ELMM $\Q = \Q(\xi,\lambda)$ that satisfies $R_2(\Pb)$. More precisely, $\cS$ is the quotient space that results from identifying pairs $(\xi^1,\lambda^1) \sim (\xi^2,\lambda^2)$ where $\xi^1 =_{N} \xi^2$ and $\lambda^1 =_{\xi^1 \bullet N} \lambda^2$. For $C \geq 1$, we also define the set $\cS(C)$ of pairs $(\xi,\lambda) \in \cS$ such that there exists at least one such measure $\Q(\xi,\lambda)$ that satisfies $R_2(\Pb)$ with constant $C$ on the right-hand side of~\eqref{eq:R2P}. To this end, we can write $\Psi = \psi_3 \circ \psi_2 \circ \psi_1$, where the maps $\psi_i$, $i = 1,2,3$, are given as follows:
\begin{itemize}
\item $\psi_1: [\overline{\gamma}_{\min},\overline{\gamma}_{\max}] \to \cS$ is defined by $\psi_1(\overline{\gamma}) := (\xi(\overline{\gamma}),\lambda(\overline{\gamma}))$, $\overline{\gamma} \in [\overline{\gamma}_{\min},\overline{\gamma}_{\max}]$, where $(\xi(\overline{\gamma}),\lambda(\overline{\gamma}))$ are the coefficients associated to a SC decomposition for $S(\overline{\gamma})$ in~\eqref{eq:lambdadecomp2}. Note that for every $\overline{\gamma} \in [\overline{\gamma}_{\min},\overline{\gamma}_{\max}]$, the density process $Z^{H}(\overline{\gamma})$ satisfies $R_2(\Pb)$ due to the trivial bounds $0 < \overline{\gamma} - \esssup \overline{\Xi} \leq Z^H(\overline{\gamma}) \leq \overline{\gamma}$.
\item $\psi_2: \cS \to (0,1) \times (\R_+ \times \R_+)^K$ is defined by $\psi_2(\xi,\lambda) := (\ell,(c_k,\varepsilon_k)_{k=1}^K)(S^{(\xi,\lambda)})$, $(\xi,\lambda) \in \cS$, where $(\ell,(c_k,\varepsilon_k)_{k=1}^K)(S^{(\xi,\lambda)})$ are given via~\eqref{eq:ell}-\eqref{eq:varepsilon2} in terms of $S^{(\xi,\lambda)}$ in~\eqref{eq:SC}.
\item $\psi_3: (0,1) \times (\R_+ \times \R_+)^K \to \R_+$ is defined by $\psi_3((\ell,(c_k,\varepsilon_k)_{k=1}^K)) = \sum_{k=1}^{K} \hat{\gamma}_{U_k^{\MV}}(\ell,c_k,\varepsilon_k)$, $(\ell,(c_k,\varepsilon_k)_{k=1}^K) \in (0,1) \times (\R_+ \times \R_+)^K$, where the map $\hat{\gamma}_{U_k^{\MV}}: (0,1) \times \R_+ \times \R_+ \to \R_+$ is given by Proposition~\ref{prop:haty} for each $k \in \{1,...,K\}$.
\end{itemize}
Note that we indeed have $\Psi(\bar \gamma) = (\psi_3 \circ \psi_2 \circ \psi_1)(\bar \gamma)$ for $\bar \gamma \in [\overline{\gamma}_{\min},\overline{\gamma}_{\max}]$; see the definition of $\Psi$ given below Lemma \ref{lem:eqconnection}.

We are ready to give the first main result of this paper, which reduces the problem of looking for a mean-variance equilibrium to finding a fixed point $\bar \gamma$ of $\Psi$. As mentioned below Lemma \ref{lem:eqconnection}, we may interpret $\Psi(\overline{\gamma})$ as the aggregate risk tolerance implied by the optimal strategies of the agents with respect to $S(\overline{\gamma})$ and $(U_k^{\MV})_{k=1}^{K}$, when we instead view those strategies as the solutions to quadratic utility problems of the form~\eqref{eq:def:gen:quad:util:indiv:prob}. Thus $S(\overline{\gamma})$ should be a mean-variance equilibrium if and only if the implied aggregate risk tolerance $\Psi(\overline{\gamma})$ matches the original $\overline{\gamma}$ that generated $S(\bar \gamma)$ as a quadratic equilibrium via Theorem \ref{thm:Qequilibria}. We now prove this formally.

\begin{theorem}
\label{thm:fixedpointeq}
Suppose that the primitives $S_0^{(1)}$, $M^{(1)}$, $D$, $\eta^k$ and $\Xi^{k,n}$, $k \in \{1,...,K\}$, satisfy Assumptions~\ref{ass:gammalowerbd} and~\ref{ass:orthogonal}, and let $\overline{\gamma} \in [\overline{\gamma}_{\min},\overline{\gamma}_{\max}]$. Then, $(1,S(\overline{\gamma}))$ is a mean-variance equilibrium if and only if $\overline{\gamma} = \Psi(\overline{\gamma})$.
\end{theorem}

\begin{proof}
For $\overline{\gamma} \in [\overline{\gamma}_{\min},\overline{\gamma}_{\max}]$, define the constants
\[
(\ell,(c_k,\varepsilon_k)_{k=1}^{K})(S(\overline{\gamma})) = \psi_2 \circ \psi_1 (\overline{\gamma}),
\]
associated with the (ex)MVH problems~\eqref{eq:MVH},~\eqref{eq:MVHex}, with respect to $S(\overline{\gamma})$ which is given by Lemma~\ref{lem:Qequilibrium2}. By Lemma~\ref{thm:Qequilibria}, $S(\overline{\gamma})$ satisfies Assumption~\ref{ass:NFL}, and so we may apply our results from Section~\ref{sec:indopt} with respect to $S(\overline{\gamma})$.

First, suppose that $(1,S(\overline{\gamma}))$ is a mean-variance equilibrium market. Then, by condition (1) of Definition~\ref{def:equilibrium}, for each agent $k \in \{1,...,K\}$ there exists a unique solution $\hat{\vartheta}^k \in \overline{\Theta}(S(\overline{\gamma}))$ to~\eqref{eq:MVoptprob}. By Theorem~\ref{thm:MVsol}, for each $k \in \{1,...,K\}$, there exists a continuous map $\hat{\gamma}_{U_k^{\MV}} : (0,1) \times \R_+ \times \R_+ \to \R_+$ such that $\hat{\vartheta}^k - \eta^k$ is the unique solution to the MVH problem~\eqref{eq:MVH} with payoff $H^k(\hat{\gamma}_{U_k^{\MV}}(\ell,c_k,\varepsilon_k)(S(\overline{\gamma}))) = \hat{\gamma}_{U_k^{\MV}}(\ell,c_k,\varepsilon_k)(S(\overline{\gamma})) - \Xi^k$. Now, on one hand, by Lemma~\ref{lem:eqconnection}, $(1,S(\overline{\gamma}))$ is a quadratic equilibrium with aggregate risk tolerance $\overline{\gamma}$. On the other hand, by Lemma~\ref{thm:Qequilibria}, $(1,S(\tilde{\gamma}))$ is the unique quadratic equilibrium with risk tolerances $\hat{\gamma}_{U_k^{\MV}}(\ell,c_k,\varepsilon_k)(S(\overline{\gamma}))$, $k = 1,...,K$, where
\[
\tilde{\gamma} := \sum_{k=1}^K \hat{\gamma}_{U_k^{\MV}}(\ell,c_k,\varepsilon_k)(S(\overline{\gamma})) = \psi_3(\ell,(c_k,\varepsilon_k)_{k=1}^{K})(S(\overline{\gamma}))  = \Psi(\overline{\gamma}).
\]
Thus, we must have $S(\overline{\gamma}) = S(\Psi(\overline{\gamma}))$. By Lemma~\ref{lmm:distinct:proc:bar:gamma}, this gives indeed that $\Psi(\overline{\gamma}) = \overline{\gamma}$.

Conversely, suppose that $\Psi(\overline{\gamma}) = \overline{\gamma}$ for some $\overline{\gamma} \in [\overline{\gamma}_{\min},\overline{\gamma}_{\max}]$.  By Theorem~\ref{thm:MVsol}, for each $k \in \{1,...,K\}$, there exists a unique optimal strategy $\hat{\vartheta}^k \in \overline{\Theta}(S(\overline{\gamma}))$ to~\eqref{eq:MVoptprob}, and a continuous map $\hat{\gamma}_{U_k^{\MV}} : (0,1) \times \R_+ \times \R_+ \to \R_+$ such that $\hat{\vartheta}^k - \eta^k$ is the unique solution to the MVH problem~\eqref{eq:MVH} with payoff $H^k(\hat{\gamma}_{U_k^{\MV}}(\ell,c_k,\varepsilon_k)(S(\overline{\gamma}))) = \hat{\gamma}_{U_k^{\MV}}(\ell,c_k,\varepsilon_k)(S(\overline{\gamma})) - \Xi^k$.

On the other hand, we have by assumption that
\[
\sum_{k=1}^K \hat{\gamma}_{U_k^{\MV}}(\ell,c_k,\varepsilon_k)(S(\overline{\gamma})) = \Psi(\overline{\gamma}) = \overline{\gamma},
\]
so that by Lemma~\ref{thm:Qequilibria}, $(1,S(\overline{\gamma}))$ is the unique quadratic equilibrium market with risk tolerances $\hat{\gamma}_{U_k^{\MV}}(\ell,c_k,\varepsilon_k)(S(\overline{\gamma}))$, $k = 1,...,K$. Therefore, $(1,S(\overline{\gamma}))$ is also a mean-variance equilibrium market, when noting that conditions (2) and (3) of Definition~\ref{def:equilibrium} are inherited from $(1,S(\overline{\gamma}))$ being a quadratic equilibrium to start with.
\end{proof}

Our strategy for showing the existence of a fixed point of $\Psi$ is to apply the intermediate value theorem, for which we need two key steps. The first is to find suitable bounds on the range of $\Psi$, which was already done in Corollary~\ref{cor:convexcomb}. The second step is now to argue the continuity of the composition $\Psi := \psi_3 \circ \psi_2 \circ \psi_1$ by considering each map individually. The continuity of $\psi_3$ follows from Proposition~\ref{prop:haty}, but the continuity of $\psi_1$ and $\psi_2$ still remains to be shown, and is truly nontrivial for $\psi_2$.

In order to obtain a continuity result for $ \psi_1$, we endow $\cS$ with the metric $d$ defined by
\begin{equation}
\label{eq:metric}
d((\xi^1,\lambda^1),(\xi^2,\lambda^2)) = \E\Big[\frac{\int_0^T(1 \wedge \|\xi_u^1 - \xi_u^2\|) \,\dd I_u^N}{I_T^N + 1}\Big] + \E[\langle \lambda^1 \bullet (\xi^1 \bullet N) - \lambda^2 \bullet (\xi^2 \bullet N)\rangle_T]^{\frac{1}{2}},
\end{equation} 
where $\|\cdot\|$ denotes the Frobenius norm on $\R^{d_3 \times d}$. Put simply, convergence under this metric is induced by the map $(\xi, \lambda) \mapsto (\xi, \lambda \bullet (\xi \bullet N))$, where we have local convergence in $\Pb \ptimes I^N$-measure of the first component $\xi$, and convergence in $\cM^2$ of the second component $\lambda \bullet (\xi \bullet N)$.

\begin{lemma}
\label{lem:xilambdacts}
Let $\overline{\gamma}_\infty \in [\overline{\gamma}_{\min},\overline{\gamma}_{\max}]$ and a sequence $(\overline{\gamma}_n)_{n \in \N} \subset [\overline{\gamma}_{\min},\overline{\gamma}_{\max}]$ such that $\overline{\gamma}_n \to \overline{\gamma}_\infty$ as $n \to \infty$. Then, we have
\begin{enumerate}[(i)]
\item $\E\Big[\frac{\int_0^T(1 \wedge \|\xi_u(\overline{\gamma}_n) - \xi_u(\overline{\gamma}_\infty)\| \,\dd I_u^N)}{I_T^N + 1}\Big] 
\to 0$ as $n \to \infty$, 
\item $\E[\langle \lambda(\overline{\gamma}_n) \bullet (\xi(\overline{\gamma}_n) \bullet N) - \lambda(\overline{\gamma}_\infty) \bullet (\xi(\overline{\gamma}_\infty) \bullet N)\rangle_T] \to 0$ as $n \to \infty$.
\end{enumerate}
\end{lemma}

Consequently, Lemma~\ref{lem:xilambdacts} gives us continuity of $\psi_1$ under the metric $d$ in~\eqref{eq:metric}.

\begin{corollary}
\label{cor:psi1cts}
Let Assumption~\ref{ass:orthogonal} hold. Then, the range of $\psi_1$ is contained in $\cS(C^*)$ for some constant $C^* \geq 1$, and $\psi_1: [\overline{\gamma}_{\min},\overline{\gamma}_{\max}] \to \cS(C^*)$ is continuous.
\end{corollary}

Thanks to Corollary~\ref{cor:psi1cts}, to prove that $\Psi$ is continuous,  it suffices to show continuity of $\psi_2$ on the domain $\cS(C) \subseteq \cS$ for fixed $C \geq 1$. To do so, we look into the dependence of the generic MVH and exMVH problems~\eqref{eq:MVH} and~\eqref{eq:MVHex}, for a given payoff $H$ that is bounded above and below away from $0$, on $(\xi,\lambda)$ (via the asset prices $S^{(\xi,\lambda)}$). We can then apply these findings on~\eqref{eq:exconst}-\eqref{eq:varepsilon2} with $H = \Xi^k$ for $k \in \{1,...,K\}$.

To that end, let $(\xi^n,\lambda^n) \xrightarrow{d} (\xi^\infty,\lambda^\infty)$ in $\cS(C)$ as $n \to \infty$, and write $S^n = S^{(\xi^n,\lambda^n)}$. Recalling~\eqref{eq:ell}-\eqref{eq:varepsilon2}, we also write $\ell^n = \ell(S^n)$, $c^n(H) = c(H, S^n)$ and $\varepsilon^n(H) = \varepsilon(H,S^n)$. We want to show that $(\ell^n,c^n(H),\varepsilon^n(H)) \to (\ell^\infty,c^\infty(H),\varepsilon^\infty(H))$ as $n \to \infty$. To achieve this, the main idea is to relate $(\ell^n)_{n\in\N}$, $(c^n(H))_{n\in\N}$ and $(\varepsilon^n(H))_{n\in\N}$ to the initial values of the solutions of corresponding families of BSDEs, with the latter all satisfying a BSDE stability result. This relationship stems from analysis of dynamic mean-variance hedging, for which we use Jeanblanc et al.~\cite[Section 3]{jeanblanc:al:12} as a reference; see also the discussion starting below~\cite[Corollary IV.4.27]{martins:23} for more details. Indeed, we have the relationships
\begin{equation}
\label{eq:bsderelation}
\ell^n = Y_0^{(2),n}, \quad c^n(H) = \frac{Y_0^{(1),n}(H)}{Y_0^{(2),n}} \quad \text{and} \quad \varepsilon^n(H) = Y_0^{(0),n}(H) - \frac{(Y_0^{(1),n}(H))^2}{Y_0^{(2),n}},
\end{equation}
where $Y^{(2),n}$, $Y^{(1),n}(H)$ and $Y^{(0),n}(H)$ satisfy the BSDEs
\begin{align*}
\dd Y_t^{(2),n} &= \frac{1}{Y_t^{(2),n}}(Y_t^{(2),n}\lambda_t^n + \nu_t^{(2),n})^\top\dd\langle N^n \rangle_t(Y_t^{(2),n}\lambda_t^n + \nu_t^{(2),n}) \\
&\quad + \nu_t^{(2),n} \,\dd N_t^n + \dd N_t^{\perp,(2),n}, \\
\dd Y_t^{(1),n}(H) &= \frac{1}{Y_t^{(2),n}}(Y_t^{(2),n}\lambda_t^n + \nu_t^{(2),n})^\top\dd\langle N^n \rangle_t(Y_t^{(1),n}(H)\lambda_t^n + \nu_t^{(1),n}) \\
&\quad + \nu_t^{(1),n} \,\dd N_t^n + \dd N_t^{\perp,(1),n}, \\
\dd Y_t^{(0),n}(H) &= \frac{1}{Y_t^{(2),n}}(Y_t^{(1),n}(H)\lambda_t^n + \nu_t^{(1),n})^\top\dd\langle N^n \rangle_t(Y_t^{(1),n}(H)\lambda_t^n + \nu_t^{(1),n}) \\
&\quad + \nu_t^{(0),n} \,\dd N_t^n + \dd N_t^{\perp,(0),n},
\end{align*}
with terminal conditions $Y_T^{(2),n} = 1$, $Y_T^{(1),n}(H) = H$ and $Y_T^{(0),n}(H) = H^2$, where $N^n = \xi^n \bullet N$, $\nu^{(i),n} \in L_{\loc}^2(N^n)$ and $N^{\perp,(i),n}$ is a continuous local martingale strongly orthogonal to $N^n$, for $i \in \{0,1,2\}$. The stability of BSDEs of this form is the main mathematical contribution of the complementary paper Constantinou et al.~\cite{constantinou:al:25b}. That result allows us to tackle the most technically challenging step in this paper, which is to prove the continuity of $\psi_2$.

For the next step, we require an additional invertibility assumption on the local martingale $N$ defined after \eqref{eq:FVeq}, where we recall also
$\kappa^N$ defined as in~\eqref{eq:RN}. This will be needed in order to obtain the convergence of the drivers of these BSDEs.

\begin{assumption}
\label{ass:cNinv}
The event $\{(\omega,t): \kappa_t^{N}(\omega) \,\, \text{is singular}\}$ is $\Pb \ptimes I^N$-null.
\end{assumption}

\begin{theorem}
\label{thm:stability2}
Let Assumption~\ref{ass:cNinv} hold. Let $C \geq 1$ and $H$ be a payoff that is bounded above and below away from $0$. For any sequence $(\xi^n,\lambda^n)$ in $\cS(C)$ such that $(\xi^n,\lambda^n) \xrightarrow{d} (\xi^\infty,\lambda^\infty) \in \cS(C)$, it holds for each $i \in \{0,1\}$ that
\begin{equation}
\label{eq:stability}
| Y_0^{(2),n} - Y_0^{(2),\infty} | \to 0 \quad \text{and} \quad | Y_0^{(i),n}(H) - Y_0^{(i),\infty}(H) | \to 0, \qquad \text{as} \,\, n \to \infty.
\end{equation}
\end{theorem}

By applying~\eqref{eq:bsderelation} and the stability result in Theorem~\ref{thm:stability2} to $H = \Xi^k$ for each agent $k$, we obtain the continuity result for $\psi_2$.

\begin{corollary}
\label{cor:psi2cts}
Let Assumptions~\ref{ass:orthogonal} and~\ref{ass:cNinv} hold. Then, the map $\psi_2: \cS \to (0,1] \times (\R \times \R_+)^K$ is continuous on $\cS(C) \subseteq \cS$ for each $C \geq 1$.
\end{corollary}

\begin{proof}
Let $C \geq 1$. Suppose that $(\xi^n,\lambda^n) \xrightarrow{d} (\xi^\infty,\lambda^\infty)$ in $\cS(C)$ as $n \to \infty$, and recall for each $n \in \N \cup \{\infty\}$ the notation $S^n := S^{(\xi^n,\lambda^n)}$ in~\eqref{eq:SC}, where we arbitrarily set $S_0^n = 0$. For each $n \in \N \cup \{\infty\}$ and $i \in \{0,1\}$, denote by $Y^{(i),n,k} := Y^{(i),n}(\Xi^k)$. Then, by Theorem~\ref{thm:stability2}, we have $Y_0^{(2),n} \to Y_0^{(2),\infty} > 0$ and $Y_0^{(i),n,k} \to Y_0^{(i),\infty,k}$ as $n \to \infty$, for $i \in \{0,1\}$; note that $\cF_0$ is trivial by assumption. Consequently, for each $k \in \{1,...,K\}$, it follows from the relationships~\eqref{eq:bsderelation} that
\[
(\ell^n,c_k^n,\varepsilon_k^n) = (\ell(S^n),c(\Xi^k,S^n),\varepsilon(\Xi^k,S^n)) \to (\ell(S^\infty),c(\Xi^k,S^\infty),\varepsilon(\Xi^k,S^\infty)) = (\ell^\infty,c_k^\infty,\varepsilon_k^\infty),
\]
as $n \to \infty$, which shows continuity of $\psi_2$ on $\cS(C)$.
\end{proof}

As outlined below Theorem~\ref{thm:fixedpointeq}, continuity of $\psi_3$ and suitable bounds for $\Psi$ are already granted by Theorem~\ref{thm:MVsol} and Corollary~\ref{cor:convexcomb}, respectively, and so we can now combine the results on the continuity of each map $\psi_1$ and $\psi_2$ to argue the existence of a fixed point.

\begin{theorem}
\label{thm:fixedpoint}
Let Assumptions~\ref{ass:gammalowerbd},~\ref{ass:orthogonal} and~\ref{ass:cNinv} hold. Then, $\Psi$ admits a fixed point $\overline{\gamma}^* \in [\overline{\gamma}_{\min},\overline{\gamma}_{\max}]$ and $(1,S(\overline{\gamma}^*))$ is a mean-variance equilibrium market.
\end{theorem}

\begin{proof}
First, note that $\Psi = \psi_3 \circ \psi_2 \circ \psi_1$, $\Psi: [\overline{\gamma}_{\min},\overline{\gamma}_{\max}] \to \R$, is a continuous function, since each map $\psi_1$, $\psi_2$ and $\psi_3$, is continuous due to Corollaries~\ref{cor:psi1cts},~\ref{cor:psi2cts}, and Theorem~\ref{thm:MVsol}, respectively. Now, set $\widetilde{\Psi}(\overline{\gamma}) := \Psi(\overline{\gamma})- \overline{\gamma}$, $\overline{\gamma} \in [\overline{\gamma}_{\min},\overline{\gamma}_{\max}]$. By Corollary~\ref{cor:convexcomb}, we get the bounds
\[
\ell(S(\overline{\gamma}_{\min}))\widetilde{\Psi}(\overline{\gamma}_{\min}) \geq 0 \quad \text{and} \quad \ell(S(\overline{\gamma}_{\max}))\widetilde{\Psi}(\overline{\gamma}_{\max}) \leq 0.
\]
Since $\ell(S(\overline{\gamma})) > 0$, we have $\widetilde{\Psi}(\overline{\gamma}_{\min}) \geq 0$ and $\widetilde{\Psi}(\overline{\gamma}_{\max}) \leq 0$. Hence, by the intermediate value theorem, there exists $\overline{\gamma}^* \in [\overline{\gamma}_{\min},\overline{\gamma}_{\max}]$ so that $\Psi(\overline{\gamma}^*) = \overline{\gamma}^*$, i.e., $\overline{\gamma}^*$ is a fixed point for $\Psi$. Therefore, by Theorem~\ref{thm:fixedpointeq}, $(1,S(\overline{\gamma}^*))$ is a mean-variance equilibrium market.
\end{proof}

The following corollary establishes a semi-explicit derivation for the implied equilibrium quadratic risk tolerance for each agent $k$, and in the case of identical agents, a fully explicit equilibrium aggregate risk tolerance.

\begin{corollary}
\label{cor:MVeqsymmetric}
Suppose that $\overline{\gamma}^* \in [\overline{\gamma}_{\min},\overline{\gamma}_{\max}]$ is a fixed point of $\Psi$. For each $k \in \{1,\ldots,K\}$, set $\gamma_k^* := \hat{\gamma}_{U_k^{\MV}}((\ell,c_k,\varepsilon_k)(S(\overline{\gamma}^*)))$, denote by $V_T^{k,*}$ the equilibrium terminal wealth of agent $k$, and set $\mu_k^* := \E[V_T^{k,*}]$ and $\sigma_k^* := \sqrt{\Var(V_T^{k,*})}$. Then, we have
\[
\gamma_{k}^* = \mu_k^* - \sigma_k^* \frac{\partial_\mu U_k^{\MV}}{\partial_\sigma U_k^{\MV}}(\mu_k^*,\sigma_k^*).
\]
Consequently, if the agents are identical, i.e., $U_k^{\MV} = U^{\MV}$ and $\Xi^k = \overline{\Xi} / K$ for all $k \in \{1,\ldots, K\}$, then it holds that
\[
\overline{\gamma}^* = \E[\overline{\Xi}] - \sqrt{\Var(\overline{\Xi})}\frac{\partial_\mu U^{\MV}}{\partial_\sigma U^{\MV}}(\E[\overline{\Xi}] / K,\sqrt{\Var(\overline{\Xi})} / K).
\]
\end{corollary}

We conclude this paper with a Markovian market where Assumptions~\ref{ass:gammalowerbd},~\ref{ass:orthogonal} and~\ref{ass:cNinv} can be traced via the market primitives, and also obtain an (explicit) mean-variance equilibrium.

\begin{example}
Suppose that the market filtration is generated and completed by a three-dimensional Brownian motion $W = (W^{1},W^{2},W^{3}) = ((W_t^{1},W_t^{2},W_t^{3}))_{t\in[0,T]}$, and that there is only one productive asset, i.e., $d_1 = 0$ and $d_2 = 1$. Moreover, suppose that the terminal dividend $D$ and the aggregate nontraded endowment $\overline{\Xi}^{\nt}$ are of the form
\[
D = d + \sum_{i=1}^3 a_i \sin(W_T^i) \quad \text{and} \quad \overline{\Xi}^{\nt} = e + \sum_{i=1}^3 b_i \sin(W_T^i)
\]
for some constants $a_i,b_i >0$, $i \in \{1,2,3\}$, such that the vectors $a = (a_1,a_2,a_3)^\top$ and $b = (b_1,b_2,b_3)^\top$ are not proportional, and $d > \sum_{i=1}^{3} a_i$, $e > \sum_{i=1}^{3} b_i$. Since the productive asset is in unit net supply, the aggregate endowment $\overline{\Xi}$ is given by
\[
\overline{\Xi} = D + \overline{\Xi}^{\nt} = x + \sum_{i=1}^{3} c_i \sin(W_T^i), \qquad x:= d + e, \quad c_i := a_i + b_i, \,\, i = 1,2,3,
\]
which is bounded and strictly positive. Moreover, suppose that there are $K \geq 2$ identical agents (with symmetric initial and terminal endowments) with common mean-variance utility
\[
U(\mu,\sigma) = \mu - \alpha \sqrt{1 + \mu^2} - \frac{\sigma^2}{2\lambda}, \qquad 0 < \alpha < 1, \quad \lambda > 0,
\]
which satisfies the required assumptions of mean-variance preferences with the bounds
\[
\frac{\sigma}{\lambda(1+\alpha)} \leq - \frac{\partial_\sigma U(\mu,\sigma)}{\partial_\mu U(\mu,\sigma)} \leq \frac{\sigma}{\lambda(1-\alpha)}.
\]
We also require that $\overline{\gamma}_{\min} > \esssup{\overline{\Xi}}$ in light of Assumption~\ref{ass:gammalowerbd}, which is equivalent to $K\lambda(1-\alpha) > \sum_{i=1}^3c_i$.

Consider the ``shifted'' aggregate endowment $\hat{H} := H(\overline{\gamma}_{\min})= \overline{\gamma}_{\min} - \overline{\Xi}$. We proceed to argue that Assumptions~\ref{ass:orthogonal} and~\ref{ass:cNinv} hold based on the coefficients of $D$ and $\hat{H}$, where the base martingale $N = (Z^{D},\hat{Z}^{HD},\hat{Z}^{H})$ is three-dimensional. Fix $t \in [0,T]$. First, by standard calculations, we get explicit expressions for the generated martingales $Z^D$, $\hat{Z}^H$ and $\hat{Z}^{HD}$ to be
\begin{align*}
Z_t^D &= \E[D \,|\, \cF_t] =  d + e^{-(T-t)/2}\sum_{i=1}^{3} a_i\sin(W_t^i), \\
\hat{Z}^H_t &= \E[\hat{H} \,|\, \cF_t] = \overline{\gamma}_{\min} - x - e^{-(T-t)/2}\sum_{i=1}^{3} c_i\sin(W_t^i), \\
\hat{Z}^{HD}_t &= \E[\hat{H}D \,|\, \cF_t] = \hat{Z}_t^HZ_t^D - \frac{1 - e^{-(T-t)}}{2}\sum_{i=1}^3a_ic_i(1 + e^{-(T-t)}\cos(2W_t^i)).
\end{align*}
This gives in particular the covariation
\[
\dd \langle \hat{Z}^H, Z^D \rangle_t = -e^{-(T-t)}\sum_{i=1}^{3}a_ic_i\cos^2(W_t^i) \,\dd t,
\]
whose density is strictly negative $\Pb\otimes \dd t\text{-a.e.}$, so that $\hat{Z}^H$ and $Z^D$ are not strongly orthogonal which verifies Assumption~\ref{ass:orthogonal}. Now, if $\dd N_t = \sigma_t^N \,\dd W_t$, the $i$-th column of $\sigma_t^N$ is
\begin{equation}
(\sigma_t^N)_{\cdot i}
=e^{-(T-t)/2}\cos(W_t^i)
\begin{pmatrix}
a_i\\[1mm]
a_i\hat{Z}_t^H-c_iZ_t^D
 +2a_ic_ie^{-(T-t)/2}(1-e^{-(T-t)})\sin(W_t^i)\\[1mm]
-c_i
\end{pmatrix}.
\end{equation}
Consequently, we get
\begin{equation}
\det \sigma_t^N
=-2e^{-2(T-t)}(1-e^{-(T-t)})\prod_{i=1}^3\cos(W_t^i)\,\Delta(W_t),
\end{equation}
where
\[
\Delta(y):=
\det\begin{pmatrix}
a_1&a_2&a_3\\
a_1c_1\sin y_1&a_2c_2\sin y_2&a_3c_3\sin y_3\\
c_1&c_2&c_3
\end{pmatrix}.
\]
Since $a$ and $b$ are not proportional, neither are $a$ and $c=a+b$. Hence, at least one minor $a_jc_k-a_kc_j$ is nonzero, so $\Delta$ is a nonzero real-analytic function. Its zero set has Lebesgue measure zero. As $W_t$ has a density for every $t>0$, $\det\sigma_t^N\neq0$ $\Pb\otimes dt$-a.e. on $[0,T]$. Therefore, $\kappa^N = \frac{\dd \langle N \rangle}{\dd I^N} = \frac{\sigma_t^N(\sigma_t^N)^\top}{\trr(\sigma_t^N(\sigma_t^N)^\top)}$ is invertible $\Pb\otimes dt$-a.e. which verifies Assumption~\ref{ass:cNinv}.

It remains to determine the equilibrium parameter. Since the agents are identical, symmetric allocations are represented by $\bar U(\mu,\sigma):=K\,U(\mu/K,\sigma/K)$, where we have
\[
\bar U_\mu(\mu,\sigma) = 1-\alpha\frac{\mu/K}{\sqrt{1+(\mu/K)^2}},
\qquad \bar U_\sigma(\mu,\sigma)=-\frac{\sigma}{K\lambda}.
\]
Thus, by Corollary~\ref{cor:MVeqsymmetric}, the aggregate equilibrium risk tolerance is
\[
\bar\gamma^* = x+K\lambda\left(1-\alpha\frac{x/K}{\sqrt{1+(x/K)^2}}\right) > x + K\lambda(1-\alpha) > x + \sum_{i=1}^3 c_i.
\]
To compute the equilibrium price, Lemma~\ref{thm:Qequilibria} gives
\[
\begin{aligned}
S_t(\overline{\gamma}^*)
&=Z_t^D-\frac{\Cov[\bar\Xi,D\mid\mathcal F_t]}{\overline{\gamma}^* - Z_t^{\overline{\Xi}}} \\
&= d + e^{-(T-t)/2}\sum_{i=1}^{3} a_i\sin(W_t^i) \\
&\quad -\frac{\frac{1 - e^{-(T-t)}}{2}\sum_{i=1}^3a_ic_i(1 + e^{-(T-t)}\cos(2W_t^i))}{\overline{\gamma}^* - x - e^{-(T-t)/2}\sum_{i=1}^{3}c_i\sin(W_t^i)}.
\end{aligned}
\]
The denominator is uniformly bounded away from zero, so $S$ is bounded and the buy-and-hold strategy is admissible. Since $a_i,c_i>0$, $i \in \{1,2,3\}$, the covariance term is nonnegative, and thus $S_t\leq Z_t^D$. At time zero, we have
\[
S_0(\overline{\gamma}^*) = d - \frac{(1-e^{-2T})\sum_{i=1}^3a_ic_i}{2K\lambda\left(1-\alpha\frac{x/K}{\sqrt{1+(x/K)^2}}\right)}.
\]
The nonlinear mean-variance preferences yield an explicit fixed-point $\overline{\gamma}^*$ which is positively affected by an increase in the common risk parameter $\lambda$ and it is negatively affected by an increase in the nonlinear curvature parameter $\alpha$. This effect passes positively on the initial price $S_0$.
\end{example}

\section{Conclusions}
\label{sec:conclusions}

Our results show that the equilibrium analysis of continuous-time CAPM models can be extended beyond quadratic (and linear mean--variance) preferences without sacrificing a tractable characterisation of asset prices. The essential observation is that quadratic utility provides a common framework within which more general mean--variance optimal investment decisions can be represented. Rather than treating the corresponding quadratic risk tolerances as exogenous preference parameters, we determine them endogenously through the agents' optimisation problems and the market-clearing condition. This approach leads to an existence result for mean--variance equilibria in continuous-time incomplete markets and, at the same time, clarifies the role of aggregate risk tolerance in determining equilibrium prices.

Our analysis also highlights the economic significance of market incompleteness and the agents' behaviour towards risk. In particular, risks that cannot be fully hedged through trading may nevertheless affect the prices of traded assets through their contribution to the aggregate endowment. The examples illustrate how stochastic-volatility factors and unspanned risks enter equilibrium pricing, while the explicit calculations for nonlinear mean--variance preferences demonstrate how changes in the agents' evaluation of expected wealth and volatility translate into changes in equilibrium risk tolerance and asset prices. The closed-form characterisation obtained for symmetric economies further shows that, in certain settings, equilibrium risk tolerance can be determined directly from aggregate endowment characteristics and preferences.

From a methodological viewpoint, an important part of the analysis is to establish that the relevant mean--variance hedging quantities respond continuously to changes in the candidate equilibrium prices. This requires stability results for systems of quadratic BSDEs with BMO growth at the origin, developed in the companion paper~\cite{constantinou:al:25b}. While our equilibrium existence theorem relies on specific boundedness and nondegeneracy conditions, the underlying connection between quadratic utility and general mean--variance optimisation suggests a broader approach to equilibrium analysis. Investigating the extent to which these conditions can be weakened, as well as the uniqueness and dynamic consistency of the resulting equilibria, provides directions for further research. \\

\appendix
\normalsize
\section{Key results on mean--variance hedging}\label{app:mean:var:hedging:prelim}

\subsection{Static theory}

We give here some results on the MVH and exMVH problems \eqref{eq:MVH} and \eqref{eq:MVHex}. We omit for conciseness the (technical) proofs of Lemma \ref{lmm:mvh:prob:linearity}, Propositions \ref{prop:UOGP}--\ref{prop:UOVP} and \ref{prop:LOP} below; they can be found in Martins \cite[Section III.2.5]{martins:23}. 

\begin{lemma}
\label{lmm:mvh:prob:linearity}
Let $H_1, H_2 \in L^2$ and $\lambda \in \R$.
	
\begin{enumerate}[label=\emph{\textbf{\arabic*)}}]
\item Suppose  there exist solutions $\vt^{\rm MVH}(H_1), \vt^{\rm MVH}(H_2) \in \overline \Theta(S)$ to \eqref{eq:MVH} for $H_1$ and $H_2$, respectively. Then $\vt^{\rm MVH}(H_1) + \lambda \vt^{\rm MVH}(H_2)$ is a solution to \eqref{eq:MVH} for the payoff $H_1 + \lambda H_2$.
	
\item Suppose  there exist solutions $(c(H_1),\vt^{\mathrm{ex}}(H_1)), (c(H_2),\vt^{\mathrm{ex}}(H_2)) \in \R \times \overline \Theta(S)$ to \eqref{eq:MVHex} for $H_1$ and $H_2$, respectively. Then $(c(H_1)+\lambda c(H_2), \vt^{\mathrm{ex}}(H_1) + \lambda \vt^{\mathrm{ex}}(H_2))$ is a solution to \eqref{eq:MVHex} for the payoff $H_1 + \lambda H_2$.
\end{enumerate}
\end{lemma}

\begin{proof}
See \cite[Lemma III.1.6]{martins:23}.
\end{proof}

If $S$ admits an equivalent local martingale measure (ELMM) with square-integrable density, then $\mathcal{G}_T(S)$ and  $\R + \mathcal{G}_T(S)$ are closed and the solutions in $\overline \Theta(S)$ and $\R \times \overline \Theta(S)$ are unique; this follows by \v Cern\' y/Kallsen \cite[Lemma 2.11]{cerny:kallsen:07}.
However, without that extra assumption, both closedness of $\mathcal{G}_T(S)$ and $\R + \mathcal{G}_T(S)$ as well as uniqueness of the solutions in $\overline \Theta(S)$ and $\R \times \overline \Theta(S)$ (if they exist) do not hold in general. 
 
In order to deal with the uniqueness issue, it is useful to introduce the notions of \emph{uniqueness of gains processes} and \emph{uniqueness of value processes} associated with a price process $S$.

\begin{definition}\label{def:uni:gains:values}
Let $(1, S)$ be a local $L^2$-market. It is said to satisfy 
\begin{itemize}
\item \emph{uniqueness of gains processes} if for any two trading strategies $\vt^1, \vt^2 \in \overline \Theta(S)$, the equality $\vt^1 \sint S_T = \vt^2 \sint S_T$ $\Pas$ implies that $\vt^1 \eq{S} \vt^2$.
\item \emph{uniqueness of value processes} if for any two trading strategies $\vt^1, \vt^2 \in \overline \Theta(S)$ and initial values $c_1, c_2 \in \R$, the equality $c_1 + \vt^1 \sint S_T =c_2 + \vt^2 \sint S_T$ $\Pas$ implies that $c_1 = c_2$ and $\vt^1 \eq{S} \vt^2$.
\end{itemize}
\end{definition}

We have the following two equivalent characterisations of uniqueness of gains and value processes. They follow immediately from the linear structure of the (extended) mean--variance hedging problems given in Lemma \ref{lmm:mvh:prob:linearity}, as well as the fact that for  $H =0$, the problems of MVH \eqref{eq:MVH} and  exMVH \eqref{eq:MVHex}  admit as solutions $\vt = 0$ and $(c, \vt) = (0, 0)$, respectively.
\begin{proposition}
\label{prop:UOGP}
Let $(1, S)$ be a local $L^2$-market. The following are equivalent:
\begin{enumerate}
\item $(1, S)$ satisfies uniqueness of gains processes.
\item For some $H \in L^2$, the MVH problem \eqref{eq:MVH}
admits a unique solution.
\item For each $H \in L^2$ for which the MVH problem \eqref{eq:MVH}
admits a solution, the solution is unique.
\end{enumerate}
\end{proposition}

\begin{proposition}
\label{prop:UOVP}
Let $(1, S)$ be a local $L^2$-market. The following are equivalent:
\begin{enumerate}
\item $(1, S)$ satisfies uniqueness of value processes.
\item For some $H \in L^2$, the exMVH problem \eqref{eq:MVHex}
admits a unique solution.
\item For each $H \in L^2$ for which the exMVH problem \eqref{eq:MVHex}
admits a solution, the solution is unique.
\end{enumerate}
\end{proposition}

\begin{proof}[Proof of Propositions \ref{prop:UOGP} and \ref{prop:UOVP}]
	See \cite[Propositions III.1.8 and III.1.9]{martins:23}.
\end{proof}

The following two results show that uniqueness of value processes implies uniqueness of gains processes and link the MVH problem \eqref{eq:MVH} and the extended exMVH problem \eqref{eq:MVHex}.

\begin{proposition}
\label{prop:UOVP link UOGP}
Let $(1, S)$ be a local $L^2$-market. The following statements hold:
\begin{enumerate}[label=\emph{\textbf{\arabic*)}}]
\item  If $(1,S)$ satisfies uniqueness of value processes, then $(1, S)$ also satisfies uniqueness of gains processes.
	
\item Suppose that $(1,S)$ satisfies uniqueness of gains processes and that the MVH problem \eqref{eq:MVH} for $H=1$ admits a solution $\vt^{\rm MVH}(1)$. Then $(1,S)$ satisfies uniqueness of value processes if and only if  $\E[(\vt^{\rm MVH}(1) \sint S_T-1)^2] > 0$.
\end{enumerate}
\end{proposition}

\begin{proof}
\textbf{1)} Since $\vt \equiv 0$ is a solution to \eqref{eq:MVH} for $H=0$ with hedging error $0$, any solution to \eqref{eq:MVH} for $0$ is also a solution to \eqref{eq:MVHex} for $0$ with $c = 0$. Thus the first assertion follows from Propositions \ref{prop:UOGP} and \ref{prop:UOVP}, using the fact that if the MVH problem \eqref{eq:MVH} for $0$ does not have a unique solution, then a fortiori the exMVH problem \eqref{eq:MVHex} for $0$ cannot have a unique solution. 
	
\textbf{2)} We assume first that $(1,S)$ satisfies uniqueness of value processes. Supposing for a contradiction that $\E[(\vt^{\rm MVH}(1) \sint S_T-1)^2] = 0$, then the exMVH problem \eqref{eq:MVHex} for $0$ has two distinct solutions $(-1, \vt^{\rm MVH}(1))$ and $(0, 0)$, which contradicts Proposition \ref{prop:UOVP}.
	
To prove the converse statement, suppose that $\E[(\vt^{\rm MVH}(1) \sint S_T-1)^2] > 0$. We claim that $(0,0)$ is the unique solution to the exMVH problem \eqref{eq:MVHex} for $H=0$. To show the claim, let $(c,\vt)$ be a solution to  the exMVH problem \eqref{eq:MVHex} for $H \equiv 0$. Then by comparing the MVH \eqref{eq:MVH} and exMVH \eqref{eq:MVHex} problems, we note that $\vt$ is also a solution to the MVH problem for $H \equiv -c$. The uniqueness of gains processes and Proposition \ref{prop:UOGP} yield that $\vt^{\rm MVH}(1)$ and $\vt$ are the unique solutions to the MVH problem \eqref{eq:MVH} for $H \equiv 1$ and $H \equiv -c$, respectively. In the case $c = 0$, we obtain in particular that $\vt = 0$ is the unique solution to the MVH problem for $H \equiv 0$, and hence the claim holds in this case since $(c,\vt) = (0,0)$. Suppose now for a contradiction that $c \neq 0$. By the linearity of MVH (see Lemma \ref{lmm:mvh:prob:linearity}), we have $\vt = -c \vt^{\rm MVH}(1)$. Since both $(0,0)$ and $(c,\vt)$ are solutions to the exMVH problem for $H \equiv 0$, we deduce that
$$c + \vt \sint S_T = c(1 - \vt^{\rm MVH}(1) \sint S_T) \equiv 0.$$
However, since $c \neq 0$, this contradicts the assumption so that this case cannot hold. Therefore, the only solution to the exMVH problem for $H\equiv 0$ is $(0,0)$, which shows  the claim. The conclusion that $(1,S)$ satisfies uniqueness of value processes then follows by Proposition \ref{prop:UOVP}.
\end{proof}

\begin{corollary}
\label{cor:UOVP link UOGP}
Let $(1, S)$ be a local $L^2$-market satisfying uniqueness of value processes such that the MVH problem \eqref{eq:MVH} for $1$ has a unique solution $\vt^{\rm MVH}(1)$. Then for $H \in L^2$, the MVH problem \eqref{eq:MVH} for $H$ has a unique solution $\vt^{\rm MVH}(H)$ if and only if the exMVH problem \eqref{eq:MVHex} for $H$ has a unique solution $(c(H), \vt^{\mathrm{ex}}(H))$, in which case we have
\begin{equation}
\label{eq:cor:UOVP link UOGP}
c(H) = \frac{\E[H (1 -\vt^{\rm MVH}(1) \sint S_T )]}{\E[(1 - \vt^{\rm MVH}(1) \sint S_T)^2]}
\end{equation}
and $\vt^{\mathrm{ex}}(H) =_S  \vt^{\rm MVH}(H) - c(H) \vt^{\rm MVH}(1).$
\end{corollary}

\begin{proof}
We first show the ``if'' statement. Since $ \vt^{\mathrm{ex}}(H)$ solves the MVH problem with payoff $H-c(H)$, it follows by the linearity of MVH (see Lemma \ref{lmm:mvh:prob:linearity}) that $\vt := \vt^{\mathrm{ex}}(H) + c(H) \vt^{\rm MVH}(1)$ is a solution to the MVH problem \eqref{eq:MVH} for $H$. The uniqueness of the solution is ensured by the uniqueness of value processes and Proposition \ref{prop:UOGP}.
	
Next, we consider the ``only if'' statement. Suppose for a contradiction that the pair $(c(H),\vt^{\mathrm{ex}}(H) )$ given in the statement of the corollary is not a solution to the exMVH problem for $H$. Thus there exists a competitor $(c,\vt)$ that achieves a lower mean squared error. This competitor may be improved by letting $\vt$ be the solution to the MVH problem \eqref{eq:MVH} for $H-c$, so that we may assume $\vt = \vartheta(H) - c\vartheta(1)$. Moreover, as $\E[\vartheta^{\rm MVH}(H) \sint S_T(1 - \vartheta^{\rm MVH}(1) \sint S_T)] = 0$ by the first-order condition of MVH, the mean squared error attained by $(c,\vt )$ is given by
\begin{align}
&\E\Big[\Big(c + \big(\vartheta^{\rm MVH}(H) - c\vartheta^{\rm MVH}(1)\big)\sint S_T - H \Big)^2\Big] \notag \\
&= c^2 \E\big[\big(1-  \vartheta^{\rm MVH}(1) \sint S_T\big)^2\big]   - 2 c \E\big[ H \big(1 - \vartheta^{\rm MVH}(1) \sint S_T \big)\big] + \E\big[\big(H - \vartheta^{\rm MVH}(H) \sint S_T \big)^2\big].\label{eq:pf:cor:UOVP link UOGP}
\end{align}
As a quadratic function of $c$, the right-hand side of \eqref{eq:pf:cor:UOVP link UOGP} has the unique minimiser $c(H)$ given by \eqref{eq:cor:UOVP link UOGP}, so that the error may be reduced further by setting $c = c(H)$. But then we obtain $(c, \vt) = (c(H), \vt^{\mathrm{ex}}(H))$, which leads to a contradiction. Therefore, the pair $(c(H), \vt^{\mathrm{ex}}(H))$ given in the statement of the corollary is indeed a solution to the exMVH problem for $H$, and its uniqueness follows by Proposition \ref{prop:UOVP}.
\end{proof}

The following result gives simple sufficient conditions for uniqueness of gains and value processes and for the existence of solutions to the MVH and exMVH problems \eqref{eq:MVH} and \eqref{eq:MVHex} in terms of a signed local martingale measure for $S$. The assumption that such a signed measure exists is not necessary for the existence of solutions to the MVH and exMVH problems, but it is a weaker assumption than the existence of an equivalent local martingale measure for $S$. By \v Cern\'y/Czichowsky \cite[Theorem 3.2]{cerny:czichowsky:22}, the conditions in part 2) imply the economic assumption of the so-called law of one price; see \cite[Definition 2.3]{cerny:czichowsky:22}.
\begin{proposition}
\label{prop:LOP}
Let $(1, S)$ be a local $L^2$-market and $Z = (Z_t)_{0 \leq t \leq T}$ a square-integrable martingale such that $Z S^j$ is a local martingale for all $j \in \{1, \ldots, d\}$. Then $Z (\vt \sint S)$ is a $\Pb$-martingale for each $\vt \in \overline \Theta(S)$. Moreover:
\begin{enumerate}[label={\emph{\textbf{ \arabic*)}}}]
\item If $Z_t \neq 0$ $\Pas$ for each $t \in [0, T]$, then $(1, S)$ satisfies uniqueness of value processes.
\item If $Z_t \neq 0$ and $Z_{t-} \neq 0$ for all $t \in [0, T]$ $\Pas$, then the MVH problem \eqref{eq:MVH} and exMVH problem \eqref{eq:MVHex} have unique solutions for each $H \in L^2$.
\end{enumerate}
\end{proposition}

\begin{proof}
See \cite[Proposition III.1.12]{martins:23}.
\end{proof}

We close this appendix by linking the zero solution of an MVH problem to a local martingale-type condition for~$S$, which is used to prove some of the main results in Section \ref{sec:representative:agent}.

\begin{lemma}
\label{lem:MVH:zero}
Let $(1, S)$ be a local $L^2$-market and $H \in L^2$. Define the square-integrable martingale $Z =(Z_t)_{0 \leq t \leq T}$ by
$Z_t := \E[H\,|\,\cF_t]$. The following are equivalent:
\begin{enumerate}
\item  $0  \in \overline{\Theta}(S)$ solves the MVH problem \eqref{eq:MVH} for $H$.
\item $Z S^j$ is a local $\Pb$-martingale for all $j \in \{1, \ldots, d\}$.
\end{enumerate}
\end{lemma}

\begin{proof}
(a) $\Rightarrow$ (b): 	As $0\in \overline{\Theta}(S)$ is a solution to \eqref{eq:MVH}, we have for all $\vt \in \TB(S)$ and $\delta \in \R$ that
$$\E[(H-\delta \vt \sint S_T)^2] \geq \E[H^2].$$
By taking $\delta\searrow 0$ and $\delta \nearrow 0$, we deduce by examining the first-order term that $\E[(\vt \sint S_T)  H] = 0$ for any $\vt \in \overline{\Theta}(S)$, i.e., $H$ is orthogonal to $\mathcal G_T(S) \subseteq L^2$. 

Now fix $j \in \{1, \ldots, d\}$ and let $\sigma$ be a stopping time. 
Since $(1, S)$ is a local $L^2$-market,  there exists a localising sequence $(\tau_n)_{n \in \N}$ such that for each $n \in \N$, the stopped process $S^{j,\tau_n}$ is an $L^2$-semimartingale. Fix $n \in \N$. Then the strategy $e^{j, \sigma, n} := (0, \ldots, 0, \one_{\rrbracket 0, \sigma \wedge \tau_n \rrbracket}, 0, \ldots, 0)$, where the indicator process is at the $j$-th position, belongs to $\Theta_{\rm simple}(S) \subseteq \overline{\Theta}(S)$. The fact that $Z$ is a $\Pb$-martingale with $Z_T = H$ and $\E[(e^{j, \sigma, n}  \sint S_T)  H] = 0$ yield
\begin{align*}
\E[S^j_{\sigma \wedge \tau_n}Z_{\sigma \wedge \tau_n} - S^j_0 Z_0 ] = \E[(S^j_{\sigma \wedge \tau_n} - S^j_0) Z_T] = \E[(e^{j,\sigma, n} \sint S_T) H]  =0.
\end{align*}
Since $\sigma$ was arbitrary, it follows that $(Z S^j)^{\tau_n}$ is a $\Pb$-martingale. As $(\tau_n)$ is a localising sequence, we conclude that $Z S^j$ is a local $\Pb$-martingale.
	
(b) $\Rightarrow$ (a): 	Fix $\vt \in \overline{\Theta}(S)$. By Proposition \ref{prop:LOP}, $ Z (\vt \sint S)$ is a $\Pb$-martingale, and hence
\begin{align*}
\E[(\vt \sint S_T -H)^2] &= \E[H^2] - 2 \E[(\vt \sint S_T) Z_T] + \E[(\vt \sint S_T)^2] \\
&= \E[H^2]  + \E[(\vt \sint S_T)^2]  \\
&\geq \E[(H-0 \sint S_T)^2].
\end{align*}
Thus $0 \in \overline{\Theta}(S)$ solves the MVH problem \eqref{eq:MVH}.
\end{proof}

\subsection{Dynamic theory}
\label{sec:MVH:dynamic}
We introduce here some  key concepts from the dynamic theory of mean--variance hedging, for which we use \v Cern\' y/Kallsen \cite{cerny:kallsen:07} as a reference, that are used later in the proof of Proposition \ref{prop:special:prop:quad:eq.opp:proc:bar:z}. \textbf{For the remainder of this section}, we once again impose Assumption \ref{ass:NFL}, which is satisfied due to~\eqref{eq:esssup}. Hence \cite[Assumption 2.1]{cerny:kallsen:07} is satisfied, since Lemma \ref{thm:Qequilibria} gives  an equivalent local martingale measure $\Q(\bar \gamma)$ for $S(\bar \gamma)$ with bounded density.

As in \cite[Definition 3.3]{cerny:kallsen:07}, we introduce the \emph{opportunity process} $L(\bar \gamma)$ as the value process of the pure investment problem, which is given by
\begin{equation}\label{eq:def:opp:proc:chap:iv}
L_t(\bar \gamma)  = \essinf_{\vartheta \in \overline{\Theta}_{t,T}(S(\bar \gamma))} \E\Big[\Big(1-\big(\vartheta\sint S(\bar \gamma)\big)_T \Big)^2 \Bigm \vert \cF_t \Big], \quad 0 \leq t \leq T,
\end{equation}
where $\overline{\Theta}_{t,T}(S(\bar \gamma)) \subseteq \overline{\Theta}(S(\bar \gamma))$ is the set of admissible strategies $\vt$ such that $\vt \I_{\llbracket 0,t \rrbracket} = 0$. We have by  \cite[Corollary 3.4 and Lemma 3.10]{cerny:kallsen:07} that $L(\bar \gamma)$ is an $(0,1]$-valued submartingale with $L_T(\bar \gamma) = 1$. Note that  $L_0(\bar \gamma)$ coincides with $\ell(\bar \gamma)$ given in \eqref{eq:ell}.  By \cite[Lemma 3.1]{cerny:kallsen:07}, there exists for each $t \in [0,T]$ a unique optimal strategy $\vt^{(t)}(1;S(\bar \gamma)) \in \TB_{t,T}(S(\bar \gamma))$ to \eqref{eq:def:opp:proc:chap:iv}; we say that it is the \emph{optimal pure investment strategy started at time $t$}. 
	
Next, we introduce the \emph{mean value process} $(\bar V_t(\bar \gamma))_{0 \leq t \leq T}$ for $\bar H(\bar \gamma)= \bar \gamma-\bar \Xi$ in the sense of \cite[Definition 4.2]{cerny:kallsen:07}. By \cite[Lemmas 3.7 and 4.1 and Proposition 3.13.1]{cerny:kallsen:07}, $\bar V(\bar \gamma)$ is the unique semimartingale  such that
$$\bar V_t(\bar \gamma) = \frac{1}{L_t(\bar \gamma)} \E \bigg[\bar H(\bar \gamma) \bigg(1- \Big(\vt^{(t)}\big(1;S(\bar \gamma)\big) \sint S(\bar \gamma)\Big)_T \bigg) \biggm \vert \cF_t \bigg], \quad 0 \leq t \leq T.$$
In particular, $\bar V_T(\bar \gamma) = \bar H(\bar \gamma)$. By \cite[Lemma 4.1]{cerny:kallsen:07}, the process $(\bar V_s(\bar \gamma) M^{(t)}_s(\bar \gamma))_{t \leq s \leq T}$ is a $\Pb$-martingale on $[t,T]$ for any $t \in [0,T]$, where $(M^{(t)}_s(\bar \gamma))_{t \leq s \leq T}$ is the $\Pb$-martingale (see \cite[Lemma 3.2]{cerny:kallsen:07}) defined on $[t,T]$ by
\begin{equation}\label{eq:def:mart:assoc:mv:proc:chap:iv}
M^{(t)}_s(\bar \gamma):= L_s(\bar \gamma)\bigg(1- \Big(\vt^{(t)}\big(1;S(\bar \gamma)\big) \sint S(\bar \gamma)\Big)_s\bigg), \quad 0 \leq t \leq s \leq T.
\end{equation}
	
The key property for our purposes is that $\bar V_t(\bar \gamma)$ satisfies the inequality
\begin{align*}
& \essinf_{\vartheta \in \overline{\Theta}_{t,T}(S(\bar \gamma))} \E\Big[\Big(\bar H(\bar \gamma) - \bar V_t(\bar \gamma) -\big(\vartheta\sint S(\bar \gamma)\big)_T \Big)^2 \Bigm \vert \cF_t \Big] \\
&\leq \essinf_{\vartheta \in \overline{\Theta}_{t,T}(S(\bar \gamma))} \E\Big[\Big(\bar H(\bar \gamma) - U -\big(\vartheta\sint S(\bar \gamma)\big)_T \Big)^2 \Bigm \vert \cF_t \Big]
\end{align*}
for any $\cF_t$-measurable random variable $U$; this follows by \cite[Theorem 4.10.2]{cerny:kallsen:07}. In particular, $\bar V_0(\bar \gamma)$ is the first component of the solution to the exMVH problem  \eqref{eq:MVHex} for $\bar H(\bar \gamma)$, and hence $\bar V_0(\bar \gamma) = c(\bar H(\bar \gamma);S(\bar \gamma)) = \bar \gamma - c(\bar \Xi;S(\bar \gamma))$.

\section{Auxiliary results and proofs for Section~\ref{sec:quadratic}}
\label{sec:proofs}

\begin{proof}[Proof of Lemma \ref{lem:ind:opt}] 
Let $\vt \in \overline{\Theta}(S)$ and set $\tilde \vt := \vt - \eta^k \in \TB(S)$. Rewriting  $U^{\rm Q}_k(x) =  -(x-\gamma_k)^2+\gamma_k^2$ yields
\begin{align*}
\E\big[U^{\rm Q}_k\big((\vt - \eta^k) \sint S_T + \Xi^k \big)\big] &=  \E[U^{\rm Q}_k(\tilde \vt \sint S_T + \Xi^k )] = - \E[(\tilde \vt \sint S_T -  H^k)^2]  + \gamma_k^2.
\end{align*}
Because $\TB(S)$ is a vector space, this shows that $\vt$ is a solution to the maximisation problem \eqref{eq:def:gen:quad:util:indiv:prob} if and only if $\tilde \vt$ is a solution to the MVH problem \eqref{eq:lem:ind:mod}, and therefore the two problems are equivalent under the assumption that $\eta^k \in \overline{\Theta}(S)$. In particular, \eqref{eq:def:gen:quad:util:indiv:prob} has a unique solution $\hat \vt^k$ if and only if \eqref{eq:lem:ind:mod} has a unique solution $\vt^{\rm MVH}(H^k)$, in which case we have the relationship $\hat \vt^k  \eq{S}  \eta^k + \vt^{\rm MVH}(H^k)$ between the solutions.
\end{proof}

\begin{proof}[Proof of Proposition \ref{prop:ind:opt:decomp}]
Since $\vt^{{\rm MVH}}(1)$ is the unique solution to \eqref{eq:pure:inv:sec:quad}, we have by Proposition \ref{prop:UOGP} that $(1,S)$ satisfies uniqueness of gains processes. Thus the assumption that $\E[(\vt^{{\rm MVH}}(1)\sint S_T - 1)^2]>0$ implies by part 2) of Proposition \ref{prop:UOVP link UOGP}  that $(1,S)$ also satisfies uniqueness of value processes; see Definition \ref{def:uni:gains:values}. Then by Corollary \ref{cor:UOVP link UOGP}, there exists a unique solution $\vt^{\rm MVH}(\Xi^k) \in \overline{\Theta}(S)$ to the MVH problem \eqref{eq:lem:ind:mod} for $\Xi^k$ if and only if there exists a unique solution to the exMVH problem for $\Xi^k$, and we have
\begin{equation}\label{eq:mvh:prob:hk:sol:decomp:aux:xi:k}
\vt^{\rm MVH}(\Xi^k)=  c_k \vt^{{\rm MVH}}(1) + \vt^{\rm ex}(\Xi^k).
\end{equation}
Moreover, by the linearity of MVH (see Lemma \ref{lmm:mvh:prob:linearity}), there exists a unique solution to \eqref{eq:lem:ind:mod} for $\Xi^k$ if and only if there exists a unique solution to \eqref{eq:lem:ind:mod} for $H^k = \gamma_k - \Xi^k$, in which case we have
$$ \vt^{\rm MVH}(H^k) = \gamma_k \vt^{{\rm MVH}}(1) - \vt^{\rm MVH}(\Xi^k).$$
Together with \eqref{eq:mvh:prob:hk:sol:decomp:aux:xi:k} and Lemma \ref{lem:ind:opt}, this shows the equivalence and gives \eqref{eq:mvh:prob:hk:sol:decomp}.
\end{proof}

\begin{proof}[Proof of Lemma \ref{lem:rep:agent:rel to ind}]
By the implication (b) $\Rightarrow$ (c) in Proposition \ref{prop:UOGP}, the map $H \mapsto \vt^{\rm MVH}(H)$ is well defined for all $H$ such that a solution $\vt^{\rm MVH}(H)$ to \eqref{eq:MVH} exists, since such a solution is unique up to $S$-equivalence. We also have by Lemma \ref{lmm:mvh:prob:linearity} that  $H \mapsto \vt^{\rm MVH}(H)$ is linear where it is defined. Because $\eta^1, \ldots, \eta^K \in \overline{\Theta}(S)$, we get from Lemma \ref{lem:ind:opt} that the MVH problem \eqref{eq:MVH} for $H^k$ has the unique solution $\vt^{\rm MVH}(H^k) = \hat \vt^k - \eta^k$ for each $k \in \{1, \ldots, K\}$.  Hence there is a unique solution to \eqref{eq:MVH} for $\bar H = \sum_{k=1}^{K} H^k$, which is given by
\begin{equation}
\label{eq:lem:rep:agent:rel to ind}
\vt^{\rm MVH}(\bar H)  \eq{S} \sum_{k = 1}^K \vt^{\rm MVH}(H^k) \eq{S} \sum_{k = 1}^K \big(\hat \vt^k - \eta^k \big) = \bar \vt - \bar \eta.
\end{equation}
Thus the claim follows by Lemmas \ref{lem:rep:agent} and \ref{lem:ind:opt}, as we have
$$\bar \vt \eq{S}  \vt^{\rm MVH}(\bar H) + \bar \eta \eq{S} \sum_{k=1}^K \big(\vt^{\rm MVH}(H^k) + \eta^k \big) \eq{S} \sum_{k=1}^K \hat \vt^k,$$
which shows \eqref{eq:quad:equil:rep:agent:strat:sum:indiv} and concludes the proof.
\end{proof}

\begin{proof}[Proof of Lemma \ref{lmm:quad:eq:mark:zs:is:loc:mart}]
Denote by $\hat \vt^1, \ldots, \hat \vt^K \in \overline{\Theta}(S)$ the unique individually optimal strategies. Then by Lemma \ref{lem:rep:agent:rel to ind}, $\bar \vt:= \sum_{k = 1}^K  \hat \vt^k \in \overline \Theta(S)$ is the unique solution to the optimisation problem \eqref{max:rep:agent} of the representative agent. Moreover, the market clearing condition \eqref{eq:market clearing} yields $\bar \vt \eq{S} \bar \eta$, so that by Lemma \ref{lem:rep:agent}, $0$ is the unique solution to the MVH problem \eqref{eq:MVH} for $\bar H$. Thus Lemma \ref{lem:MVH:zero} yields that $(\bar Z_t S^j_t)_{0 \leq t \leq T}$ is a local $\Pb$-martingale for each $j \in \{1, \ldots, d_1 + d_2\}$, as claimed.
\end{proof}

\begin{remark}\label{rmk:gkw:decomp:bar:g:m:1:detail:pred:quad:var}
The choice of $\bar \xi^{(1)}$ in \eqref{eq:gkw:decomp:bar:g:m:1} is only unique up to $M^{(1)}$-equivalence. Because the components of $M^{(1)}$ may be linearly dependent, the components $\bar \xi^i \sint M^i$ need not be well defined in general (see Cherny/Shiryaev \cite{CS2002}), but we choose a particular integrand $\bar \xi^{(1)}= (\bar \xi^1,\ldots,\bar \xi^{d_1})$ with $\bar \xi^i \in L^2(M^i)$ for $i \in \{1, \ldots, d_1\}$ as follows: Applying the Gram--Schmidt algorithm to $(M^1, \ldots, M^{d_1}, \bar Z) \in \cM^{2}_{0,\loc}$ yields a unique decomposition of the form
$$\bar Z_t = \bar Z_0 + \sum_{i=1}^{d_1} \bar \xi^{i} \sint M^{i}_t + M^{\bar  {Z}}_t, \quad 0 \leq t \leq T,$$
where $ \sum_{i=1}^{I} \bar \xi^{i} \sint M^{i} $ is strongly orthogonal to $\sum_{i=I+1}^{d_1} \bar \xi^{i} \sint M^{i} + M^{\bar Z}$ for each \mbox{$I \in \{1, \ldots, d_1\}$}. This orthogonality property and the square-integrability of $\bar Z$ yield that $ \sum_{i=1}^{I} \bar \xi^{i} \sint M^{i}$ is a square-integrable martingale for each $I$, and hence so is $\bar \xi^{i} \sint M^{i}$. For this choice of $\bar \xi^{(1)} := (\bar \xi^1, \ldots, \bar \xi^{d_1})$ and as $M^{(1)}$ is a locally square-integrable martingale by assumption, the predictable quadratic variation
$$\langle \bar \xi^{i} \sint M^{i}, M^j\rangle = {{}\bar \xi^{i}}^{\top} \!\! \sint \langle M^{i}, M^j\rangle$$
is well defined for each $j \in \{1, \ldots, d_1\}$ and ${{}\bar \xi^{i}}^\top \!\! \sint ([ M^{i}, M^j] - \langle M^{i}, M^j\rangle)$ is a local $\Pb$-martingale. In particular, the process
\begin{equation}\label{eq:gkw:decomp:bar:g:m:1:detail:pred:quad:var}
\langle \bar \xi^{(1)} \sint M^{(1)}, M^j\rangle = \sum_{i = 1}^{d_1} \bar \xi^i \sint \langle M^i, M^j \rangle
\end{equation}
is well defined. Although we choose $\bar \xi^{(1)}$ as above, note that the martingale $\bar \xi^{(1)} \sint M^{(1)}$ is independent of that choice due to \eqref{eq:gkw:decomp:bar:g:m:1}. Hence the right-hand side of \eqref{eq:gkw:decomp:bar:g:m:1:detail:pred:quad:var} is also independent of any choice of $\bar \xi^{(1)}$ such that the individual summands are well defined. Likewise, the finite-variation part in \eqref{eq:thm:equil:S1} below also does not depend on the choice of $\bar \xi^{(1)}$.
\end{remark}

The following result is a generalisation of Theorem \ref{thm:equil:new}.

\begin{theorem}
\label{thm:equil}
Assume that $\bar Z_t \neq 0$ and $\bar Z_{t-} \neq 0$ for all  $t \in [0, T]$ $\Pas$ If there exists a quadratic equilibrium $(1, S^{(1)}, S^{(2)})$, it is unique and explicitly given by
\begin{align}
S^j_t &= S_0^j + M^j_t - \int_0^t  \frac{\dd \langle \bar Z, M^j \rangle_s}{\bar Z_{s-}} \notag\\
\label{eq:thm:equil:S1}
&= S_0^j + M^j_t - \sum_{i = 1}^{d_1} \int_0^t \frac{\bar \xi^i_s}{\bar Z_{s-}}\dd \langle M^i, M^j \rangle_s,
\quad j \in \{1, \ldots, d_1\}, \\
\label{eq:thm:equil:S2}
S^j_t &= \frac{E[\bar H D^j \mid \cF_t]}{\bar Z_t} = \frac{E[\bar Z_T D^j \mid \cF_t]}{\bar Z_t},
\quad j \in \{d_1+1, \ldots, d_1 + d_2\}.
\end{align}
Conversely, the market $(1, S) = (1, S^{(1)}, S^{(2)})$ defined by \eqref{eq:thm:equil:S1} and \eqref{eq:thm:equil:S2} is a quadratic equilibrium if and only if $S^{(2)}$ is a local $L^2$-semimartingale and for each \mbox{$j \in \{d_1+1, \ldots, d_1+d_2\}$}, we have $e^j \in \overline{\Theta}(S)$.
\end{theorem}

\begin{proof}

\textbf{(a)} We show that a quadratic equilibrium $(1, S)$ must be given by \eqref{eq:thm:equil:S1} and \eqref{eq:thm:equil:S2}. Lemma \ref{lmm:quad:eq:mark:zs:is:loc:mart} yields for each $j \in \{1, \ldots, d_1+d_2\}$ that $\bar Z S^j$ is a local $\Pb$-martingale. 
We first consider $j \in \{1, \ldots, d_1\}$. Recall the decomposition \eqref{eq:gkw:decomp:bar:g:m:1} for $\bar Z$ and the dynamics \eqref{eq:def:prim:fin:assets:decomp} for $S^j$. Applying the product formula to $\bar Z S^j$ and rearranging terms, we obtain
\begin{align}
&\bar Z_t S^j_t - \bar Z_{-} \sint M^j_t - S^j_{-}\sint\bar Z_t -  \sum_{i = 1}^{d_1} \bar \xi^{i} \sint \big([ M^{i}, M^j] - \langle M^{i}, M^j\rangle\big)_t \notag \\
&= \bar Z_{-} \sint A^j_t + \sum_{i = 1}^{d_1} \bar \xi^{i} \sint \langle M^{i}, M^j\rangle_t \label{eq:pf:thm:equil:FVmart}
\end{align}
for $0\leq t \leq T$. Note that $\bar ZS^j$, $M^j$, $\bar Z$ and $\bar \xi^i \sint([M^i, M^j ] - \langle M^i, M^j \rangle)$ are local $\Pb$-martingales (for the latter, this is shown  in Remark \ref{rmk:gkw:decomp:bar:g:m:1:detail:pred:quad:var}), whereas $A^j$ and $\bar \xi^i \sint \langle M^i, M^j\rangle$ are predictable finite-variation processes. Thus both sides of \eqref{eq:pf:thm:equil:FVmart} must vanish, as they are null at $0$. By assumption, we have $\bar Z_t \neq 0$ and $\bar Z_{t-} \neq 0$ for all  $t \in [0, T]$ $\Pas$ Since $\bar Z$ is also càdlàg, this implies that $1/\bar Z_-$ is finite-valued and càglàd, thus locally bounded.
Integrating $1/\bar Z_{-}$ against the right-hand side of \eqref{eq:pf:thm:equil:FVmart}, which vanishes as we have shown, yields $ A^j = - \sum_{i = 1}^{d_1} \frac{\bar \xi^{i}}{\bar Z_{-}} \sint \langle M^{i}, M^j\rangle$.  Plugging into \eqref{eq:def:prim:fin:assets:decomp} then shows \eqref{eq:thm:equil:S1}.

Next, consider $j\in \{d_1+ 1, \ldots, d_1+d_2\}$.  By \eqref{eq:def:prim:prod:assets:divid} and as  $\bar Z_T = \bar H$, we have $\bar Z_T S^j_T = \bar H D^j$. Since $e^j  \in \overline{\Theta}(S)$ and $\bar Z$ is a square-integrable martingale, it follows from Proposition \ref{prop:LOP} that $\bar Z S^j =  \bar Z S^j_0 + \bar Z (e^j\sint S)$ is a $\Pb$-martingale, so that $\bar Z_t S^j_t = E[\bar H D^j \mid \cF_t].$ Since $\bar Z_t \neq 0$ $\Pas$, this yields \eqref{eq:thm:equil:S2}. We have thus shown that any equilibrium $(1,S)$ must satisfy \eqref{eq:thm:equil:S1} and \eqref{eq:thm:equil:S2}.

\textbf{(b)} Next, we show the converse statement. Define $(1, S) = (1, S^{(1)}, S^{(2)})$ by \eqref{eq:thm:equil:S1} and \eqref{eq:thm:equil:S2} and assume that $S^{(2)}$ is a local $L^2$-semimartingale and $e^j  \in \overline{\Theta}(S)$ for $j \in \{d_1+1, \ldots, d_1+d_2\}$. We claim that $(1,S)$ is a quadratic equilibrium. It is clear from \eqref{eq:thm:equil:S1} and \eqref{eq:thm:equil:S2} that $S^{(1)}$ and $S^{(2)}$ satisfy \eqref{eq:def:prim:fin:assets:decomp} and \eqref{eq:def:prim:prod:assets:divid}, respectively. Note that $S^{(1)}$ is a special semimartingale and the local martingale part $M^{(1)}$ is locally square-integrable, by assumption. Thus by \v Cern\' y/Kallsen \cite[Lemma A.2]{cerny:kallsen:07}, $S^{(1)}$ is also a local $L^2$-semimartingale so that $(1, S^{(1)}, S^{(2)})$ is a local $L^2$-market.
Next, we want to show that $\bar Z S^j$ is a local $\Pb$-martingale for $j \in \{1, \ldots, d_1 + d_2\}$. This is clear for $j  \in \{d_1+1, \ldots, d_1 + d_2\}$ by the construction \eqref{eq:thm:equil:S2}. For $j  \in \{1, \ldots, d_1 \}$, we use a result on local $\cE$-martingales by Choulli et al.~\cite{choullietal:1998} as in the proof of part 2) of Proposition \ref{prop:LOP}. Indeed, the assumptions on  $\bar Z$ yield $\bar Z = \bar Z_0 \,\cE(\bar N)$ for some local $\Pb$-martingale $\bar N = (\bar N_t)_{0 \leq t \leq T}$, namely, $\bar N = (1/\bar Z_-) \sint \bar Z$.
Since for $j \in\{1, \ldots, d_1\}$, we have
$$\sum_{i = 1}^{d_1}  \frac{\bar \xi^i_t}{\bar Z_{t-}}\diff \langle M^i, M^j \rangle_t = \frac{1}{\bar Z_{t-}} d \langle \bar Z, M^j \rangle_t = d \langle \bar N, M^j \rangle_t,$$
we obtain that $S^{j}$ given by \eqref{eq:thm:equil:S1} is a local $\cE$-martingale by \cite[Corollary 3.16]{choullietal:1998} (which generalises Girsanov's theorem to local $\cE$-martingales). Thus by \cite[Definition 3.11]{choullietal:1998} with $n=0$, $\bar Z S^j$ is a local $\Pb$-martingale.

Now note that $\bar Z_{t} \neq 0$  and $\bar Z_{t-} \neq 0$ for all $t \in [0, T]$ $\Pas$ by the assumptions on $\bar Z$, and $\bar Z S^j$ is a local $\Pb$-martingale for $j \in \{1, \ldots, d_1 + d_2\}$ as shown above. Hence for each agent $k \in \{1, \ldots, K\}$, the MVH problem \eqref{eq:MVH} for $H^k$ has a unique solution $\vt(H^k)$ by part 2) of Proposition \ref{prop:LOP}. Since moreover \mbox{$\eta^k \in \TB(S)$} by the assumption on $S^{(2)}$, it follows by  Lemma  \ref{lem:ind:opt} that  the individual optimisation problem \eqref{eq:def:gen:quad:util:indiv:prob} for agent $k$ has a unique solution $\hat \vt^k$. This shows condition 1) in Definition \ref{def:equilibrium} of an equilibrium market. Moreover, the strategy $0$ solves the MVH problem \eqref{eq:MVH} for $\bar H$ by Lemma \ref{lem:MVH:zero}. Thus Lemmas \ref{lem:rep:agent} and \ref{lem:rep:agent:rel to ind} yield \mbox{$\sum_{k=1}^K \hat \vt^k = \bar \vt = \bar \eta$}, i.e., the market clears and condition 2) is satisfied. Finally, condition 3) is satisfied by assumption, and thus $(1,S)$ is a quadratic equilibrium.
\end{proof}

\begin{proof}[Proof of Theorem \ref{thm:equil:new}]
The result follows directly from Theorem \ref{thm:equil} and the fact that $\ol Z S^j$ is a local $\Pb$-martingale for each $j \in \{1, \ldots, d_1 +d_2\}$ by Lemma \ref{lmm:quad:eq:mark:zs:is:loc:mart}.
\end{proof}

\begin{proof}[Proof of Lemma \ref{lem:quad equil:suf cond}]
(a)	This is (c) for $p_1 = 1$ and $p_2 = 2$, so that $q_1 = \infty$ and $q_2 = 2$.
	
(b) This is (c) for $p_1 = p_2 =\infty$.

(c) We only consider the case that $p_1, p_2  \in (1, \infty)$. The arguments for the other cases are  very similar and therefore omitted. Fix $j \in \{d_1 + 1, \ldots, d_1 + d_2\}$ and define $\bar \Q \approx \Pb$ by
$$\frac{\dd \bar \Q}{\dd \Pb} = \frac{\bar H}{\E[\bar H]} =: Z^{\bar Q}_T \in L^{1-p_1}(\Pb) \cap L^{p_2}(\Pb).$$
By \eqref{eq:thm:equil:S2} and the Bayes rule, $S^j$ is a (true) $\bar \Q$-martingale with $S^j_T = D^j$. Thus the inequalities of Hölder and Doob (with constant $C_{q_1}$) give
\begin{align*}
\E_{\Pb} \Big[\sup_{t \in [0,T]} |S^j_t|^2\Big] &= E_{\bar \Q}\bigg[\frac{1}{Z^{\bar Q }_T}\sup_{t \in [0,T]} |S^j_t|^2\bigg] \leq  \E_{\bar \bar \Q}\bigg[\bigg(\frac{1}{Z^{\bar Q}_T}\bigg)^{p_1}\bigg]^{1/p_1}  \E_{\bar \Q}\Big[\sup_{t \in [0,T]} |S^{j}_t|^{2q_1}\Big]^{1/q_1} \\
&\leq  \E_{\Pb}\bigg[\bigg(\frac{1}{Z^{\bar Q}_T}\bigg)^{p_1-1}\bigg]^{1/p_1} C_{q_1} \E_{\bar \Q}[(D^j)^{2 q_1}]^{1/q_1}  \\
&= C_{q_1} \E_{\Pb}\big[(Z^{\bar Q}_T)^{1-p_1}\big]^{1/p_1}\E_{\Pb}\big[Z^{\bar Q}_T (D^j)^{2 q_1}\big]^{1/q_1} \\
&\leq C_{q_1} \E_{\Pb}\big[(Z^{\bar Q}_T)^{1-p_1}\big]^{1/p_1} \E_{\Pb}\big[\big(Z^{\bar Q}_T \big)^{p_2}\big]^{1/(q_1 p_2)}\E_{\Pb}[(D^j)^{2 q_1 q_2}]^{1/(q_1 q_2)} \\
&< \infty
\end{align*}
by the assumptions. This implies that $S^j$ is an $L^2$-semimartingale.
\end{proof}

\section{Auxiliary results and proofs for Section~\ref{sec:MV}}
\label{app:auxiliary}

\begin{lemma}
\label{lem:MVefficient0}
Suppose the market $(1,S)$ satisfies Assumption~\ref{ass:NFL} and $S$ is not a local $\Pb$-martingale. A strategy $\vartheta \in \overline{\Theta}(S)$ is mean-variance efficient with respect to $H = 0$ if and only if $\vartheta =_S a\vartheta^{\MVH}(1)$ for some $a > 0$. In that case, we have
\[
\E[\vartheta \bullet S_T] = a(1-\ell) \quad \text{and} \quad \Var[\vartheta \bullet S_T] = a^2\ell(1-\ell).
\]
\end{lemma}

\begin{proof}
Since a strategy is mean-variance efficient if and only if it satisfies the equivalent conditions (a) and (b) of Eberlein/Kallsen~\cite[Rule 10.43]{eberlein:kallsen:19}, the first assertion follows directly from the equivalent condition (e) in~\cite[Rule 10.43]{eberlein:kallsen:19}. Then, the second assertion is given by~\cite[Rule 10.47]{eberlein:kallsen:19}.
\end{proof}

\begin{lemma}
\label{lem:MVefficientk}
For any $k \in \{1,...,K\}$ and $\vartheta \in \overline{\Theta}(S)$, we have the decompositions
\[
\E[V_T^k(\vartheta)] = c_k + \E[(\vartheta - \eta^k + \vartheta^{\ex}(\Xi^k)) \bullet S_T] \quad \text{and} \quad \Var[V_T^k(\vartheta)] = \varepsilon_k + \Var[(\vartheta - \eta^k + \vartheta^{\ex}(\Xi^k)) \bullet S_T].
\]
Consequently, $\vartheta$ is mean-variance efficient for agent $k$ if and only if $\vartheta - \eta^k + \vartheta^{\ex}(\Xi^k)$ is mean-variance efficient with respect to $H = 0$.
\end{lemma}

\begin{proof}
Since $(c_k,\vartheta^{\ex}(\Xi^k))$ is the unique solution to the exMVH problem~\eqref{eq:MVHex} with $H = \Xi^k$, we have that $c_k + \vartheta^{\ex}(\Xi^k) \bullet S_T$ is the orthogonal projection of $\Xi^k$ onto the set
\[
\{x + \vartheta \bullet S_T: x \in \R, \vartheta \in \overline{\Theta}(S)\} \subseteq L^2,
\]
which is closed in $L^2$ by \v{C}ern\'{y}/Kallsen~\cite[Lemma 2.9]{cerny:kallsen:08} and Assumption~\ref{ass:NFL}. Thus, we obtain the orthogonal decomposition of the form
\begin{equation}
\label{eq:orthgonaldecomp}
\Xi^k = c_k + \vartheta^{\ex}(\Xi^k) \bullet S_T + \tilde{\Xi}^k,
\end{equation}
where $\tilde{\Xi}^k \in L^2$ and satisfies $\E[\tilde{\Xi}^k] = \E[(\tilde{\vartheta}\bullet S_T)\tilde{\Xi}^k] = 0$ for all $\tilde{\vartheta} \in \overline{\Theta}(S)$. Moreover, we have by~\eqref{eq:varepsilon2} that
\[
\Var[\tilde{\Xi}^k] = \varepsilon_k.
\]
Now, for any $\vartheta \in \overline{\Theta}(S)$, plugging~\eqref{eq:orthgonaldecomp} in the definition of $V_T^k(\vartheta)$ yields
\begin{align*}
\E[V_T^k(\vartheta)] &= \E[(\vartheta - \eta^k) \bullet S_T + c_k + \vartheta^{\ex}(\Xi^k) \bullet S_T + \tilde{\Xi}^k] \\
&= c_k + \E[(\vartheta - \eta^k + \vartheta^{\ex}(\Xi^k)) \bullet S_T], \\
\Var[V_T^k(\vartheta)] &= \Var[(\vartheta - \eta^k) \bullet S_T + c_k + \vartheta^{\ex}(\Xi^k)\bullet S_T + \tilde{\Xi}^k] \\
&= \varepsilon_k + \Var[(\vartheta - \eta^k + \vartheta^{\ex}(\Xi^k)) \bullet S_T].
\end{align*}
The consequent statement follows by the fact that $c_k$ and $\varepsilon_k$ do not depend on $\vartheta$.
\end{proof}

\begin{proof}[Proof of Proposition~\ref{prop:efficientfrontier}]
(a) $\Leftrightarrow$ (b): By Lemma~\ref{lem:MVefficientk}, it follows that $\vartheta$ is mean-variance efficient for agent $k$ if and only if $\vartheta - \eta^k + \vartheta^{\ex}(\Xi^k)$ is mean-variance efficient with respect to $H = 0$. By Lemma~\ref{lem:MVefficient0}, the latter statement is equivalent to
\[
\vartheta - \eta^k + \vartheta^{\ex}(\Xi^k) =_S a_k\vartheta^{\MVH}(1),
\]
for some $a_k > 0$. In this case, $\gamma_k > c_k$ given by $\gamma_k := a_k + c_k$ is the unique choice so that $\vartheta =_S \vartheta^k(\gamma_k)$.

(b) $\Leftrightarrow$ (c): Since $(c_k,\vartheta^{\ex}(\Xi^k))$ is the unique solution to the exMVH problem~\eqref{eq:MVHex} with $H = \Xi^k$, it follows by fixing $c_k$ that $\vartheta^{\ex}(\Xi^k)$ is also the unique solution to the MVH problem~\eqref{eq:MVH} with $H = \Xi^k - c_k$, so that $\vartheta^{\ex}(\Xi^k) = \vartheta^{\MVH}(\Xi^k - c_k)$. Then, by the linearity property of the MVH problem (Lemma~\ref{lmm:mvh:prob:linearity}), we have for any $\gamma_k > c_k$ that
\[
\vartheta^k(\gamma_k) = (\gamma_k - c_k)\vartheta^{\MVH}(1) + \eta^k - \vartheta^{\ex}(\Xi^k) = \vartheta^{\MVH}(\gamma_k - \Xi^k) + \eta^k.
\]
Thus, for a strategy $\vartheta \in \overline{\Theta}(S)$, we have that $\vartheta =_S \vartheta^k(\gamma_k)$ if and only if $\vartheta - \eta^k$ is the unique solution to the MVH problem~\eqref{eq:MVH} with $H^k(\gamma_k) = \gamma_k - \Xi^k$.

(c) $\Leftrightarrow$ (d): By linearity of the MVH problem (Lemma~\ref{lmm:mvh:prob:linearity}), $\vartheta$ is a solution to the quadratic utility problem~\eqref{eq:def:gen:quad:util:indiv:prob} with risk tolerance $\gamma_k > c_k$ if and only if $\vartheta - \eta^k$ is a solution to the MVH problem~\eqref{eq:MVH} with $H^k(\gamma_k) = \gamma_k - \Xi^k$; in particular, the solution to~\eqref{eq:def:gen:quad:util:indiv:prob} is unique by the uniqueness of solutions to~\eqref{eq:MVH}.

We note that in all cases, $\gamma_k > c_k > 0$ due to~\eqref{eq:ckpositive}.
\end{proof}

\begin{proof}[Proof of Corollary~\ref{cor:efficientfrontier}]
By the equivalence of (a) and (b) in Proposition~\ref{prop:efficientfrontier}, we can express the efficient frontier $\cE^k$ for agent $k$ as
\[
\cE^k = \Big\{\Big(\E[V_T^k(\vartheta^k(\gamma_k))], \sqrt{\Var[V_T^k(\vartheta^k(\gamma_k))]}\Big): \gamma_k \geq c_k\Big\},
\]
where $\vartheta^k(\gamma_k) = (\gamma_k - c_k)\vartheta^{\MVH}(1) + \eta^k - \vartheta^{\ex}(\Xi^k)$. Then, by Lemmas~\ref{lem:MVefficient0} and~\ref{lem:MVefficientk}, we get
\begin{align*}
\E[V_T^k(\vartheta^k(\gamma_k))] &= c_k + \E[(\gamma_k - c_k)\vartheta^{\MVH}(1) \bullet S_T] = c_k + (1-\ell)(\gamma_k - c_k) = \mu_k(\gamma_k), \\
\Var[V_T^k(\vartheta^k(\gamma_k))] &= \varepsilon_k + \Var[(\gamma_k - c_k)\vartheta^{\MVH}(1) \bullet S_T] = \varepsilon_k + \ell(1-\ell)(\gamma_k - c_k)^2 = \sigma_k^2(\gamma_k),
\end{align*}
which shows~\eqref{eq:efficientfrontierexpl}.
\end{proof}

\begin{proof}[Proof of Corollary~\ref{cor:MVoptprobequiv}]
It suffices to show that for any strategy $\vartheta \in \overline{\Theta}(S)$, there exists a mean-variance efficient strategy $\vartheta' \in \overline{\Theta}(S)$ for agent $k$ such that
\begin{equation}
\label{eq:MVefficient}
\E[V_T^k(\vartheta')] \geq \E[V_T^k(\vartheta)] \quad \text{and} \quad \Var[V_T^k(\vartheta')] \leq \Var[V_T^k(\vartheta)].
\end{equation}
Using also the fact that each mean-variance utility function $U_k^{\MV}$ is strictly increasing in the first variable and strictly decreasing in the second variable, we deduce that a strategy $\hat{\vartheta}^k \in \overline{\Theta}(S)$ is a solution to the maximisation problem~\eqref{eq:MVoptprob} if and only if it maximises~\eqref{eq:MVoptprob} among all mean-variance efficient strategies for agent $k$. Then, the result follows by plugging the parameterisation of the efficient frontier in Corollary~\ref{cor:efficientfrontier}.

Let us now prove the above claim. First, note that the decompositions in Lemma~\ref{lem:MVefficientk} hold for any $\vartheta \in \overline{\Theta}(S)$. In the case that $\ell \in (0,1)$, we have by Corollary~\ref{cor:efficientfrontier} and Lemma~\ref{lem:MVefficient0} that the mean-variance efficient strategy $\vartheta' = \vartheta^k(\gamma_k)$ for agent $k$, where
\[
\gamma_k := c_k + \sqrt{\frac{\Var[(\vartheta - \eta^k + \vartheta^{\ex}(\Xi^k)) \bullet S_T]}{\ell(1-\ell)}},
\]
satisfies $\Var[V_T^k(\vartheta')] = \Var[V_T^k(\vartheta)]$; this secures the second inequality in~\eqref{eq:MVefficient} Moreover, $\vartheta'$ must satisfy the first inequality in~\eqref{eq:MVefficient} because it is mean-variance efficient.

In the case that $\ell = 0$, in that $S$ is a local $\Pb$-martingale, it follows by Proposition~\ref{prop:LOP} that $(\vartheta - \eta^k + \vartheta^{\ex}(\Xi^k)) \bullet S$ is a (true) $\Pb$-martingale. Then, we have by Lemma~\ref{lem:MVefficientk} that
\[
\E[V_T^k(\vartheta^k(c_k))] = c_k = \E[V_T^k(\vartheta)] \quad \text{and} \quad \Var[V_T^k(\vartheta^k(c_k))] = \varepsilon_k \leq \Var[V_T^k(\vartheta)],
\]
which shows the claim~\eqref{eq:MVefficient} with the mean-variance efficient strategy $\vartheta' = \vartheta_k(c_k)$.
\end{proof}

Henceforth, we fix a triplet $\theta = (\ell,c,\varepsilon) \in (0,1) \times \R_+ \times \R_+$ and a mean-variance utility function $U^{\MV}$, and we denote by $\ccS_{U^{\MV}}: \R \times \R_+ \to \R$ the slope of $U^{\MV}$ in the $\sigma$-$\mu$-plane, given by
\[
\ccS_{U^{\MV}}(\mu,\sigma) = - \frac{\partial_\sigma U^{\MV}(\mu,\sigma)}{\partial_\mu U^{\MV}(\mu,\sigma)}, \qquad (\mu,\sigma) \in \R \times \R_+,
\]
and we recall from~\eqref{eq:mvbounds} that it satisfies
\begin{equation}
\label{eq:SUbounds}
\frac{\sigma}{\lambda^{\max}} \leq \ccS_{U^{\MV}}(\mu,\sigma) \leq \frac{\sigma}{\lambda^{\min}}, \qquad (\mu,\sigma) \in \R \times \R_+,
\end{equation}
for some constants $\lambda^{\min},\lambda^{\max} \in (0,\infty)$. In the spirit of Corollaries~\ref{cor:efficientfrontier} and~\ref{cor:MVoptprobequiv}, we define the functions $\mu^\theta: [c,\infty) \to \R$ and $\sigma^\theta: [c,\infty) \to \R_+$ by
\[
\mu^\theta(\gamma) := \ell c + (1-\ell)\gamma \quad \text{and} \quad (\sigma^\theta(\gamma))^2 := \varepsilon + \ell(1-\ell)(\gamma-c)^2, \qquad \gamma \geq c.
\]
We drop the $\theta$-superscript wherever the choice of $\theta \in (0,1) \times \R \times \R_+$ is unambinguous.

\begin{lemma}
\label{lem:quasiconc1oc}
Let $f: A \to \R$ be a continuously differentiable and strictly quasiconcave function, for some non-empty interval $A \subset \R$. Suppose there exist $a,b \in A^\circ$ with $a \leq b$ such that $f'(x) > 0$ for all $x \in A^{\circ} \cap (-\infty,a)$ and $f'(x) < 0$ for all $x \in A^\circ \cap (b,\infty)$. Then, $f$ admits a unique maximiser $x^* \in [a,b]$ which solves the first-order condition $f'(x^*) = 0$. Moreover, there is a nowhere dense subset $B$ of $A^\circ$ so that
\begin{align*}
f'(x) &> 0, \quad x \in (-\infty,x^*) \cap (A^\circ \setminus B), \\
f'(x) &< 0, \quad x \in (x^*,\infty) \cap (A^\circ \setminus B).
\end{align*}
\end{lemma}

\begin{proof}
Since $f$ is continuous, it attains a maximum on the compact interval $[a,b]$. Let $x^* \in [a,b]$ be such a maximiser. By the assumptions on $f'$, we have
\begin{align*}
&f(x) < f(a) \leq f(x^*), &x \in A^{\circ} \cap (-\infty,a), \\
&f(x) < f(b) \leq f(x^*), &x \in A^{\circ} \cap (b,\infty),
\end{align*}
in which case $x^*$ is a global maximiser of $f$ on $A$. Uniqueness of $x^*$ follows by strict quasiconcavity of $f$. Since $x^* \in [a,b] \subseteq A^\circ$, Fermat's theorem gives the first-order condition $f'(x^*) = 0$.

Next, denote by $B \subset A^\circ$ the set of points where $f'$ vanishes. Now, to show that $f'(x) > 0$ for $x \in (-\infty,x^*)\cap (A \setminus B)$, i.e., that $f$ is strictly increasing on $(-\infty,x^*) \cap (A \setminus B)$, suppose by contradiction that there exist $-\infty \leq x_1 < x_2 < x^*$, $x_1,x_2\in (A \setminus B)^\circ$, such that $f(x_1) \geq f(x_2)$. Noting that we can write $x_2$ as a convex combination of $x_1$ and $x^*$, in that $x_2 = \lambda x_1 + (1-\lambda)x^*$ for some $\lambda \in (0,1)$,  we have by strict quasiconcavity of $f$ and optimality of $x^*$ for $f$ that
\[
f(x_2) = f(\lambda x_1 + (1-\lambda)x^*) > \min\{f(x_1),f(x^*)\} = f(x_1),
\]
and so we arrive at a contradiction. Similarly, seeking a contradiction, suppose that there exist $x^* < x_1 < x_2 \leq \infty$, $x_1,x_2 \in (A \setminus B)^\circ$, such that $f(x_1) \leq f(x_2)$. As above, we can write $x_1 = \lambda x^* + (1-\lambda)x_2$ for some $\lambda \in (0,1)$, and by strict quasiconcavity of $f$ and optimality of $x^*$ for $f$, we get
\[
f(x_1) = f(\lambda x^* + (1-\lambda)x_2) > \min\{f(x^*),f(x_2)\} = f(x_2),
\]
which is a contradiction. Note that $B$ is relatively closed because $f'$ is continuous; if $B$ contained a non-empty interval, then $f'=0$ there and the mean value theorem would make $f$ constant there, contradicting strict quasiconcavity and strict monotonicity on either side of $x^*$. Therefore, $B$ has empty interior and is nowhere dense.
\end{proof}

\begin{lemma}
\label{lem:1OC}
There exists a unique maximiser $\hat{\gamma}_{U^{\MV}} \geq c$ to~\eqref{eq:MVoptprobabstract}. Moreover, $\hat{\gamma}_{U^{\MV}}$ satisfies the bounds $\hat{\gamma}_{U^{\MV}} \in [c + \lambda^{\min}/\ell, c + \lambda^{\max}/\ell]$ and solves the first-order condition
\begin{equation}
\label{eq:1OC}
\ccS_{U^{\MV}}(\mu(\hat{\gamma}_{U^{\MV}}),\sigma(\hat{\gamma}_{U^{\MV}})) = \frac{\sigma(\hat{\gamma}_{U^{\MV}})}{\ell(\hat{\gamma}_{U^{\MV}}-c)} = \frac{\sigma(\hat{\gamma}_{U^{\MV}})}{\hat{\gamma}_{U^{\MV}}- \mu(\hat{\gamma}_{U^{\MV}})}.
\end{equation}
In addition, there is a dense subset $A$ of $[c,\infty)$ so that
\begin{align}
\label{eq:1OC2}
\ccS_{U^{\MV}}(\mu(\gamma),\sigma(\gamma)) > \frac{\sigma(\gamma)}{\ell(\gamma-c)} \quad \text{and} \quad \ccS_{U^{\MV}}(\mu(\gamma),\sigma(\gamma)) < \frac{\sigma(\gamma)}{\ell(\gamma-c)},
\end{align}
for every $\gamma \in (\hat{\gamma}_{U^{\MV}}(\theta),\infty) \cap A$ and $\gamma \in [c,\hat{\gamma}_{U^{\MV}}(\theta)) \cap A$, respectively.
\end{lemma}

\begin{proof}
First, we claim that the map $\gamma \mapsto U^{\MV}(\mu(\gamma),\sigma(\gamma))$, $\gamma \geq c$, is strictly quasiconcave. Indeed, for $\gamma_1, \gamma_2 \geq c$, $\gamma_1 \neq \gamma_2$, and $\lambda \in (0,1)$, we have
\begin{align*}
&U^{\MV}(\mu(\lambda \gamma_1 + (1-\lambda)\gamma_2),\sigma(\lambda \gamma_1 + (1-\lambda)\gamma_2)) \\
&\qquad\qquad = U^{\MV}(\lambda \mu(\gamma_1) + (1-\lambda)\mu(\gamma_2),\sigma(\lambda \gamma_1 + (1-\lambda)\gamma_2)) \\
&\qquad\qquad \geq U^{\MV}(\lambda \mu(\gamma_1) + (1-\lambda)\mu(\gamma_2),\lambda\sigma(\gamma_1) + (1-\lambda)\sigma(\gamma_2)) \\
&\qquad\qquad > \min\{U^{\MV}(\mu(\gamma_1),\sigma(\gamma_1)),U^{\MV}(\mu(\gamma_2),\sigma(\gamma_2))\},
\end{align*}
where we use linearity of the map $\gamma \mapsto \mu(\gamma)$, convexity of the map $\gamma \mapsto \sigma(\gamma)$, $\gamma \geq c$ (easy to verify), that the bivariate map $(\mu,\sigma) \mapsto U^{\MV}(\mu,\sigma)$, $(\mu,\sigma) \in \R \times \R_+$, is strictly quasiconcave, and decreasing in the second variable. Now, we have
\begin{align*}
\partial_\gamma U^{\MV}(\mu(\gamma),\sigma(\gamma)) &= \partial_\mu U^{\MV}(\mu(\gamma),\sigma(\gamma))\mu'(\gamma) + \partial_\sigma U^{\MV}(\mu(\gamma),\sigma(\gamma))\sigma'(\gamma) \\
&= (1-\ell)\partial_\mu U^{\MV}(\mu(\gamma),\sigma(\gamma))\Big(1 - \frac{\ell(\gamma - c)}{\sigma(\gamma)}\ccS_{U^{\MV}}(\mu(\gamma),\sigma(\gamma))\Big),
\end{align*}
which, due to the bounds~\eqref{eq:SUbounds} of $\ccS_{U^{\MV}}$, is strictly positive for $\gamma \in (c,c + \lambda^{\min}/\ell)$ and is strictly negative for $\gamma > c + \lambda^{\max}/\ell$. Hence, applying Lemma~\ref{lem:quasiconc1oc} with $A = [c,\infty)$ and $[a,b] = [c + \lambda^{\min}/\ell, c+\lambda^{\max}/\ell]$, the map $\gamma \mapsto U^{\MV}(\mu(\gamma),\sigma(\gamma))$ admits a unique maximiser $\hat{\gamma}_{U^{\MV}} \in [c+\lambda^{\min}/\ell,c+\lambda^{\max}/\ell]$. Moreover, $\hat{\gamma}_{U^{\MV}} > c$, so that the maximiser is indeed an interior point of $[c,\infty)$, hence solving the first-order condition~\eqref{eq:1OC}. The inequalities~\eqref{eq:1OC2} are also a consequence of Lemma~\ref{lem:quasiconc1oc}. This completes the proof.
\end{proof}

For the following result, we write $\hat{\gamma}_{U^{\MV}}(\theta)$ to emphasise the dependence on the chosen triplet $\theta = (\ell,c,\varepsilon)$.

\begin{lemma}
\label{lem:hatgammacty}
The map $\theta = (\ell,c,\varepsilon) \mapsto \hat{\gamma}_{U^{\MV}}(\ell,c,\varepsilon)$, $\theta = (\ell,c,\varepsilon) \in (0,1) \times \R_+ \times \R_+$, is continuous.
\end{lemma}

\begin{proof}
Fix a sequence $(\theta^n)_{n\in\N\cup\{\infty\}} = ((\ell^n,c^n,\varepsilon^n))_{n\in\N\cup\{\infty\}}$ of elements in $(0,1) \times \R \times \R_+$ such that $\theta^n = (\ell^n,c^n,\varepsilon^n) \to (\ell^\infty,c^\infty,\varepsilon^\infty) = \theta^\infty$ as $n \to \infty$. Writing $\gamma^n := \hat{\gamma}_{U^{\MV}}(\theta^n)$ for each $n \in \N \cup \{\infty\}$, we need to show that $\gamma^n \to \gamma^\infty$ as $n \to \infty$. Applying~\eqref{eq:1OC2} with parameter $\theta^\infty$ yields
\[
\ccS_{U^{\MV}}(\mu^{\theta^\infty}(\gamma^\infty+ \delta),\sigma^{\theta^\infty}(\gamma^\infty + \delta)) > \frac{\sigma^{\theta^\infty}(\gamma^\infty + \delta)}{\ell^\infty(\gamma^\infty - c^\infty + \delta)}, \qquad \delta \in A_1,
\]
where $A_1$ is a dense subset of $(0,\infty)$. By continuity of $\ccS_{U^{\MV}}$, $\theta \mapsto \mu^{\theta}$ and $\theta \mapsto \sigma^{\theta}$, it follows that
\[
\ccS_{U^{\MV}}(\mu^{\theta^n}(\gamma^\infty+ \delta),\sigma^{\theta^n}(\gamma^\infty + \delta)) > \frac{\sigma^{\theta^n}(\gamma^\infty + \delta)}{\ell^n(\gamma^\infty - c^n + \delta)}, \qquad \delta \in A_1,
\]
which coincides with~\eqref{eq:1OC2} with parameter $\theta^n$, and so $\gamma^n < \gamma^\infty + \delta$ for such $n \in \N$ large enough. By choosing $\delta \in A_1$ arbitrarily small and noting the density of $A_1$, we have $\limsup_{n\to\infty} \gamma^n \leq \gamma^\infty$. If $\gamma^\infty = c^\infty + \lambda^{\min}/\ell^\infty$, this shows that $\lim_{n\to\infty} \gamma^n = \gamma^\infty$. If $\gamma^\infty > c^\infty + \lambda^{\min}/\ell^\infty$, we use the other inequality in~\eqref{eq:1OC2} to obtain
\[
\ccS_{U^{\MV}}(\mu^{\theta^\infty}(\gamma^\infty - \delta),\sigma^{\theta^\infty}(\gamma^\infty - \delta)) < \frac{\sigma^{\theta^\infty}(\gamma^\infty - \delta)}{\ell^\infty(\gamma^\infty - c^\infty - \delta)}, \qquad \delta \in A_2,
\]
where $A_2$ is a dense subset of $(0,\gamma^\infty-c^\infty)$. By continuity of $\ccS_{U^{\MV}}$, $\theta \mapsto \mu^{\theta}$ and $\theta \mapsto \sigma^{\theta}$, it follows that
\[
\ccS_{U^{\MV}}(\mu^{\theta^n}(\gamma^\infty - \delta),\sigma^{\theta^n}(\gamma^\infty - \delta)) < \frac{\sigma^{\theta^n}(\gamma^\infty - \delta)}{\ell^n(\gamma^\infty - c^n - \delta)}, \qquad \delta \in A_2,
\]
which coincides with~\eqref{eq:1OC2} with parameter $\theta^n$, and so $\gamma^n > \gamma^\infty - \delta$ for such $n \in \N$ large enough. By choosing $\delta \in A_2$ arbitrarily small and noting the density of $A_2$, we have $\liminf_{n\to\infty} \gamma^n \geq \gamma^\infty$, which completes the proof.
\end{proof}

\begin{proof}[Proof of Proposition~\ref{prop:haty}]
This follows by combining Lemmas~\ref{lem:1OC} and~\ref{lem:hatgammacty}.
\end{proof}

\begin{proposition}
\label{prop:special:prop:quad:eq.opp:proc:bar:z}
For $\bar \gamma> \esssup \overline{\Xi}$, define $\bar c(\bar \gamma) := \sum_{k=1}^K c_k(\bar \gamma)$. Then
\begin{equation}
\label{eq:rel:ell:bar:gamma:bar:c:bar:gamma}
\bar \gamma-\E^{\Pb}[\bar \Xi] = \big(\bar \gamma-\bar c(\bar \gamma)\big) \ell(\bar \gamma).
\end{equation}
\end{proposition}

\begin{proof}
\begin{sloppypar}
Fix  $\bar \gamma > \esssup\overline{\Xi}$ and write $\vt^{(0)}(\bar \gamma)$ as a shorthand for $\vt^{(0)}(1;S(\bar \gamma))$. Recall from Lemma~\ref{thm:Qequilibria} that $\bar Z(\bar \gamma) S(\bar \gamma)$ is a local $\Pb$-martingale, where $\bar Z(\bar \gamma)$ is a  strictly positive $\Pb$-martingale. Since $\vt^{(0)}(\bar \gamma) \in \TB(S(\bar \gamma))$, Proposition \ref{prop:LOP} yields that $\bar Z(\bar \gamma) (\vt^{(0)}(\bar \gamma) \sint S(\bar \gamma))$ is a true $\Pb$-martingale, and hence so is $\bar Z(\bar \gamma) (1-\vt^{(0)}(\bar \gamma) \sint S(\bar \gamma))$. As mentioned above \eqref{eq:def:mart:assoc:mv:proc:chap:iv}, $(\bar V_s(\bar \gamma) M^{(0)}_s(\bar \gamma))_{0 \leq s \leq T}$ is also a $\Pb$-martingale. Moreover, by \eqref{eq:def:mart:assoc:mv:proc:chap:iv} and as $L_T(\bar \gamma) = 1$, we have
\end{sloppypar}
$$\bar V_T(\bar \gamma) M^{(0)}_T(\bar \gamma) = \bar H(\bar \gamma) L_T(\bar \gamma) \Big(1- \big(\vt^{(0)}(\bar \gamma) \sint S(\bar \gamma)\big)_T\Big) = \bar Z_T(\bar \gamma)  \Big(1- \big(\vt^{(0)}(\bar \gamma) \sint S(\bar \gamma)\big)_T\Big).$$
Thus by taking expectations under $\Pb$, we obtain
$$\bar V_0(\bar \gamma) M^{(0)}_0(\bar \gamma) = \bar Z_0(\bar \gamma) \Big(1- \big(\vt^{(0)}(\bar \gamma) \sint S(\bar \gamma)\big)_0\Big) =  \bar Z_0(\bar \gamma).$$
Recall that $\bar V_0(\bar \gamma) = \bar \gamma - c(\bar \Xi;S(\bar \gamma))$ and $M^{(0)}_0(\bar \gamma) = L_0(\bar \gamma) = \ell(\bar \gamma)$ by \eqref{eq:def:mart:assoc:mv:proc:chap:iv}, whereas we have $\bar Z_0(\bar \gamma) = \E^{\Pb}[\bar H(\bar \gamma)] = \bar \gamma-\E^{\Pb}[\bar \Xi]$. Plugging in yields
$$ \bar \gamma - \E^{\Pb}[\bar \Xi] =  \Big(\bar \gamma-c\big(\bar \Xi;S(\bar \gamma)\big)\Big) \ell(\bar \gamma).$$
Since the linearity of exMVH (see Lemma \ref{lmm:mvh:prob:linearity})  yields
\begin{equation*}
c\big(\bar \Xi ;S(\bar \gamma)\big) = \sum_{k=1}^K c\big( \Xi^k ;S(\bar \gamma)\big) = \sum_{k=1}^K c_k(\bar \gamma) = \bar c(\bar \gamma),
\end{equation*}
we obtain \eqref{eq:rel:ell:bar:gamma:bar:c:bar:gamma}, which concludes the proof.
\end{proof}

\begin{proof}[Proof of Corollary~\ref{cor:convexcomb}]
Fix $\overline{\gamma} > \esssup \overline{\Xi}$ and set $\overline{c}(S(\overline{\gamma})) := \sum_{k=1}^{K} c_k(S(\overline{\gamma}))$. First, we have from Theorem~\ref{thm:MVsol} upon rearranging that $\ell(S(\overline{\gamma}))\Psi(\overline{\gamma})$ satisfies the bounds
\begin{equation}
\label{eq:Psibds}
\sum_{k=1}^{K} \lambda_k^{\min} + \ell(S(\overline{\gamma}))\overline{c}(S(\overline{\gamma})) \leq \ell(S(\overline{\gamma}))\Psi(\overline{\gamma}) \leq \sum_{k=1}^{K} \lambda_k^{\max} + \ell(S(\overline{\gamma}))\overline{c}(S(\overline{\gamma})).
\end{equation}
On the other hand, we have by Proposition~\ref{prop:special:prop:quad:eq.opp:proc:bar:z} the identity
\begin{equation}
\label{eq:ellconstid}
\overline{\gamma} - \E[\overline{\Xi}] = \big(\overline{\gamma} - \overline{c}(S(\overline{\gamma}))\big) \ell(S(\overline{\gamma})).
\end{equation}
Adding $(1-\ell(S(\overline{\gamma})))\overline{\gamma}$ to the inequality~\eqref{eq:Psibds}, we observe that the left and right-hand sides simplify via~\eqref{eq:ellconstid}, respectively, into
\begin{align*}
&\sum_{k=1}^{K} \lambda_k^{\min} + \overline{\gamma} - \ell(S(\overline{\gamma}))(\overline{\gamma} - \overline{c}(S(\overline{\gamma})) = \sum_{k=1}^{K} \lambda_k^{\min} + \E[\overline{\Xi}] = \overline{\gamma}_{\min}, \\
&\sum_{k=1}^{K} \lambda_k^{\max} + \overline{\gamma} - \ell(S(\overline{\gamma}))(\overline{\gamma} - \overline{c}(S(\overline{\gamma})) = \sum_{k=1}^{K} \lambda_k^{\max} + \E[\overline{\Xi}] = \overline{\gamma}_{\max},
\end{align*}
and hence the claim~\eqref{eq:convexcomb} follows.
\end{proof}

\begin{proof}[Proof of Lemma~\ref{lmm:distinct:proc:bar:gamma}]
By Lemma~\ref{thm:Qequilibria}, $(1,S(\bar \gamma))$ is the unique quadratic equilibrium with aggregate risk tolerance $\bar \gamma$ and  $\bar Z(\bar \gamma) S(\bar \gamma)$ is a local martingale. This shows the ``if'' statement. To prove the converse, suppose for a contradiction that $(1,S(\bar \gamma))$ is also a quadratic equilibrium with respect to some risk tolerances $\gamma_1, \ldots, \gamma_K$ such that $\sum_{k=1}^K \gamma_k =: \bar \gamma' \neq \bar \gamma$.  Then by Lemma~\ref{lmm:quad:eq:mark:zs:is:loc:mart}, the process $\bar Z(\bar \gamma') S(\bar \gamma)$  is also a local martingale. Taking differences yields that
$$\big(\bar Z(\bar \gamma')-\bar Z(\bar \gamma) \big) S(\bar \gamma) = (\bar \gamma'-\bar \gamma) S(\bar \gamma)$$
is a local martingale as well, and so is $S(\bar \gamma)$ because $\bar \gamma' \neq \bar \gamma$. Thus the implication (a) $\Rightarrow$ (d) in  \cite[Lemma 2.27]{martins:23} leads to a contradiction of Assumption~\ref{ass:orthogonal}, so that $\bar \gamma' \neq \bar \gamma$ cannot hold. This concludes the proof of the equivalence. The second statement follows from the first, since if $S(\bar \gamma) = S(\bar \gamma')$, then $(1,S(\bar \gamma))$ is the unique quadratic equilibrium with aggregate risk tolerance $\bar \gamma'$ by Lemma~\ref{thm:Qequilibria} so that $\bar \gamma = \bar \gamma'$.
\end{proof}

\begin{proof}[Proof of Lemma~\ref{lem:Qequilibrium2}]
Fix $\overline{\gamma} \in [\overline{\gamma}_{\min},\overline{\gamma}_{\max}]$. The decomposition~\eqref{eq:xidecomp1} for $S^{(1)}(\overline{\gamma})$, where $A^{(1)}(\overline{\gamma})$ is given in~\eqref{eq:FVeq}, follows directly by plugging the identity $Z^H(\overline{\gamma}) = \overline{\gamma} - \overline{\gamma}_{\min} + \hat{Z}^H$ into~\eqref{eq:financialeq}. Likewise, plugging $H(\overline{\gamma}) = \overline{\gamma} - \overline{\gamma}_{\min} + H(\overline{\gamma}_{\min})$ and the aforementioned identity into~\eqref{eq:productiveeq} yields
\[
S_t^{(2)}(\overline{\gamma}) = \frac{(\overline{\gamma} - \overline{\gamma}_{\min})Z_t^D + \hat{Z}_t^{HD}}{\overline{\gamma} - \overline{\gamma}_{\min} + \hat{Z}_t^H}, \qquad t \in [0,T].
\]
Setting $t = 0$ gives $S_0^{(2)}(\overline{\gamma})$ in~\eqref{eq:FVeq}. By It\^{o}'s formula, the semimartingale dynamics for $S_t^{(2)}(\overline{\gamma})$ for each $t \in [0,T]$ are given by
\begin{multline*}
\dd S_t^{(2)}(\overline{\gamma}) = \frac{(\overline{\gamma} - \overline{\gamma}_{\min})\dd Z_t^D + \dd \hat{Z}_t^{HD}}{\overline{\gamma} - \overline{\gamma}_{\min} + \hat{Z}_t^H} - \frac{((\overline{\gamma} - \overline{\gamma}_{\min})Z_t^D + \hat{Z}_t^{HD})\dd \hat{Z}_t^H}{(\overline{\gamma} - \overline{\gamma}_{\min} + \hat{Z}_t^H)^2} \\
- \frac{\dd \langle (\overline{\gamma} - \overline{\gamma}_{\min})Z^D + \hat{Z}^{HD}, \hat{Z}^H \rangle_t}{(\overline{\gamma} - \overline{\gamma}_{\min} + \hat{Z}_t^H)^2} + \frac{((\overline{\gamma} - \overline{\gamma}_{\min})Z_t^D + \hat{Z}_t^{HD})\dd \langle \hat{Z}^H \rangle_t}{(\overline{\gamma} - \overline{\gamma}_{\min} + \hat{Z}_t^H)^3}.
\end{multline*}
By separating the local martingale and finite variation parts, we obtain~\eqref{eq:LMeq}. Finally, given the expression for the local martingale part $M^{(2)}(\overline{\gamma})$, the finite variation part $A_t^{(2)}(\overline{\gamma})$ can be simplified further for each $t \in [0,T]$ as
\[
\begin{aligned}
\dd A_t^{(2)}(\overline{\gamma})
&= - \frac{\dd \langle (\overline{\gamma} - \overline{\gamma}_{\min})Z^D + \hat{Z}^{HD}, \hat{Z}^H \rangle_t}{(\overline{\gamma} - \overline{\gamma}_{\min} + \hat{Z}_t^H)^2} \\
&\quad + \frac{((\overline{\gamma} - \overline{\gamma}_{\min})Z_t^D + \hat{Z}_t^{HD})\dd \langle \hat{Z}^H \rangle_t}{(\overline{\gamma} - \overline{\gamma}_{\min} + \hat{Z}_t^H)^3} \\
&= - \frac{\dd \langle M^{(2)}(\overline{\gamma}),\hat{Z}^H \rangle_t}{\overline{\gamma} - \overline{\gamma}_{\min} + \hat{Z}_t^H},
\end{aligned}
\]
which agrees with the expression given in~\eqref{eq:FVeq}. This concludes the decomposition~\eqref{eq:xidecomp1} for $S^{(2)}(\overline{\gamma})$ as well.
\end{proof}

\begin{lemma}
\label{lem:stochlogbmo}
For each $\overline{\gamma} \in [\overline{\gamma}_{\min},\overline{\gamma}_{\max}]$, the density process $Z^H(\overline{\gamma})$ satisfies $R_2(\Pb)$ with some constant $C^* > 0$ independent of $\overline{\gamma}$. Moreover, the family of stochastic logarithms $L^H(\overline{\gamma}) = (L_t^H(\overline{\gamma}))_{t\in[0,T]}$ defined by
\begin{equation}
\label{eq:stochlog}
L^H(\overline{\gamma}) := \frac{1}{Z^H(\overline{\gamma})} \bullet Z^H(\overline{\gamma}), \qquad \ \overline{\gamma} \in [\overline{\gamma}_{\min},\overline{\gamma}_{\max}],
\end{equation}
is dominated in $BMO$ by $L^H(\overline{\gamma}_{\min})$.
\end{lemma}

\begin{proof}
Set $\underline{\gamma} := \esssup \overline{\Xi}$ which lies in $[0,\infty)$ due to~\eqref{eq:esssup}, and so $0 \leq Z^H(\underline{\gamma}) \leq \underline{\gamma}$. Then, due to the relationship $Z^H(\overline{\gamma}) = \overline{\gamma} - \underline{\gamma} + Z^H(\underline{\gamma})$, we get $\overline{\gamma} - \underline{\gamma} \leq Z^H(\overline{\gamma}) \leq \overline{\gamma}$, and hence
\[
\E\Big[\Big(\frac{Z_T^H(\overline{\gamma})}{Z_t^H(\overline{\gamma})}\Big)^2 \,\Big|\, \cF_t\Big] \leq \frac{\overline{\gamma}^2}{(\overline{\gamma} - \underline{\gamma})^2} < \infty,  \qquad t \in [0,T],
\]
meaning that $Z^H(\overline{\gamma})$ satisfies $R_2(\Pb)$. Since the function $\overline{\gamma} \mapsto \frac{\overline{\gamma}^2}{(\overline{\gamma} - \underline{\gamma})^2}$ is decreasing on $(\underline{\gamma},\overline{\gamma}_{\max}]$, we can set $C^* := \frac{\overline{\gamma}_{\min}^2}{(\overline{\gamma}_{\min} - \underline{\gamma})^2}$, meaning that for every $\overline{\gamma} \in [\overline{\gamma}_{\min},\overline{\gamma}_{\max}]$, $Z^H(\overline{\gamma})$ satisfies $R_2(\Pb)$ with respect to $C^* \geq 1$ which is independent of $\overline{\gamma}$. Thus, by Delbaen et al.~\cite[Lemma 4.2]{delbaen:al:97}, the stochastic logarithm $L^H(\overline{\gamma})$ belongs to $BMO$ for any $\overline{\gamma} \in [\overline{\gamma}_{\min},\overline{\gamma}_{\max}]$. We express $L^H(\overline{\gamma})$ further as
\begin{equation}
\label{eq:stochlogdyn}
L^H(\overline{\gamma}) = \frac{1}{Z^H(\overline{\gamma})} \bullet Z^H(\overline{\gamma}) = \frac{1}{\overline{\gamma} - \underline{\gamma} + Z^H(\underline{\gamma})} \bullet (\overline{\gamma} - \underline{\gamma} + Z^H(\underline{\gamma})) =  \frac{1}{\overline{\gamma} - \underline{\gamma} + Z^H(\underline{\gamma})} \bullet Z^H(\underline{\gamma}).
\end{equation}
Thus, for $0 \leq s \leq t \leq T$, we have
\begin{multline}
\langle L^H(\overline{\gamma}) \rangle_t - \langle L^H(\overline{\gamma}) \rangle_s = \int_s^t \frac{1}{(\overline{\gamma} - \underline{\gamma} + Z^H(\underline{\gamma}))^2} \,\dd \langle Z^H(\underline{\gamma}) \rangle_u \\
\label{eq:incrementbd1}
\leq \int_s^t \frac{1}{(\overline{\gamma}_{\min} - \underline{\gamma} + Z^H(\underline{\gamma}))^2} \,\dd \langle Z^H(\underline{\gamma}) \rangle_u = \langle L^H(\overline{\gamma}_{\min}) \rangle_t - \langle L^H(\overline{\gamma}_{\min}) \rangle_s,
\end{multline}
which tells us that the increments of $\langle L^H(\overline{\gamma}) \rangle$ are dominated by those of $\langle L^H(\overline{\gamma}_{\min}) \rangle$, where the choice of $\overline{\gamma} \in [\overline{\gamma}_{\min},\overline{\gamma}_{\max}]$ is arbitrary.
\end{proof}

\begin{lemma}
\label{lem:Thetaclosed}
Let $S = (S_t)_{t\in[0,T]}$ be a continuous semimartingale with canonical decomposition $S = S_0 + M + A$, where $M = (M_t)_{t\in[0,T]}$ is a local martingale and $A = (A_t)_{t\in[0,T]}$ is a finite variation process. Define the set $\Theta(S) := L^2(M) \cap L^2(A)$. If there exists an ELMM $\Q$ for $S$ that satisfies $R_2(\Pb)$, then $\Theta(S) = \overline{\Theta}(S)$, and the set of attainable gains $\cG_T(\Theta(S)) := \{\vartheta \bullet S_T: \vartheta \in \Theta(S)\}$ is closed in $L^2$.
\end{lemma}

\begin{proof}
Because $\Q$ satisfies $R_2(\Pb)$, the implication (2) $\Rightarrow$ (1) in Delbaen et al.~\cite[Theorem 4.1]{delbaen:al:97} gives that $\cG_T(\Theta(S))$ is closed in $L^2$. Thus, by \v{C}ern\'{y}/Kallsen~\cite[Corollary 2.9]{cerny:kallsen:07}, $\cG_T(\Theta(S)) = \cG_T(\overline{\Theta}(S))$. By Proposition~\ref{prop:LOP}, the terminal gain $\vartheta \bullet S_T$ of a strategy $\vartheta \in \overline{\Theta}(S)$ uniquely determines $\vartheta$ up to $S$-equivalence. Thus, the equality $\cG_T(\Theta(S)) = \cG_T(\overline{\Theta}(S))$ implies that $\Theta(S) = \overline{\Theta}(S)$.
\end{proof}

\begin{proof}[Proof of Lemma~\ref{lem:xilambdadecomp}] Fix $\overline{\gamma} \in [\overline{\gamma}_{\min},\overline{\gamma}_{\max}]$.

(i) We construct $\xi(\overline{\gamma}) = (\xi_t(\overline{\gamma}))_{t\in[0,T]} \in L_{\loc}^2(N; \R^{d_3 \times d})$ explicitly by setting for every $t \in [0,T]$ each entry $\xi_t^{ij}(\overline{\gamma})$ to be
\begin{equation}
\label{eq:xidef}
\xi_t^{ij}(\overline{\gamma}) =
\begin{cases}
1, & 1 \leq i = j \leq d_1, \\
(\overline{\gamma} - \overline{\gamma}_{\min})/Z_t^H(\overline{\gamma}), & d_1 + 1 \leq i = j \leq d, \\
1 / Z_t^H(\overline{\gamma}), & d + 1 \leq i \leq d + d_2, j = i - d_2, \\
- ((\overline{\gamma} - \overline{\gamma}_{\min})Z_t^{D,j-d_1} + \hat{Z}_t^{HD,j-d_1})/ Z_t^H(\overline{\gamma})^2, & i = d_3, d_1 + 1 \leq j \leq d, \\
0, & \text{otherwise}.
\end{cases}
\end{equation}
Then, $M(\overline{\gamma}) = \xi(\overline{\gamma}) \bullet N$ follows by~\eqref{eq:LMeq} and the definition of $N = (M^{(1)}, Z^D, \hat{Z}^{HD}, \hat{Z}^H)$; uniqueness up to $N$-equivalence is immediately implied. Now, because $\hat{Z}^H$ is strictly positive and continuous, we have $\Pb[E] = 1$, where
\[
E := \{\hat{Z}_t^H > 0 \,\, \text{for all} \,\, t \in [0,T]\}.
\]
On the set $E$, we have that $(Z_t^H(\overline{\gamma}))^{-1} = (\overline{\gamma} - \overline{\gamma}_{\min} + \hat{Z}_t^H)^{-1}$ is well defined for all $t \in [0,T]$ and $\overline{\gamma} \in [\overline{\gamma}_{\min},\overline{\gamma}_{\max}]$. Hence on $E$, the first $d$ rows of $\xi_t(\overline{\gamma})$ are linearly independent for all $t \in [0,T]$ and $\overline{\gamma} \in [\overline{\gamma}_{\min},\overline{\gamma}_{\max}]$, since the submatrix formed by these columns has nonzero entries along the diagonal and null entries elsewhere. Thus,~\eqref{eq:xifullrank} holds.

(ii) Since $\Q(\overline{\gamma})$ has the density process $Z^H(\overline{\gamma})$ and $S(\overline{\gamma})$ is continuous, by Lemma~\ref{thm:Qequilibria} and Lemma~\ref{lem:stochlogbmo} yield that $\Q(\overline{\gamma})$ is an ELMM for $S(\overline{\gamma})$ that satisfies $R_2(\Pb)$. Hence, by Lemma~\ref{lem:Thetaclosed}, $\cG_T(\Theta(S(\overline{\gamma})))$ is closed in $L^2$. In particular, since $\Q(\overline{\gamma})$ is an ELMM for $S(\overline{\gamma})$ with square-integrable density, we obtain from Delbaen et al.~\cite[Theorem 3.7]{delbaen:al:97} that $S(\overline{\gamma})$ satisfies the so-called inequality $D_2(\Pb)$. Thus, we obtain from~\cite[Lemma 3.1]{delbaen:al:97} a process $\lambda(\overline{\gamma}) \in L^2(M(\overline{\gamma}); \R^d)$ such that the finite variation part $A(\overline{\gamma})$ of $S(\overline{\gamma})$ can be expressed as
\[
A(\overline{\gamma}) = \int_0^{\cdot} \dd \langle M(\overline{\gamma}) \rangle_u \lambda_u(\overline{\gamma}),
\]
and so, the SC decomposition~\eqref{eq:lambdadecomp1} for $S(\overline{\gamma})$ is implied. In fact, we can construct $\lambda(\overline{\gamma})$ explicitly by setting for each $t \in [0,T]$,
\begin{equation}
\label{eq:lambdadef}
\lambda_t(\overline{\gamma}) = - \frac{\xi_t^\dagger(\overline{\gamma})e_t^{d_3}}{\overline{\gamma} - \overline{\gamma}_{\min} + \hat{Z}_t^H}, \qquad t \in [0,T],
\end{equation}
where $(e_t^{d_3})_{t\in[0,T]} \equiv (0, ..., 0,1)^\top \in \R^{d_3}$, and $\xi^\dagger(\overline{\gamma}) = (\xi_t^\dagger(\overline{\gamma}))_{t\in[0,T]}$ denotes the left-inverse of $\xi(\overline{\gamma})$, which is well-defined since $\xi(\overline{\gamma})$ has linearly independent columns by~\eqref{eq:xifullrank}. Indeed, we have
\[
\dd \langle M(\overline{\gamma}), \hat{Z}^H \rangle_t = \dd \langle M(\overline{\gamma}), N \rangle_t e_t^{d_3} = \dd \langle M(\overline{\gamma}) \rangle_t \xi_t^\dagger(\overline{\gamma}) e_t^{d_3}, \qquad t \in [0,T],
\]
and so, the SC decomposition~\eqref{eq:lambdadecomp1} for $S(\overline{\gamma})$ follows by recalling the dynamics~\eqref{eq:FVeq} of $A(\overline{\gamma})$. Then,~\eqref{eq:lambdadecomp2} follows by~\eqref{eq:lambdadecomp1} and part (i). Lastly, uniqueness of $\lambda(\overline{\gamma})$ up to $M(\overline{\gamma})$-equivalence is immediate by uniqueness of $\xi(\overline{\gamma})$ up to $N$-equivalence.
\end{proof}

\begin{lemma}
\label{lem:incrementbound}
\begin{sloppypar}
For $M = (M_t)_{t\in[0,T]}$ an $\R^d$-valued continuous local martingale and $\lambda \in L_{\loc}^2(M; \R^d)$, let $S^\lambda = (S_t^\lambda)_{t\in[0,T]}$ be a continous semimartingale of the form
\end{sloppypar}
\[
S_t^\lambda = S_0^\lambda + \int_0^t \,\dd \langle M \rangle_s \lambda_s + M_t, \qquad t \in [0,T].
\]
If there exists an ELMM $\Q^\lambda$ for $S^\lambda$ with continuous density process $Z^\lambda= \cE(L^\lambda)$ for some continuous local martingale $L^\lambda= (L_t^\lambda)_{t\in[0,T]}$, then the local martingales $M$ and $L^\lambda + \lambda \bullet M$ are strongly orthogonal, and it holds that
\begin{equation}
\label{eq:incrementbd2}
\langle \lambda \bullet M \rangle_t - \langle \lambda \bullet M \rangle_s \leq \langle L^\lambda \rangle_t - \langle L^\lambda \rangle_s, \qquad 0 \leq s \leq t \leq T.
\end{equation}
Consequently, if $\{L^{\lambda}: \lambda \in L_{\loc}^2(M; \R^d)\}$ is dominated in $BMO$, then the family of stochastic integrals $\{\lambda \bullet M: \lambda \in L_{\loc}^2(M; \R^d)\}$ is dominated in $BMO$.
\end{lemma}

\begin{proof}
Since $\Q^\lambda$ is an ELMM for $S^\lambda$, we obtain by Girsanov's theorem that
\[
\dd \langle M,L^\lambda \rangle_t = - \dd\langle M \rangle_t\lambda_t = -\dd\langle M, \lambda \bullet M \rangle_t, \qquad t \in [0,T].
\]
Rearranging yields that $\langle M, L^\lambda + \lambda \bullet M \rangle \equiv 0$, which shows the strong orthogonality. We thus have the orthogonal decomposition $L^\lambda = -\lambda \bullet M + (L^\lambda + \lambda \bullet M)$, which yields for all $0 \leq s \leq t \leq T$ that
\[
\langle L^\lambda \rangle_t - \langle L^\lambda \rangle_s = \langle \lambda \bullet M \rangle_t - \langle \lambda \bullet M \rangle_s + \langle L^\lambda + \lambda \bullet M \rangle_t - \langle L^\lambda + \lambda \bullet M \rangle_s,
\]
which implies~\eqref{eq:incrementbd2}. The final assertion follows immediately.
\end{proof}

\begin{corollary}
\label{cor:xilambdabmo}
For $\overline{\gamma} \in [\overline{\gamma}_{\min},\overline{\gamma}_{\max}]$, we have that $(\xi(\overline{\gamma}),\lambda(\overline{\gamma})) \in \cS(C^*)$ and the family $\{\lambda(\overline{\gamma}) \bullet(\xi(\overline{\gamma}) \bullet N): \overline{\gamma} \in [\overline{\gamma}_{\min},\overline{\gamma}_{\max}]\}$ is dominated in $BMO$.
\end{corollary}

\begin{proof}
For $\overline{\gamma} \in [\overline{\gamma}_{\min},\overline{\gamma}_{\max}]$, we recall that $(\xi(\overline{\gamma}),\lambda(\overline{\gamma})) \in \cS$, in that they are the SC decomposition coefficients for $S(\overline{\gamma})$ as in~\eqref{eq:lambdadecomp2}, which we restate to be
\[
S_t(\overline{\gamma}) = S_0(\overline{\gamma}) + \int_0^t (\xi_s^\top(\overline{\gamma})\dd \langle N \rangle_s \xi_s(\overline{\gamma}) 
\lambda_s(\overline{\gamma}) + \xi_s(\overline{\gamma}) \dd N_s), \qquad t \in [0,T],
\]
and that $Z^{H}(\overline{\gamma})$ is the underlying ELMM density for $S(\overline{\gamma})$. By Lemma~\ref{lem:stochlogbmo}, it directly follows that $(\xi(\overline{\gamma}),\lambda(\overline{\gamma})) \in \cS(C^*)$ and that the family $\{L^H(\overline{\gamma}): \overline{\gamma} \in [\overline{\gamma}_{\min},\overline{\gamma}_{\max}]\}$ is dominated in $BMO$. Therefore, by the last statement of Lemma~\ref{lem:incrementbound}, the family $\{\lambda(\overline{\gamma}) \bullet(\xi(\overline{\gamma}) \bullet N): \overline{\gamma} \in [\overline{\gamma}_{\min},\overline{\gamma}_{\max}]\}$ is dominated in $BMO$.
\end{proof}


\begin{lemma}
\label{lem:xilambdacts2}
Let $\overline{\gamma} \in [\overline{\gamma}_{\min},\overline{\gamma}_{\max}]$ and $(\xi(\overline{\gamma}),\lambda(\overline{\gamma})) \in \cS(C^*)$. Then, the following hold: 
\begin{enumerate}
\item
\begin{equation}
\label{eq:xicts}
\Pb\Big[\lim_{\overline{\gamma}' \to \overline{\gamma}} \xi_t(\overline{\gamma}') = \xi_t(\overline{\gamma}) \,\, \text{for all} \,\, t \in [0,T] \,\, \text{and} \,\, \overline{\gamma} \in [\overline{\gamma}_{\min},\overline{\gamma}_{\max}]\Big] = 1,
\end{equation}
\item 
\begin{equation}
\label{eq:lambdacts}
\Pb\Big[\lim_{\overline{\gamma}' \to \overline{\gamma}} \lambda_t(\overline{\gamma}') = \lambda_t(\overline{\gamma}) \,\, \text{for all} \,\, t \in [0,T] \,\, \text{and} \,\, \overline{\gamma} \in [\overline{\gamma}_{\min},\overline{\gamma}_{\max}]\Big] = 1,
\end{equation}
\item
\begin{equation}
\label{eq:lambdacts2}
\lim_{\overline{\gamma}' \to \overline{\gamma}} \langle \lambda(\overline{\gamma}') \bullet (\xi(\overline{\gamma}') \bullet N) - \lambda(\overline{\gamma}) \bullet (\xi(\overline{\gamma}) \bullet N) \rangle_T = 0 \,\, \Pb\text{-a.s.}
\end{equation}
\end{enumerate}
\end{lemma}

\begin{proof}
(i) Recall the explicit construction of $\xi(\overline{\gamma})$ by~\eqref{eq:xidef}. Because $\hat{Z}^H$ is nonnegative and continuous, we have $\Pb[E] = 1$, where
\[
E := \{\hat{Z}_t^H > 0 \,\, \text{for all} \,\, t \in [0,T]\}.
\]
On the set $E$, we have that $(Z_t^H(\overline{\gamma}))^{-1} = (\overline{\gamma} - \overline{\gamma}_{\min} + \hat{Z}_t^H)^{-1}$ is well-defined for all $t \in [0,T]$ and $\overline{\gamma} \in [\overline{\gamma}_{\min},\overline{\gamma}_{\max}]$. Hence, for each $t \in [0,T]$ and $\overline{\gamma} \in [\overline{\gamma}_{\min},\overline{\gamma}_{\max}]$, by the identity $Z_t^H(\overline{\gamma}) = \overline{\gamma} - \overline{\gamma}_{\min} + \hat{Z}_t^H$, it is clear that $\lim_{\overline{\gamma}' \to \overline{\gamma}} Z_t^H(\overline{\gamma}') = Z_t^H(\overline{\gamma})$, and so, it follows from the construction~\eqref{eq:xidef} of $\xi(\overline{\gamma})$ that $\lim_{\overline{\gamma}' \to \overline{\gamma}} \xi_t(\overline{\gamma}') = \xi_t(\overline{\gamma})$, which proves~\eqref{eq:xicts}.

(ii) On the same set $E$,~\eqref{eq:lambdacts} follows from part (i) and the construction~\eqref{eq:lambdadef} of $\lambda(\overline{\gamma})$ for every $\overline{\gamma} \in [\overline{\gamma}_{\min}, \overline{\gamma}_{\max}]$ when noting that the map $\overline{\gamma} \mapsto \xi_t^\dagger(\overline{\gamma})$, $\overline{\gamma} \in [\overline{\gamma}_{\min},\overline{\gamma}_{\max}]$, is $\Pb\text{-a.s.}$ continuous for all $t \in [0,T]$ by Stewart~\cite[Theorem 4.3]{stewart:69}.

(iii) Next, since $\hat{Z}^H$ is a component of $N$, and so $\langle \hat{Z}^H \rangle$ is absolutely continuous with respect to $I^N$, there exists a predictable process $\kappa^{\hat{Z}^H}$ such that $\dd \langle \hat{Z}^H \rangle = \kappa^{\hat{Z}^H} \,\dd I^N$. Indeed, $\kappa^{\hat{Z}^H}$ can be taken to be the $(d_3,d_3)$-entry of $\kappa^N$. Now, fix some $\overline{\gamma}_* \in (\underline{\gamma},\overline{\gamma}]$, where $\underline{\gamma} = \esssup \overline{\Xi}$, and recall from~\eqref{eq:stochlog} that $L^H(\overline{\gamma}_*) = \frac{1}{Z^H(\overline{\gamma}_*)} \bullet Z^H(\overline{\gamma}_*)$. Thus, by applying the inequalities~\eqref{eq:incrementbd1} and~\eqref{eq:incrementbd2}, we obtain for $\overline{\gamma}' \geq \overline{\gamma}_*$ and $0 \leq s \leq t \leq T$ that
\[
\langle (\xi(\overline{\gamma}')\lambda(\overline{\gamma}')) \bullet N \rangle_t - \langle (\xi(\overline{\gamma}')\lambda(\overline{\gamma}')) \bullet N \rangle_s \leq \langle L^H(\overline{\gamma}_*) \rangle_t - \langle L^H(\overline{\gamma}_*) \rangle_s.
\]
Differentiating with respect to $I^N$ and plugging in the dynamics~\eqref{eq:stochlogdyn} of $L^H(\overline{\gamma}_*)$ yields
\[
0 \leq \lambda^\top(\overline{\gamma}')\xi^\top(\overline{\gamma}')\kappa^N\xi(\overline{\gamma}')\lambda(\overline{\gamma}') \leq \frac{\kappa^{\hat{Z}^H}}{(\overline{\gamma}_* - \overline{\gamma}_{\min} + \hat{Z}^H)^2}, \quad \Pb \ptimes I^N\text{-a.e.}
\]
for $\overline{\gamma}' \geq \overline{\gamma}_*$. Since we have the domination
\[
\int_0^T \frac{\kappa_t^{\hat{Z}^H}}{(\overline{\gamma}_* - \overline{\gamma}_{\min} + \hat{Z}_t^{H})^2} \,\dd I_t^N \leq \frac{\langle \hat{Z}^H \rangle_T}{(\overline{\gamma}_* - \overline{\gamma}_{\min})^2} < \infty,
\]
noting the first assertion and part (i), it follows by the dominated convergence theorem that
\[
\int_0^T (\lambda_t^\top(\overline{\gamma}')\xi_t^\top(\overline{\gamma}') - \lambda_t^\top(\overline{\gamma})\xi_t^\top(\overline{\gamma}))\kappa_t^N(\lambda_t(\overline{\gamma}')\xi_t(\overline{\gamma}') - \lambda_t(\overline{\gamma})\xi_t(\overline{\gamma})) \,\dd I_t^N \to 0,
\]
$\Pb\text{-a.s.}$ as $\overline{\gamma}' \to \overline{\gamma}$, which gives~\eqref{eq:lambdacts2}.
\end{proof}

\begin{proof}[Proof of Lemma~\ref{lem:xilambdacts}]
Let $\overline{\gamma}_\infty \in [\overline{\gamma}_{\min},\overline{\gamma}_{\max}]$ and a sequence $(\overline{\gamma}_n)_{n \in \N} \subset [\overline{\gamma}_{\min},\overline{\gamma}_{\max}]$ such that $\overline{\gamma}_n \to \overline{\gamma}_\infty$ as $n \to \infty$.	
	
(i) By part (i) of Lemma~\ref{lem:xilambdacts2}, we have
\[
\Pb\Big[\lim_{n\to\infty}\xi_t(\overline{\gamma}_n) = \xi_t(\overline{\gamma}_\infty) \,\, \text{for all} \,\, t \in [0,T]\Big] = 1.
\]
Then, by applying twice the dominated convergence theorem with respect to $\int_0^T \cdot \,\dd I^N$ and $\E[\cdot]$, with majorant 1 in both cases, we obtain
\[
\lim_{n\to\infty} \E\Big[\frac{1}{I_T^N + 1} \int_0^T (1 \wedge \|\xi_t(\overline{\gamma}_n) - \xi_t(\overline{\gamma}_\infty)\|) \,\dd I_t^N\Big] = 0.
\]

(ii) By Corollary~\ref{cor:xilambdabmo}, the family $\{\lambda(\overline{\gamma}) \bullet (\xi(\overline{\gamma}) \bullet N): \overline{\gamma} \in [\overline{\gamma}_{\min},\overline{\gamma}_{\max}]\}$ is dominated in $BMO$. Hence, by~\cite[Corollary A.2]{constantinou:al:25b}, the sequence $\{\langle \lambda(\overline{\gamma}) \bullet (\xi(\overline{\gamma}) \bullet N) \rangle_T: \overline{\gamma} \in [\overline{\gamma}_{\min},\overline{\gamma}_{\max}]\}$ is uniformly integrable. Thus, the bound $\langle N^1 - N^2 \rangle \leq 2 \langle N^1 \rangle + 2 \langle N^2 \rangle$ for local martingales $N^1,N^2$, yields
\[
\{\langle \lambda(\overline{\gamma}_n) \bullet (\xi(\overline{\gamma}_n) \bullet N) - \lambda(\overline{\gamma}_\infty) \bullet (\xi(\overline{\gamma}_\infty)\bullet N) \rangle_T: n \in \N \}
\]
is also uniformly integrable. By part (ii) of Lemma~\ref{lem:xilambdacts2}, we have
\[
\lim_{n\to\infty} \langle \lambda(\overline{\gamma}_n) \bullet (\xi(\overline{\gamma}_n) \bullet N) - \lambda(\overline{\gamma}_\infty) \bullet (\xi(\overline{\gamma}_\infty) \bullet N) \rangle_T = 0 \,\, \Pb\text{-a.s.},
\]
and so, by the (generalised) dominated convergence theorem, we arrive at
\[
\lim_{n\to\infty} \E[\langle \lambda(\overline{\gamma}_n) \bullet (\xi(\overline{\gamma}_n) \bullet N) - \lambda(\overline{\gamma}_\infty) \bullet (\xi(\overline{\gamma}_\infty) \bullet N) \rangle_T] = 0,
\]
which completes the proof.
\end{proof}

\begin{proof}[Proof of Corollary~\ref{cor:psi1cts}]
By Lemma~\ref{lem:stochlogbmo}, we can see immediately that the range of $\psi_1$ is contained in $\cS(C^*)$ with constant $C^* := \frac{\overline{\gamma}_{\min}^2}{(\overline{\gamma}_{\min} - \underline{\gamma})^2}$, where $\underline{\gamma} := \esssup \overline{\Xi}$. Now, recalling that the space $\cS(C^*)$ is equipped with the metric $d$ defined in~\eqref{eq:metric}, continuity of $\psi_1$ follows when combining parts (i) and (ii) of Lemma~\ref{lem:xilambdacts2}.
\end{proof}

To prove Theorem~\ref{thm:stability2} and Corollary~\ref{cor:psi2cts}, the final main ingredient is the general stability result for quadratic BSDEs given in Constantinou et al.~\cite[Theorem 5.2]{constantinou:al:25b}, which we state below as Theorem~\ref{thm:stability}.

Fix $l \in \N$. We denote by $|\cdot|$ the Euclidean norm on $\R^l$, and for a symmetric positive semidefinite matrix $C \in \R^{l \times l}$, we denote by $|\cdot|_{C}$ the seminorm under the pseudometric induced by $C$, in that
\[
|z|_C = (z^\top C z)^{\frac{1}{2}}, \qquad z \in \R^l.
\]
Fix $M = (M_t)_{t\in[0,T]}$ to be an $\R^l$-valued continuous local martingale. We may assume that $\trr \kappa^M \equiv 1$ $\Pb\ptimes I^M\text{-a.e.}$, in which case we have $|\cdot|_{\kappa^M} \leq |\cdot|$ $\Pb \ptimes I^M\text{-a.e.}$ We define $\cL_{BMO}^2(I^M)$ the space of $\R$-valued, predictable processes $\xi = (\xi_t)_{t\in[0,T]}$ satisfying
\[
\|\xi\|_{\cL_{BMO}^2(I^M)}^2 := \Big\|\sup_{\tau \in \Tau_{0,T}} \E_\tau\Big[\int_\tau^T \xi_u^2 \,\dd I_u^M\Big]\Big\|_{\infty} < \infty.
\]

This is a natural generalisation of the $\cL_{BMO}^2$ space with respect to the identity process $I_t = t$, $t \in [0,T]$.

Given an $\R^l$-valued predictable process $h = (h_t)_{t\in[0,T]}$ such that $h \bullet M \in BMO$, we have $|h|_{\kappa^M} \in \cL_{BMO}^2(I^M)$ and the identity $\|h \bullet M \|_{BMO_2} = \| |h|_{\kappa^M} \|_{\cL_{BMO}^2(I^M)}$.

\begin{assumption}
\label{ass:growthcondunif}
We say that a family $(f^n)_{n\in\N\cup\{\infty\}}$ of predictable functions $f^n: \Omega \times [0,T] \times \R \times \R^l \to \R$ satisfies
\begin{enumerate}
\item \begin{sloppypar}
a \emph{uniform quadratic growth bound} in $z$ (with respect to $M$), if there exist constants $C_1 \geq 0$, $C_2 > 0$, independent of $n$, and families of predictable processes $(\chi^n)_{n\in\N\cup\{\infty\}} = ((\chi_t^n)_{t\in[0,T]})_{n\in\N\cup\{\infty\}}$, $(\phi^n)_{n\in\N\cup\{\infty\}} = ((\phi_t^n)_{t\in[0,T]})_{n\in\N\cup\{\infty\}}$, that are bounded in $\cL_{BMO}^2(I^M)$, such that $\Pb \ptimes I^M\text{-a.e.}$ for all $n \in \N \cup \{\infty\}$, $y \in \R$ and $z \in \R^l$, we have
\end{sloppypar}
\begin{equation}
\label{eq:quadgrowth}
|f^n(y,z)| \leq (\chi^n)^2 + (\phi^n)^2|y| + \frac{C_2}{2}|z|_{\kappa^M}^2;
\end{equation}
\item a \emph{uniform Lipschitz bound} in $y$ (with respect to $M$), if there exists a constant $L_1 \geq 0$ independent of $n$, such that $\Pb\ptimes I^M\text{-a.e.}$ for all $n\in\N\cup\{\infty\}$, $y_1,y_2\in \R$ and $z \in \R^l$, we have
\begin{equation}
\label{eq:localLipYn}
|f^n(y_1,z) - f^n(y_2,z)| \leq L_1|y_1 - y_2|;
\end{equation}
If $L_1 > 0$, there exists a constant $c_M$ so that $\dd I_t^M \leq c_M \,\dd t$ $
\Pb\text{-a.s.}$ for all $t \in [0,T]$.
\item \begin{sloppypar}
a \emph{uniform local Lipschitz bound} in $z$ (with respect to $M$), if there exists a constant $L_2 > 0$ independent of $n$, and a family of predictable processes $(\xi^n)_{n\in\N\cup\{\infty\}} = ((\xi_t^n)_{t\in[0,T]})_{n\in\N\cup\{\infty\}}$ that is bounded in $\cL_{BMO}^2(I^M)$, such that $\Pb \ptimes I^M\text{-a.e.}$ for all $n \in \N \cup \{\infty\}$, $y \in \R$ and $z_1,z_2 \in \R^l$, we have
\end{sloppypar}
\begin{equation}
\label{eq:localLipZn}
|f^n(y,z_1) - f^n(y,z_2)| \leq L_2|z_1 - z_2|_{\kappa^M}(|\xi^n| + |z_1|_{\kappa^M} + |z_2|_{\kappa^M});
\end{equation}
\end{enumerate}
\end{assumption}

\begin{theorem}
\label{thm:stability}
Let $(Y^n)_{n\in\N\cup\{\infty\}} = ((Y_t^n)_{t\in[0,T]})_{n\in\N\cup\{\infty\}}$ be a sequence of continuous semimartingales, where for each $n \in \N \cup \{\infty\}$, $(Y^n,\zeta^n,M^{\perp,n}) = ((Y_t^n,\zeta_t^n,M_t^{\perp,n}))_{t\in[0,T]}$ solves the BSDE
\begin{align}
\label{eq:bsdeseq}
\dd Y_t^n &= -f_t^n(Y_t^n,\zeta_t^n) \,\dd I_t^M - \frac{g_t}{2} \,\dd \langle M^{\perp,n} \rangle_t + \zeta_t^n \,\dd M_t + \,\dd M_t^{\perp,n}, & t \in [0,T], \\
\nonumber
Y_T^n &= G^n,
\end{align}
where $G^n \in L^\infty(\cF_T)$, $g = (g_t)_{t\in[0,T]} \in \cL_{\unif}^\infty$, $f^n: \Omega \times [0,T] \times \R \times \R^l \to \R$ is a predictable function, $Y^n \in \cL^\infty$, $\zeta^n \in L_{\loc}^2(M; \R^l)$, and $M^{\perp,n}$ is a continuous local martingale orthogonal to $M$. Suppose that the following conditions hold:
\begin{enumerate}[(a)]
\item $C_Y := \sup_{n\in\N\cup\{\infty\}} \|Y^n\|_\infty < \infty$,
\item The family $(f^n)_{n\in\N\cup\{\infty\}}$ satisfies Assumption~\ref{ass:growthcondunif}.
\item For any $(y,z) \in \R \times \R^l$, it holds that
\[
\lim_{n\to\infty} f^n(y,z) = f^\infty(y,z), \quad \text{in} \,\, \Pb \ptimes I^M\text{-measure},
\]
with $(\chi^n - \chi^\infty)^2 \bullet I_T^M \xrightarrow{L^1} 0$ as $n \to \infty$, where $(\chi^n)_{n\in\N\cup\{\infty\}}$ is the family of predictable processes from Assumption~\ref{ass:growthcondunif}(i).
\item $G^n \xrightarrow{\Pb} G^\infty$ as $n \to \infty$.
\end{enumerate}
Then, the families of martingales $(\zeta^n \bullet M)_{n\in\N\cup\{\infty\}}$ and $(M^{\perp,n})_{n\in\N\cup\{\infty\}}$ are bounded in $BMO$, and it holds as $n \to \infty$ that
\begin{align}
\label{eq:stability1}
&\sup_{t\in[0,T]} |Y_t^n - Y_t^\infty| \xrightarrow{\Pb} 0, \\
\label{eq:stability2}
&\langle (\zeta^n - \zeta^\infty) \bullet M \rangle_T + \langle M^{\perp,n} - M^{\perp,\infty} \rangle_T \xrightarrow{\Pb} 0.
\end{align}
\end{theorem}

Recall from Section~\ref{sec:fixedpoint} that we are dealing with $Y^{(2),n}$, $Y^{(1),n}(H)$ and $Y^{(0),n}(H)$, satisfying the BSDEs for each $t \in [0,T]$,
\begin{align*}
\dd Y_t^{(2),n} &= \frac{1}{Y_t^{(2),n}}(Y_t^{(2),n}\lambda_t^n + \nu_t^{(2),n})^\top\dd\langle N^n \rangle_t(Y_t^{(2),n}\lambda_t^n + \nu_t^{(2),n}) \\
&\quad + \nu_t^{(2),n} \,\dd N_t^n + \dd N_t^{\perp,(2),n}, \\
\dd Y_t^{(1),n}(H) &= \frac{1}{Y_t^{(2),n}}(Y_t^{(2),n}\lambda_t^n + \nu_t^{(2),n})^\top\dd\langle N^n \rangle_t(Y_t^{(1),n}(H)\lambda_t^n + \nu_t^{(1),n}) \\
&\quad + \nu_t^{(1),n} \,\dd N_t^n + \dd N_t^{\perp,(1),n}, \\
\dd Y_t^{(0),n}(H) &= \frac{1}{Y_t^{(2),n}}(Y_t^{(1),n}(H)\lambda_t^n + \nu_t^{(1),n})^\top\dd\langle N^n \rangle_t(Y_t^{(1),n}(H)\lambda_t^n + \nu_t^{(1),n}) \\
&\quad + \nu_t^{(0),n} \,\dd N_t^n + \dd N_t^{\perp,(0),n},
\end{align*}
with terminal conditions $Y_T^{(2),n} = 1$, $Y_T^{(1),n}(H) = H$ and $Y_T^{(0),n}(H) = H^2$, where $H$ is bounded above and below away from $0$, $N^n = \xi^n \bullet N$, $\nu^{(i),n} \in L_{\loc}^2(N^n)$ and $N^{\perp,(i),n}$ is a continuous local martingale strongly orthogonal to $N^n$, for $i \in \{0,1,2\}$.

At first glance, we observe that all the drivers involve the term $\frac{1}{Y^{(2),n}}$ which in principle could cause some issues. In particular, it is not obvious if condition (a) of Theorem~\ref{thm:stability} is satisfied, that is, if the semimartingale solutions $Y^{(2),n}$,  $Y^{(1),n}(H)$ and $Y^{(0),n}(H)$ are uniformly bounded. The answer to this issue lies in the following lemma, where the semimartingale solutions to these BSDEs are in fact nonnegative with $Y^{(2),n}$ and $Y^{(1),n}(H)$ being bounded away from 0.

\begin{lemma}
\label{lem:aprioriboundY}
For all $n \in \N \cup \{\infty\}$ and $t \in [0,T]$, we have the a priori bounds
\begin{align*}
0 < C^{-1} &\leq Y_t^{(2),n} \leq 1, \\
0 < C^{-1} \essinf H &\leq Y_t^{(1),n}(H) \leq \esssup H, \\
0 &\leq Y^{(0),n}(H) \leq \esssup H^2.
\end{align*}
\end{lemma}

\begin{proof}
Fix $n \in \N \cup \{\infty\}$, and denote by $\Q^n := \Q(\xi^n,\lambda^n)$ the ELMM for $S^n = S^{(\xi^n,\lambda^n)}$, which exists by definition of $\cS$. Since each $\Q^n$ satisfies $R_2(\Pb)$ with the fixed constant $C \geq 1$, Jeanblanc et al.~\cite[Lemmas 1.5 and 2.1]{jeanblanc:al:12} gives $\delta \leq Y^{(2),n} \leq 1$, where the explicit lower bound $\delta = C^{-1}$ is given at the end of the proof of~\cite[Lemma 2.1]{jeanblanc:al:12}. This shows the uniform bounds for $Y^{(2),n}$. Next, the remark after~\cite[Proposition 2.6]{jeanblanc:al:12} yields that
\[
Y_t^{(1),n}(H) = Y_t^{(2),n} \E^{\Q^{*,n}}[H \,|\,\cF_t], \qquad t \in [0,T],
\]
where $\Q^{n,*}$ is the so-called \emph{variance-optimal martingale measure} (VOMM) for $S^n$; for details, see Delbaen/Schachermayer~\cite{DS1997}. Since $S^n$ given by~\eqref{eq:SC} is continuous, the VOMM is equivalent to $\Pb$ (as opposed to generally being a signed measure) by~\cite[Theorem 1.3]{DS1997}. Since $H$ is bounded above and below away from $0$, the uniform bounds for $Y^{(1),n}(H)$ follow accordingly. Finally, by~\cite[Lemma 1.5]{jeanblanc:al:12} and noting again that $H$ is bounded above and below away from $0$, we obtain the uniform bounds for $Y^{(0),n}(H)$.
\end{proof}

In order to apply Theorem~\ref{thm:stability} on $Y^{(2),n}$ and $Y^{(1),n}(H)$, using the boundedness away from 0 from Lemma~\ref{lem:aprioriboundY}, we can look at
\begin{align*}
(\hat{Y}^{(2),n}, \hat{\nu}^{(2),n},\hat{N}^{\perp,(2),n}) &:= (\log Y^{(2),n}, \nu^{(2),n}/Y^{(2),n}, (1/Y^{(2),n})\bullet N^{\perp,(2),n}), \\
(\hat{Y}^{(1),n}(H),\hat{\nu}^{(1),n},\hat{N}^{\perp,(1),n}) &:= (\log Y^{(1),n}(H), \nu^{(1),n}/Y^{(1),n}(H), (1/Y^{(1),n}(H)\bullet N^{\perp,(1),n}),
\end{align*}
which satisfy the BSDEs for each $t \in [0,T]$,
\begin{align*}
\dd \hat{Y}_t^{(2),n} &= (\lambda_t^n + \hat{\nu}_t^{(2),n})^\top \dd \langle N^n \rangle_t (\lambda_t^n + \hat{\nu}_t^{(2),n}) + \hat{\nu}_t^{(2),n} \,\dd N_t^n + \dd \hat{N}_t^{\perp,(2),n} \\
&\qquad - \frac{1}{2}((\hat{\nu}_t^{(2),n})^\top \dd \langle N^n \rangle_t\hat{\nu}_t^{(2),n} + \dd \langle \hat{N}^{\perp,(2),n} \rangle_t), \\
\dd \hat{Y}_t^{(1),n}(H) &= (\lambda_t^n + \hat{\nu}_t^{(2),n})^\top \dd \langle N^n \rangle_t (\lambda_t^n + \hat{\nu}_t^{(1),n}) + \hat{\nu}_t^{(1),n} \,\dd N_t^n + \dd \hat{N}_t^{\perp,(1),n} \\
&\qquad - \frac{1}{2}((\hat{\nu}_t^{(1),n})^\top\dd \langle N^n \rangle_t \hat{\nu}_t^{(1),n} + \dd \langle \hat{N}^{\perp,(1),n} \rangle_t), 
\end{align*}
with terminal conditions $\hat{Y}_T^{(2),n} = 0$ and $\hat{Y}_T^{(1),n} = \log H$. This transformation indeed makes it simpler to show stability for $Y^{(2),n}$ and $Y^{(1),n}(H)$, as the drivers of $\hat{Y}^{(2),n}$ and $\hat{Y}^{(1),n}(H)$ are even independent $\hat{Y}^{(2),n}$ and $\hat{Y}^{(1),n}(H)$, respectively.

Before we proceed to the proof of Theorem~\ref{thm:stability2}, we note that the local martingales $N^n = \xi^n \bullet N$ depend on $n$. This causes difficulties in the orthogonality argument between $N^n$ and $N^{\perp,(0),n}$, $\hat{N}^{\perp,(i),n}$, $i \in \{1,2\}$, as it also depends on $n$ and is nonstandard for BSDE stability results. To tackle this issue, we have the Galtchouk-Kunita-Watanabe decomposition for $\nu^{(0),n} \bullet N^n + N^{\perp,(0),n}$ and $\hat{\nu}^{(i),n} \bullet N^n + \hat{N}^{\perp,(i),n}$, $i \in \{1,2\}$, in terms of $N$ given by
\begin{align}
\label{eq:GKWdecomp1}
\hat{\nu}^{(i),n} \bullet N^n + \hat{N}^{\perp,(i),n} &= \hat{\zeta}^{(i),n} \bullet N + \widetilde{N}^{\perp,(i),n}, \qquad i \in \{1,2\}, \\
\label{eq:GKWdecomp2}
\nu^{(0),n} \bullet N^n + N^{\perp,(0),n} &= \zeta^{(0),n} \bullet N + \widetilde{N}^{\perp,(0),n},
\end{align}
for some $\zeta^{(0),n}, \hat{\zeta}^{(i),n} \in L_{\loc}^2(N)$, $i \in \{1,2\}$, and local martingales $\widetilde{N}^{\perp,(i),n}$, $i \in \{0,1,2\}$, strongly orthogonal to $N$, and thus also to $N^n = \xi^n \bullet N$. 

In the next result, for $i \in \{1,2\}$, we show how to find $\xi^n\hat{\nu}^{(i),n}$ and $\xi^n \nu^{(0),n}$ explicitly via a predictable $\R^{d_3 \times d_3}$-valued process $\pi^n = (\pi_t^n)_{t\in[0,T]}$ that gives the orthogonal projection onto the random range of $\kappa^N\hat{\zeta}^{(i),n}$ and $\kappa^N\zeta^{(0),n}$ in $\R^{d_3}$ with respect to the pseudometric induced by $\kappa^N$. But, as this construction entails the orthogonal projection onto the random range of $(\kappa^N)^{\frac{1}{2}}\xi^n$, we use Assumption~\ref{ass:cNinv} where $\kappa^N$ has full rank $\Pb \ptimes I^N\text{-a.e.}$ for continuity in $n$ of the pseudoinverse $((\kappa^N)^{\frac{1}{2}}\xi^n)^\dagger$ to persist; see Stewart~\cite[Theorem 4.3]{stewart:69} for a sufficient condition on the continuity of the pseudoinverse for a sequence of rectangular matrices.

\begin{lemma}
\label{lem:projection} 
\begin{sloppypar}
Let Assumption~\ref{ass:cNinv} hold. There exists a family $(\pi^n)_{n\in\N\cup\{\infty\}}= ((\pi_t^n)_{t\in[0,T]})_{n\in\N\cup\{\infty\}}$ of $\R^{d_3 \times d_3}$-valued predictable processes such that each of the statements hold $\Pb \ptimes I^N\text{-a.e.}$:
\end{sloppypar}
\begin{enumerate}
\item $(\pi^n)^\top \kappa^N = \kappa^N \pi^n = (\pi^n)^\top \kappa^N \pi^n$ for $n \in \N \cup \{\infty\}$.
\item $y^\top (\pi^n)^\top \kappa^N \pi^n y \leq y^\top \kappa^N y$ for all $y \in \R^{d_3}$ and $n \in \N \cup \{\infty\}$.
\item $(\kappa^N)^{\frac{1}{2}}\xi^n \hat{\nu}^{(i),n} = (\kappa^N)^{\frac{1}{2}}\pi^n \hat{\zeta}^{(i),n}$ and $(\kappa^N)^{\frac{1}{2}}\xi^n \nu^{(0),n} = (\kappa^N)^{\frac{1}{2}}\pi^n \zeta^{(0),n}$ for $n \in \N \cup \{\infty\}$ and $i \in \{1,2\}$.
\item $\pi^n \to \pi^\infty$ as $n \to \infty$.
\end{enumerate}
\end{lemma}

\begin{proof} First, we note that we only use Assumption~\ref{ass:cNinv} to prove item (iv), as otherwise it is not necessary.
For a rectangular $m_1 \times m_2$ matrix $C$, the construction in Albert~\cite[Theorem 3.4]{albert:72} gives the Moore-Penrose inverse $C^\dagger$ of $C$ by
\[
C^\dagger = \lim_{\delta \searrow 0}(C^\top C + \delta\,\id_{m_2})^{-1} C^\top,
\]
where $\id_{m_2}$ denotes the $m_2 \times m_2$ identity matrix. Because the inversion map is $\R^{m_2\times m_1}$-Borel-measurable, it follows that $C \mapsto C^\dagger $ is measurable on the set of all $m_1 \times m_2$ matrices. For a square symmetric matrix $C$, it is easy to see that we have the identity,
\begin{equation}
\label{eq:moorepenroseprop}
C^\dagger CC = C = CCC^\dagger.
\end{equation}
Hence, for fixed $n 
\in \N \cup \{\infty\}$, since $\kappa^N$ is positive semidefinite, we can define the predictable process $((\kappa^N)^{\frac{1}{2}}\xi^n)^\dagger = (((\kappa_t^N)^{\frac{1}{2}}\xi_t^n)^\dagger)_{t\in[0,T]}$ as the Moore-Penrose inverse of $(\kappa^N)^{\frac{1}{2}}\xi^n$. Moreover, by~\cite[Corollary 3.5]{albert:72}, the orthogonal projection on the range of $(\kappa^N)^{\frac{1}{2}}\xi^n$ with respect to the Euclidean metric is given $\Pb \ptimes I^N\text{-a.e.}$ by the predictable process $\tilde{\pi}^n = (\tilde{\pi}_t^n)_{t\in[0,T]}$ defined as
\begin{equation}
\label{eq:predprojectiondef}
\tilde{\pi}_t^n := (\kappa_t^N)^{\frac{1}{2}}\xi_t^n((\kappa_t^N)^{\frac{1}{2}}\xi_t^n)^\dagger, \qquad t \in [0,T].
\end{equation}
Finally, we construct the predictable process $\pi^n = (\pi_t^n)_{t\in[0,T]}$ by
\begin{equation}
\label{eq:pidef}
\pi_t^n = ((\kappa_t^N)^{\frac{1}{2}})^\dagger\tilde{\pi}_t^n(\kappa_t^N)^{\frac{1}{2}} = ((\kappa_t^N)^{\frac{1}{2}})^\dagger (\kappa_t^N)^{\frac{1}{2}} \xi_t^n((\kappa_t^N)^{\frac{1}{2}}\xi_t^n)^\dagger(\kappa_t^N)^{\frac{1}{2}}, \qquad t \in [0,T].
\end{equation}
We now proceed to show that $\pi^n$ satisfies conditions (i)-(iv).

(i) Noting symmetry of $\kappa^N$ and $\tilde{\pi}^n$ and using~\eqref{eq:moorepenroseprop}, we have $\Pb\ptimes I^N\text{-a.e.}$,
\begin{align*}
(\pi^n)^\top \kappa^N &= (\kappa^N)^{\frac{1}{2}} \tilde{\pi}^n ((\kappa^N)^{\frac{1}{2}})^\dagger\kappa^N = (\kappa^N)^{\frac{1}{2}} \tilde{\pi}^n ((\kappa^N)^{\frac{1}{2}})^\dagger(\kappa^N)^{\frac{1}{2}}(\kappa^N)^{\frac{1}{2}} = (\kappa^N)^{\frac{1}{2}} \tilde{\pi}^n (\kappa^N)^{\frac{1}{2}} =: \tilde{\alpha}^n, \\
\kappa^N \pi^n &= \kappa^N ((\kappa^N)^{\frac{1}{2}})^\dagger \tilde{\pi}^n (\kappa^N)^{\frac{1}{2}} = (\kappa^N)^{\frac{1}{2}} (\kappa^N)^{\frac{1}{2}} ((\kappa^N)^{\frac{1}{2}})^{\dagger}\tilde{\pi}^n (\kappa^N)^{\frac{1}{2}} = (\kappa^N)^{\frac{1}{2}} \tilde{\pi}^n (\kappa^N)^{\frac{1}{2}} = \tilde{\alpha}^n,
\end{align*}
which together show the first equality. Next, as $\tilde{\pi}^n$ is a projection, it holds that $(\tilde{\pi}^n)^2 = \tilde{\pi}^n$ $\Pb \ptimes I^N\text{-a.e.}$, and using the above, we get $\Pb\ptimes I^N\text{-a.e.}$,
\begin{equation}
\label{eq:piprop}
(\pi^n)^\top\kappa^N \pi^n = (\kappa^N)^{\frac{1}{2}} \tilde{\pi}^n ((\kappa^N)^{\frac{1}{2}})^\dagger (\kappa^N)^{\frac{1}{2}} \tilde{\pi}^n (\kappa^N)^{\frac{1}{2}} = (\kappa^N)^{\frac{1}{2}} (\tilde{\pi}^n)^2 (\kappa^N)^{\frac{1}{2}} = (\kappa^N)^{\frac{1}{2}} \tilde{\pi}^n (\kappa^N)^{\frac{1}{2}} = \tilde{\alpha}^n,
\end{equation}
where the second equality follows by plugging in the definition~\eqref{eq:predprojectiondef} of $\tilde{\pi}^n$ on the right and observing the simplification once again through~\eqref{eq:moorepenroseprop}.

(ii) Since $\tilde{\pi}^n$ is by construction $\Pb\ptimes I^N\text{-a.e.}$ an orthogonal projection with respect to the Euclidean metric, we have $x^\top x \geq x^\top \tilde{\pi}^n x$ for all $x \in \R^{d_3}$ $\Pb \ptimes I^N\text{-a.e.}$ Given $y \in \R^{d_3}$, setting $x = (\kappa^N)^{\frac{1}{2}}y$ then yields $\Pb \ptimes I^N\text{-a.e.}$,
\[
y^\top (\pi^n)^\top \kappa^N \pi^n  y = y^\top \tilde{\alpha}^n y = ((\kappa^N)^{\frac{1}{2}}y)^\top \tilde{\pi}^n ((\kappa^N)^{\frac{1}{2}}y) = x^\top \tilde{\pi}^n x \leq x^\top x = y^\top \kappa^N y.
\]

(iii) Fix $i \in \{1,2\}$. Recalling~\eqref{eq:GKWdecomp1} and~\eqref{eq:GKWdecomp2}, where $\hat{N}^{\perp,(i),n}$, $\tilde{N}^{\perp,(i),n}$, $N^{\perp,(0),n}$, $\tilde{N}^{\perp,(0),n}$, are strongly orthogonal to $N^n = \xi^n \bullet N$, taking the quadratic covariation of~\eqref{eq:GKWdecomp1} and~\eqref{eq:GKWdecomp2} with $N^n$, yield $\Pb \ptimes I^N\text{-a.e.}$,
\[
(\xi^n \hat{\nu}^{(i),n})^\top \kappa^N \xi^n = (\hat{\zeta}^{(i),n})^\top \kappa^N\xi^n \quad \text{and} \quad (\xi^n \nu^{(0),n})^\top \kappa^N \xi^n = (\zeta^{(0),n})^\top \kappa^N\xi^n.
\]
By taking differences, it follows that the predictable processes $d^{(i),n} = (d_t^{(i),n})_{t\in[0,T]}$ and $d^{(0),n} = (d_t^{(0),n})_{t\in[0,T]}$ defined by
\[
d_t^{(i),n} := (\kappa_t^N)^{\frac{1}{2}}(\xi_t^n\hat{\nu}_t^{(i),n} - \hat{\zeta}_t^{(i),n}) \quad \text{and} \quad d_t^{(0),n} := (\kappa_t^N)^{\frac{1}{2}}(\xi_t^n\nu_t^{(0),n} - \zeta_t^{(0),n}), \qquad t \in [0,T],
\]
satisfy
\[
(d^{(i),n})^\top (\kappa^N)^{\frac{1}{2}} \xi^n = 0 = (d^{(0),n})^\top (\kappa^N)^{\frac{1}{2}} \xi^n, \quad \Pb \ptimes I^N\text{-a.e.}
\]
Since $\tilde{\pi}^n$ is the Euclidean projection on the range of $(\kappa^N)^{\frac{1}{2}}\xi^n$, we have $\tilde{\pi}^nd^{(i),n} = 0 = \tilde{\pi}^n d^{(0),n}$, i.e.,
\[
\tilde{\pi}^n(\kappa^N)^{\frac{1}{2}}\hat{\zeta}^{(i),n} = \tilde{\pi}^n(\kappa^N)^{\frac{1}{2}}\xi^n\hat{\nu}^{(i),n} \quad \text{and} \quad \tilde{\pi}^n(\kappa^N)^{\frac{1}{2}}\zeta^{(0),n} = \tilde{\pi}^n(\kappa^N)^{\frac{1}{2}}\xi^n\nu^{(0),n}, \quad \Pb \ptimes I^N\text{-a.e.}
\]
Since $(\kappa^N)^{\frac{1}{2}}\xi^n\hat{\nu}^{(i),n}$ and $(\kappa^N)^{\frac{1}{2}}\xi^n\nu^{(0),n}$ belong to the range of $(\kappa^N)^{\frac{1}{2}}\xi^n$, they are invariant under $\tilde{\pi}^n$, and hence we may omit $\tilde{\pi}^n$ from the right-hand side of both equalities above. Writing $\tilde{\pi}^n$ by~\eqref{eq:pidef} in terms of $\pi^n$ on the left-hand side above by arguing similarly as for~\eqref{eq:piprop} gives
\[
(\kappa^N)^{\frac{1}{2}}\pi^n\hat{\zeta}^{(i),n} = (\kappa^N)^{\frac{1}{2}}\xi^n\hat{\nu}^{(i),n} \quad \text{and} \quad (\kappa^N)^{\frac{1}{2}}\pi^n\zeta^{(0),n} = (\kappa^N)^{\frac{1}{2}}\xi^n\nu^{(0),n}, \quad \Pb \ptimes I^N\text{-a.e.}
\]

(iv) By assumption, $\xi^n \to \xi^\infty$ $\Pb \ptimes I^N\text{-a.e.}$ as $n \to \infty$. Moreover, by construction of $\cS$, each $\xi^n$ has $\Pb \ptimes I^N\text{-a.e.}$ (full) rank $d$, and since $\kappa^N$ is crucially invertible by Assumption~\ref{ass:cNinv}, it follows that $(\kappa^N)^{\frac{1}{2}}\xi^n$ has the same rank $d$ for all $n$. Hence, by Stewart~\cite[Theorem 4.3]{stewart:69}, we have
\[
\lim_{n\to\infty} ((\kappa^N)^{\frac{1}{2}}\xi^n)^\dagger = ((\kappa^N)^{\frac{1}{2}}\xi^\infty)^\dagger, \quad \Pb \ptimes I^N\text{-a.e.}
\]
Therefore, from the definitions~\eqref{eq:predprojectiondef} and~\eqref{eq:pidef} of $\tilde{\pi}^n$ and $\pi^n$, respectively, we get indeed that $\tilde{\pi}^n \to \tilde{\pi}^\infty$ and $\pi^n \to \pi^\infty$ $\Pb \ptimes I^N\text{-a.e.}$ as $n \to \infty$.
\end{proof}

Using the decompositions~\eqref{eq:GKWdecomp1},~\eqref{eq:GKWdecomp2}, and subsequently Lemma~\ref{lem:projection}(iii), we can express for each $n \in \N \cup \{\infty\}$ the dynamics for $\hat{Y}^{(2),n}$, $\hat{Y}^{(1),n}(H)$ and $Y^{(0),n}(H)$ in the form
\begin{align*}
\dd \hat{Y}_t^{(2),n} &= -f_t^{(2),n}(\hat{\zeta}_t^{(2),n}) \,\dd I_t^N - \frac{g_t}{2} \,\dd \langle \widetilde{N}^{\perp,(2),n} \rangle_t + \hat{\zeta}_t^{(2),n} \,\dd N_t + \,\dd \widetilde{N}_t^{\perp,(2),n}, \\
\dd \hat{Y}_t^{(1),n}(H) &= -f_t^{(1),n}(\hat{\zeta}_t^{(1),n}) \,\dd I_t^N - \frac{g_t}{2} \,\dd \langle \widetilde{N}^{\perp,(1),n} \rangle_t + \hat{\zeta}_t^{(1),n} \,\dd N_t + \,\dd \widetilde{N}_t^{\perp,(1),n}, \\
\dd {Y}_t^{(0),n}(H) &= -f_t^{(0),n} \,\dd I_t^N  + \zeta_t^{(0),n} \,\dd N_t + \,\dd \widetilde{N}_t^{\perp,(0),n},
\end{align*}
where $g \equiv 1$, and with $\widetilde{\lambda}^n \equiv \xi^n \lambda^n$, the corresponding drivers are explicitly given for $z \in \R^{d_3}$ by
\begin{align*}
f_t^{(2),n}(z) &= \frac{1}{2} |z|_{\kappa_t^N}^2 - |\widetilde{\lambda}_t^n + \pi_t^n z|_{\kappa_t^N}^2, \\
f_t^{(1),n}(z) &= \frac{1}{2}|z|_{\kappa_t^N}^2 - (\widetilde{\lambda}_t^n + \pi_t^n\hat{\zeta}_t^{(2),n})^\top\kappa_t^N(\widetilde{\lambda}_t^n + \pi_t^nz), \\
f_t^{(0),n} &= -\exp(2\hat{Y}_t^{(1),n}(H) - \hat{Y}_t^{(2),n})|\widetilde{\lambda}_t^n + \pi_t^n\hat{\zeta}_t^{(1),n}|_{\kappa_t^N}^2.
\end{align*}

\begin{proof}[Proof of Theorem~\ref{thm:stability2}]
We check, in order, that the families of the BSDEs for $\hat{Y}^{(2),n}$, $\hat{Y}^{(1),n}(H)$ and $Y^{(0),n}(H)$ satisfy all the conditions in Theorem~\ref{thm:stability}.
		
Condition (a) is satisfied for all the families of BSDEs due to the uniform bounds given in Lemma~\ref{lem:aprioriboundY}, and noting that $H$ is bounded above and below away from 0. Condition (d) holds trivially for all the families of BSDEs as $\hat{Y}_T^{(2),n} = 0$, $\hat{Y}_T^{(2),n}(H) = \log H$ and $Y_T^{(0),n} = H^2$ do not depend on $n \in \N \cup \{\infty\}$.
	
To this end, we check conditions (b) and (c), for the families of the BSDEs for $\hat{Y}^{(2),n}$, $\hat{Y}^{(1),n}(H)$ and $Y^{(0),n}(H)$, in this order. First, we derive the bounds~\eqref{eq:quadgrowth}-\eqref{eq:localLipZn} for $\hat{Y}^{(2),n}$ to establish condition (b); note that~\eqref{eq:localLipYn} is trivial for all BSDEs considered here. For all $z \in \R^{d_3}$, we have by Lemma~\ref{lem:projection}(ii) that
\[
|f^{(2),n}(z)| \leq |\widetilde{\lambda}^n + \pi^nz|_{\kappa^N}^2 + \frac{1}{2}|z|_{\kappa^N}^2 \leq 2|\widetilde{\lambda}^n|_{\kappa^N}^2 + \frac{5}{2}|z|_{\kappa^N}^2, \qquad \Pb \ptimes I^N\text{-a.e.},
\]
where $(|\widetilde{\lambda}^n|_{\kappa^N})_{n\in\N\cup\{\infty\}}$ is bounded in $\cL_{BMO}^2(I^N)$, which shows the uniform quadratic growth bound~\eqref{eq:quadgrowth}. Noting the Cauchy-Schwarz-type inequality
\begin{equation}
\label{eq:CSineq}
x_1^\top \kappa^N x_2 - x_2^\top \kappa^M x_1 = (x_1 - x_2)^\top \kappa^N (x_1 + x_2) \leq |x_1 - x_2|_{\kappa^N}|x_1 + x_2|_{\kappa^N}, \qquad x_1,x_2 \in \R^{d_3},
\end{equation}
together with Lemma~\ref{lem:projection}(ii), we obtain $\Pb\ptimes I^N\text{-a.e.}$ for all $z_1,z_2 \in \R^{d_3}$, 
\begin{align*}
|f^{(2),n}(z_1) - f^{(2),n}(z_2)| &\leq |\pi^n(z_1-z_2)|_{\kappa^N}|2\widetilde{\lambda}^n + \pi^n(z_1+z_2)|_{\kappa^N} + \frac{1}{2}|z_1 - z_2|_{\kappa^N}|z_1 + z_2|_{\kappa^N} \\
&\leq |z_1 - z_2|_{\kappa^N}(2|\widetilde{\lambda}^n|_{\kappa^N} + |z_1 + z_2|_{\kappa^N}) + \frac{1}{2}|z_1 - z_2|_{\kappa^N}|z_1 + z_2|_{\kappa^N} \\
&\leq 2|z_1 - z_2|_{\kappa^N}(|\widetilde{\lambda}^n|_{\kappa^N} + |z_1 + z_2|_{\kappa^N}),
\end{align*}
which shows~\eqref{eq:localLipZn}, and hence justifies condition (b) for $\hat{Y}^{(2),n}$. Now, note that $(\xi^n,\lambda^n) \xrightarrow{d} (\xi^\infty,\lambda^\infty)$ in $\cS(C)$ as $n \to \infty$ implies in particular that
\begin{equation}
\label{eq:lambdaconv}
|\widetilde{\lambda}^n|^2_{\kappa^N} \bullet I_T^N \xrightarrow{L^2} |\widetilde{\lambda}^\infty|^2_{\kappa^N} \bullet I_T^N, \quad \text{as} \,\, n \to \infty. 
\end{equation}
Since~\eqref{eq:lambdaconv} yields that $|\widetilde{\lambda}^n - \widetilde{\lambda}^\infty|^2_{\kappa^N} \bullet I_T^N \xrightarrow{L^1} 0$ as $n\to\infty$, this secures the second convergence in condition (c) for $\hat{Y}^{(2),n}$. This also implies in particular that $|\widetilde{\lambda}^n|^2_{\kappa^N} \bullet I_T^N \to |\widetilde{\lambda}^\infty|^2_{\kappa^N} \bullet I_T^N$ $\Pb \ptimes I^N\text{-a.e.}$ as $n \to\infty$, and together with Lemma~\ref{lem:projection}(iv), the first convergence in condition (c) for $\hat{Y}^{(2),n}$ follows. Therefore, by Theorem~\ref{thm:stability}, the stability results~\eqref{eq:stability1},~\eqref{eq:stability2} hold, and $(|\hat{\zeta}^{(2),n}|_{\kappa^N})_{n\in\N\cup\{\infty\}}$ is bounded in $\cL_{BMO}^2(I^N)$, and $(\widetilde{N}^{\perp,(2),n})_{n\in\N\cup\{\infty\}}$ is bounded in $BMO$. Consequently, the stability result~\eqref{eq:stability1} also holds for $Y^{(2),n} = \exp(\hat{Y}^{(2),n})$, and in particular for $t = 0$, which gives~\eqref{eq:stability} for $i = 2$. Next, we derive the bounds~\eqref{eq:quadgrowth}-\eqref{eq:localLipZn} for $\hat{Y}^{(1),n}$ to establish condition (b). For all $z \in \R^{d_3}$, we have $\Pb\ptimes I^N\text{-a.e.}$ due to Lemma~\ref{lem:projection}(ii),
\begin{multline*}
|f^{(1),n}(z)| \leq |\widetilde{\lambda}^n + \pi^n\hat{\zeta}^{(2),n}|_{\kappa^N}|\widetilde{\lambda}^n + \pi^nz|_{\kappa^N} + \frac{1}{2}|z|_{\kappa^N}^2 \\ \leq \frac{1}{2}\Big(|\widetilde{\lambda}^n + \pi^n\hat{\zeta}^{(2),n}|_{\kappa^N}^2 + |\widetilde{\lambda}^n + \pi^nz|_{\kappa^N}^2\Big) + \frac{1}{2}|z|_{\kappa^N}^2 \leq 2|\widetilde{\lambda}^n|_{\kappa^N}^2 + |\hat{\zeta}^{(2),n}|_{\kappa^N}^2 + \frac{3}{2}|z|_{\kappa^N}^2,
\end{multline*}
where $(|\widetilde{\lambda}^n|_{\kappa^N})_{n\in\N\cup\{\infty\}}$ and $(|\hat{\zeta}^{(2),n}|_{\kappa^N})_{n\in\N\cup\{\infty\}}$ are bounded in $\cL_{BMO}^2(I^N)$, which shows the uniform quadratic growth bound~\eqref{eq:quadgrowth}. Using the Cauchy-Schwarz-type inequality~\eqref{eq:CSineq} and Lemma~\ref{lem:projection}(ii), we have $\Pb\ptimes I^N\text{-a.e.}$ for all $z_1,z_2 \in \R^{d_3}$, 
\begin{align*}
|f^{(1),n}(z_1) - f^{(1),n}(z_2)| &\leq (\widetilde{\lambda}^n + \pi^n\hat{\zeta}^{(2),n})^\top\kappa^N(\pi^n|z_1-z_2|) + \frac{1}{2}|z_1 - z_2|_{\kappa^N}|z_1 + z_2|_{\kappa^N} \\
&\leq |z_1 - z_2|_{\kappa^N}(|\widetilde{\lambda}^n|_{\kappa^N} + |\hat{\zeta}^{(2),n}|_{\kappa^N} + |z_1 + z_2|_{\kappa^N}),
\end{align*}
which shows the uniform quadratic growth bound~\eqref{eq:quadgrowth}, and hence justifies condition (b) for $\hat{Y}^{(1),n}(H)$. Now, note that the boundedness in $\cL_{BMO}^2(I^N)$ of $(|\hat{\zeta}^{(2),n}|_{\kappa^N})_{n\in\N\cup\{\infty\}}$ give by~\cite[Corollary A.2]{constantinou:al:25b} that the family $((|\hat{\zeta}^{(2),n}|_{\kappa^N} \bullet I_T^N)^p)_{n\in\N\cup\{\infty\}}$ is uniformly integrable for all $p \in [1,\infty)$. Together with the stability result~\eqref{eq:stability2} for $\hat{\zeta}^{(2),n}$ as $n \to \infty$, we get by Vitali's convergence theorem for each $p 
\in [1,\infty)$ that
\begin{equation}
\label{eq:zeta2conv}
|\hat{\zeta}^{(2),n}|^2_{\kappa^N} \bullet I_T^N \xrightarrow{L^p} |\hat{\zeta}^{(2),\infty}|^2_{\kappa^N} \bullet I_T^N, \quad \text{as} \,\, n \to \infty. 
\end{equation}
Since~\eqref{eq:zeta2conv} implies in particular that $|\hat{\zeta}^{(2),n} - \hat{\zeta}^{(2),\infty}|^2_{\kappa^N} \bullet I_T^N \xrightarrow{L^1} 0$ as $n \to \infty$, and recalling that $|\widetilde{\lambda}^n - \widetilde{\lambda}^{\infty}|^2_{\kappa^N} \bullet I_T^N \xrightarrow{L^1} 0$ as $n\to\infty$ as $n \to \infty$, this secures the second convergence in condition (c) for $\hat{Y}^{(1),n}$. Moreover,~\eqref{eq:zeta2conv} also implies that $|\hat{\zeta}^{(2),n}|^2_{\kappa^N} \bullet I_T^N \to |\hat{\zeta}^{(2),\infty}|^2_{\kappa^N} \bullet I_T^N$ $\Pb \ptimes I^N\text{-a.e.}$ as $n \to\infty$, and together with $|\widetilde{\lambda}^n|^2_{\kappa^N} \bullet I_T^N \to |\widetilde{\lambda}^\infty|^2_{\kappa^N} \bullet I_T^N$ $\Pb \ptimes I^N\text{-a.e.}$ as $n \to\infty$ and Lemma~\ref{lem:projection}(iv), the first convergence in condition (c) for $\hat{Y}^{(1),n}$ follows. Therefore, the stability results~\eqref{eq:stability1},~\eqref{eq:stability2} hold, and $(|\hat{\zeta}^{(1),n}|_{\kappa^N})_{n\in\N\cup\{\infty\}}$ is bounded in $\cL_{BMO}^2(I^N)$, and $(\widetilde{N}^{\perp,(1),n})_{n\in\N\cup\{\infty\}}$ is bounded in $BMO$. Consequently, the stability result~\eqref{eq:stability1} also holds for $Y^{(1),n}(H) = \exp(\hat{Y}^{(1),n}(H))$, and in particular for $t = 0$, which gives~\eqref{eq:stability} for $i = 1$. Lastly, due to Lemma~\ref{lem:aprioriboundY} and Lemma~\ref{lem:projection}(ii), we have
\[
|f_t^{(0),n}| \leq 2C \esssup H^2 (|\widetilde{\lambda}^n|_{\kappa^N}^2 + |\hat{\zeta}^{(1),n}|_{\kappa^N}^2), \qquad \Pb\ptimes I^N\text{-a.e.},
\]
where $(|\widetilde{\lambda}^n|_{\kappa^N})_{n\in\N\cup\{\infty\}}$ and $(|\hat{\zeta}^{(1),n}|_{\kappa^N})_{n\in\N\cup\{\infty\}}$ are bounded in $\cL_{BMO}^2(I^N)$, which shows the uniform quadratic growth bound~\eqref{eq:quadgrowth}; the uniform local Lipschitz bound~\eqref{eq:localLipZn} is trivially satisfied. Now, note that the boundedness in $\cL_{BMO}^2(I^N)$ of $(|\hat{\zeta}^{(1),n}|_{\kappa^N})_{n\in\N\cup\{\infty\}}$ give by~\cite[Corollary A.2]{constantinou:al:25b} that the family $((|\hat{\zeta}^{(1),n}|_{\kappa^N} \bullet I_T^N)^p)_{n\in\N\cup\{\infty\}}$ is uniformly integrable for all $p \in [1,\infty)$. Together with the stability result~\eqref{eq:stability2} for $\hat{\zeta}^{(1),n}$ as $n \to \infty$, we get by Vitali's convergence theorem for each $p 
\in [1,\infty)$ that
\begin{equation}
\label{eq:zeta1conv}
|\hat{\zeta}^{(1),n}|^2_{\kappa^N} \bullet I_T^N \xrightarrow{L^p} |\hat{\zeta}^{(1),\infty}|^2_{\kappa^N} \bullet I_T^N, \quad \text{as} \,\, n \to \infty. 
\end{equation}
Since~\eqref{eq:zeta1conv} implies in particular that $|\hat{\zeta}^{(1),n} - \hat{\zeta}^{(1),\infty}|^2_{\kappa^N} \bullet I_T^N \xrightarrow{L^1} 0$ as $n \to \infty$, and recalling that $|\widetilde{\lambda}^n - \widetilde{\lambda}^{\infty}|^2_{\kappa^N} \bullet I_T^N \xrightarrow{L^1} 0$ as $n \to \infty$, this secures the second convergence in condition (c) for $Y^{(0),n}$. Moreover,~\eqref{eq:zeta1conv} also implies that $|\hat{\zeta}^{(1),n}|^2_{\kappa^N} \bullet I_T^N \to |\hat{\zeta}^{(1),\infty}|^2_{\kappa^N} \bullet I_T^N$ $\Pb \ptimes I^N\text{-a.e.}$ as $n \to\infty$, and together with $|\widetilde{\lambda}^n|^2_{\kappa^N} \bullet I_T^N \to |\widetilde{\lambda}^\infty|^2_{\kappa^N} \bullet I_T^N$ $\Pb \ptimes I^N\text{-a.e.}$ as $n \to\infty$ and Lemma~\ref{lem:projection}(iv), the first convergence in condition (c) for $Y^{(0),n}$ follows. Therefore, the stability results~\eqref{eq:stability1},~\eqref{eq:stability2} hold, and $(|\zeta^{(0),n}|_{\kappa^N})_{n\in\N\cup\{\infty\}}$ is bounded in $\cL_{BMO}^2(I^N)$, and $(\widetilde{N}^{\perp,(0),n})_{n\in\N\cup\{\infty\}}$ is bounded in $BMO$. Thus, we also get~\eqref{eq:stability} for $i = 0$.
\end{proof}

\begin{proof}[Proof of Corollary~\ref{cor:MVeqsymmetric}]
By Theorem~\ref{thm:fixedpoint}, $(1,S(\overline{\gamma}^*))$ is a mean-variance equilibrium. In particular, we have $\overline{\gamma}^* = \sum_{k=1}^K \gamma_k^*$. Fix $k \in \{1,\ldots,K\}$.
By Corollary~\ref{cor:efficientfrontier}, we have $\mu_k^* = \mu_k(\gamma_k^*)$ and $\sigma_k^* = \sigma_k(\gamma_k^*)$, where the maps $\mu_k$ and $\sigma_k$ are defined under~\eqref{eq:SUbounds}. Hence, the first-order condition of Proposition~\ref{prop:haty} gives
\[
\gamma_k^* - \mu_k^* = \frac{\sigma_k^*}{\ccS_{U_k^{\MV}}(\mu_k^*,\sigma_k^*)},
\]
which proves the first identity.

Now, suppose that all agents are identical. In this case, market clearing implies in the allocation
\[
V_T^{k,*} = \frac{\overline{\Xi}}{K}, \quad k = 1, \ldots, K. 
\]
Consequently, we get $\mu_k^* = \E[\overline{\Xi}]/K$ and $\sigma_k^* = \sqrt{\Var(\overline{\Xi})}/K$. Moreover, the first part of the corollary implies that all $\gamma_k^*$ are the same. Since $\overline{\gamma}^* = \sum_{k=1}^K \gamma_k^*$, we obtain
\[
\overline{\gamma}^* = K \Big(\frac{\E[\overline{\Xi}]}{K} + \frac{\sqrt{\Var(\overline{\Xi})}/K}{\ccS_{U^{\MV}}(\E[\overline{\Xi}]/K,\sqrt{\Var(\overline{\Xi})}/K)}\Big) = \E[\overline{\Xi}] + \frac{\sqrt{\Var(\overline{\Xi})}}{\ccS_{U^{\MV}}(\E[\overline{\Xi}]/K,\sqrt{\Var(\overline{\Xi})}/K)},
\]
which proves the second assertion.
\end{proof}

\section*{Acknowledgements}
This work originated in David Martins's doctoral research at ETH Z\"{u}rich under the supervision of Martin Schweizer with co-examiners Christoph Czichowsky and Martin Larsson. Nikolaos Constantinou and Martin Herdegen are very grateful to their coauthors for their longstanding collaboration on this project.

\section*{Declarations}
\noindent\textbf{Competing interests.} The authors declare that they have no competing interests.

\noindent \textbf{Declaration of generative AI and AI-assisted technologies in the writing process.} ChatGPT (OpenAI) was used to enhance the clarity and readability of the manuscript and in the preparation of two examples.
All substantive content, including mathematics and interpretations, was authored and verified entirely by the authors, who assume full responsibility for the final work.

\clearpage
\addcontentsline{toc}{section}{References}
\typeout{CAPM-VERIFIED-REFERENCES-20261119-R4}

\typeout{CAPM-VERIFIED-REFERENCES-20261119-R4-END}


\begin{thebibliography}{10}

\bibitem{albert:72}
A.~Albert.
\newblock {\em {Regression and the Moore-Penrose Pseudoinverse}}, volume~94 of
  {\em Mathematics in Science and Engineering}.
\newblock Academic Press, New York, first edition, 1972.
\newblock \href {https://doi.org/10.1016/S0076-5392(08)X6167-3}
  {\path{doi:10.1016/S0076-5392(08)X6167-3}}.

\bibitem{allingham:91}
M.~Allingham.
\newblock {Existence Theorems in the Capital Asset Pricing Model}.
\newblock {\em Econometrica}, 59(4):1169--1174, 1991.
\newblock \href {https://doi.org/10.2307/2938180} {\path{doi:10.2307/2938180}}.

\bibitem{basak:chabakauri:10}
S.~Basak and G.~Chabakauri.
\newblock {Dynamic Mean-Variance Asset Allocation}.
\newblock {\em The Review of Financial Studies}, 23(8):2970--3016, 2010.
\newblock \href {https://doi.org/10.1093/rfs/hhq028}
  {\path{doi:10.1093/rfs/hhq028}}.

\bibitem{berk:97}
J.~B. Berk.
\newblock {Necessary Conditions for the CAPM}.
\newblock {\em Journal of Economic Theory}, 73(1):245--257, 1997.
\newblock \href {https://doi.org/10.1006/jeth.1996.2218}
  {\path{doi:10.1006/jeth.1996.2218}}.

\bibitem{bjorketal:17}
T.~Bj{\"o}rk, M.~Khapko, and A.~Murgoci.
\newblock {On time-inconsistent stochastic control in continuous time}.
\newblock {\em Finance and Stochastics}, 21(2):331--360, 2017.
\newblock \href {https://doi.org/10.1007/s00780-017-0327-5}
  {\path{doi:10.1007/s00780-017-0327-5}}.

\bibitem{bobrovnytska:schweizer:04}
O.~Bobrovnytska and M.~Schweizer.
\newblock {Mean-Variance Hedging and Stochastic Control: Beyond the Brownian
  Setting}.
\newblock {\em IEEE Transactions on Automatic Control}, 49(3):396--408, 2004.
\newblock \href {https://doi.org/10.1109/TAC.2004.824468}
  {\path{doi:10.1109/TAC.2004.824468}}.

\bibitem{breeden}
D.~T. Breeden.
\newblock {An intertemporal asset pricing model with stochastic consumption and
  investment opportunities}.
\newblock {\em Journal of Financial Economics}, 7(3):265--296, 1979.
\newblock \href {https://doi.org/10.1016/0304-405X(79)90016-3}
  {\path{doi:10.1016/0304-405X(79)90016-3}}.

\bibitem{briand:hu:08}
P.~Briand and Y.~Hu.
\newblock {Quadratic BSDEs with convex generators and unbounded terminal
  conditions}.
\newblock {\em Probability Theory and Related Fields}, 141(3--4):543--567,
  2008.
\newblock \href {https://doi.org/10.1007/s00440-007-0093-y}
  {\path{doi:10.1007/s00440-007-0093-y}}.

\bibitem{cerny:czichowsky:22}
A.~{\v C}ern{\'y} and C.~Czichowsky.
\newblock {The law of one price in quadratic hedging and mean--variance
  portfolio selection}.
\newblock {\em Finance and Stochastics}, 29(3):847--884, 2025.
\newblock \href {https://doi.org/10.1007/s00780-025-00563-7}
  {\path{doi:10.1007/s00780-025-00563-7}}.

\bibitem{cerny:kallsen:07}
A.~{\v C}ern{\'y} and J.~Kallsen.
\newblock {On the structure of general mean-variance hedging strategies}.
\newblock {\em The Annals of Probability}, 35(4):1479--1531, 2007.
\newblock \href {https://doi.org/10.1214/009117906000000872}
  {\path{doi:10.1214/009117906000000872}}.

\bibitem{cerny:kallsen:08}
A.~{\v C}ern{\'y} and J.~Kallsen.
\newblock {A Counterexample Concerning the Variance-Optimal Martingale
  Measure}.
\newblock {\em Mathematical Finance}, 18(2):305--316, 2008.
\newblock \href {https://doi.org/10.1111/j.1467-9965.2007.00334.x}
  {\path{doi:10.1111/j.1467-9965.2007.00334.x}}.

\bibitem{choullietal:1998}
T.~Choulli, L.~Krawczyk, and C.~Stricker.
\newblock {$\mathcal{E}$-martingales and their applications in mathematical
  finance}.
\newblock {\em The Annals of Probability}, 26(2):853--876, 1998.
\newblock \href {https://doi.org/10.1214/aop/1022855653}
  {\path{doi:10.1214/aop/1022855653}}.

\bibitem{constantinou:al:25b}
N.~Constantinou, C.~Czichowsky, M.~Herdegen, and D.~Martins.
\newblock {A stability result for quadratic BSDEs with BMO growth at the
  origin}.
\newblock {\em Stochastic Processes and their Applications}, 204:105127, 2027.
\newblock \href {https://doi.org/10.1016/j.spa.2026.105127}
  {\path{doi:10.1016/j.spa.2026.105127}}.

\bibitem{czichowsky13}
C.~Czichowsky.
\newblock {Time-consistent mean-variance portfolio selection in discrete and
  continuous time}.
\newblock {\em Finance and Stochastics}, 17(2):227--271, 2013.
\newblock \href {https://doi.org/10.1007/s00780-012-0189-9}
  {\path{doi:10.1007/s00780-012-0189-9}}.

\bibitem{czichowsky:schweizer:11}
C.~Czichowsky and M.~Schweizer.
\newblock {Closedness in the Semimartingale Topology for Spaces of Stochastic
  Integrals with Constrained Integrands}.
\newblock In C.~Donati-Martin, A.~Lejay, and A.~Rouault, editors, {\em
  {S{\'e}minaire de Probabilit{\'e}s XLIII}}, volume 2006 of {\em Lecture Notes
  in Mathematics}, pages 413--436. Springer, Berlin, Heidelberg, 2011.
\newblock \href {https://doi.org/10.1007/978-3-642-15217-7_18}
  {\path{doi:10.1007/978-3-642-15217-7_18}}.

\bibitem{czichowsky:schweizer:13}
C.~Czichowsky and M.~Schweizer.
\newblock {Cone-constrained continuous-time Markowitz problems}.
\newblock {\em The Annals of Applied Probability}, 23(2):764--810, 2013.
\newblock \href {https://doi.org/10.1214/12-AAP855}
  {\path{doi:10.1214/12-AAP855}}.

\bibitem{dana:99}
R.-A. Dana.
\newblock {Existence, uniqueness and determinacy of equilibrium in C.A.P.M.
  with a riskless asset}.
\newblock {\em Journal of Mathematical Economics}, 32(2):167--175, 1999.
\newblock \href {https://doi.org/10.1016/S0304-4068(98)00050-0}
  {\path{doi:10.1016/S0304-4068(98)00050-0}}.

\bibitem{delbaen:al:97}
F.~Delbaen, P.~Monat, W.~Schachermayer, M.~Schweizer, and C.~Stricker.
\newblock {Weighted norm inequalities and hedging in incomplete markets}.
\newblock {\em Finance and Stochastics}, 1(3):181--227, 1997.
\newblock \href {https://doi.org/10.1007/s007800050021}
  {\path{doi:10.1007/s007800050021}}.

\bibitem{delbaen:schachermayer:96}
F.~Delbaen and W.~Schachermayer.
\newblock {Attainable claims with $p$'th moments}.
\newblock {\em Annales de l'Institut Henri Poincar{\'e}, Probabilit{\'e}s et
  Statistiques}, 32(6):743--763, 1996.
\newblock URL: \href{https://numdam.org/item/AIHPB_1996__32_6_743_0/}{\path{https://numdam.org/item/AIHPB_1996__32_6_743_0/}}.

\bibitem{DS1997}
F.~Delbaen and W.~Schachermayer.
\newblock {The variance-optimal martingale measure for continuous processes}.
\newblock {\em Bernoulli}, 2(1):81--105, 1996.
\newblock Amendments and corrections in \emph{Bernoulli}, 2(4):379--380, 1996.
\newblock \href {https://doi.org/10.3150/bj/1193758791}
  {\path{doi:10.3150/bj/1193758791}}.

\bibitem{delbaen:schachermayer:98}
F.~Delbaen and W.~Schachermayer.
\newblock {A Simple Counterexample to Several Problems in the Theory of Asset
  Pricing}.
\newblock {\em Mathematical Finance}, 8(1):1--11, 1998.
\newblock \href {https://doi.org/10.1111/1467-9965.00041}
  {\path{doi:10.1111/1467-9965.00041}}.

\bibitem{eberlein:kallsen:19}
E.~Eberlein and J.~Kallsen.
\newblock {\em {Mathematical Finance}}.
\newblock Springer Finance. Springer, Cham, 2019.
\newblock \href {https://doi.org/10.1007/978-3-030-26106-1}
  {\path{doi:10.1007/978-3-030-26106-1}}.

\bibitem{escauriaza:al:21}
L.~Escauriaza, D.~C. Schwarz, and H.~Xing.
\newblock {Radner equilibrium and systems of quadratic BSDEs with discontinuous
  generators}.
\newblock {\em The Annals of Applied Probability}, 32(5):3492--3536, 2022.
\newblock \href {https://doi.org/10.1214/21-AAP1765}
  {\path{doi:10.1214/21-AAP1765}}.

\bibitem{frei:al:12}
C.~Frei, M.~Mocha, and N.~Westray.
\newblock {BSDEs in utility maximization with BMO market price of risk}.
\newblock {\em Stochastic Processes and their Applications}, 122(6):2486--2519,
  2012.
\newblock \href {https://doi.org/10.1016/j.spa.2012.03.007}
  {\path{doi:10.1016/j.spa.2012.03.007}}.

\bibitem{hens:al:02}
T.~Hens, J.~Laitenberger, and A.~L{\"o}ffler.
\newblock {Two remarks on the uniqueness of equilibria in the CAPM}.
\newblock {\em Journal of Mathematical Economics}, 37(2):123--132, 2002.
\newblock \href {https://doi.org/10.1016/S0304-4068(02)00012-5}
  {\path{doi:10.1016/S0304-4068(02)00012-5}}.

\bibitem{herdegen:al:21}
M.~Herdegen, J.~Muhle-Karbe, and D.~Possama{\"i}.
\newblock {Equilibrium asset pricing with transaction costs}.
\newblock {\em Finance and Stochastics}, 25(2):231--275, 2021.
\newblock \href {https://doi.org/10.1007/s00780-021-00449-4}
  {\path{doi:10.1007/s00780-021-00449-4}}.

\bibitem{jacod:shiryaev:03}
J.~Jacod and A.~N. Shiryaev.
\newblock {\em {Limit Theorems for Stochastic Processes}}, volume 288 of {\em
  Grundlehren der mathematischen Wissenschaften}.
\newblock Springer, Berlin, Heidelberg, second edition, 2003.
\newblock \href {https://doi.org/10.1007/978-3-662-05265-5}
  {\path{doi:10.1007/978-3-662-05265-5}}.

\bibitem{jeanblanc:al:12}
M.~Jeanblanc, M.~Mania, M.~Santacroce, and M.~Schweizer.
\newblock {Mean-variance hedging via stochastic control and BSDEs for general
  semimartingales}.
\newblock {\em The Annals of Applied Probability}, 22(6):2388--2428, 2012.
\newblock \href {https://doi.org/10.1214/11-AAP835}
  {\path{doi:10.1214/11-AAP835}}.

\bibitem{kardaras:al:22}
C.~Kardaras, H.~Xing, and G.~{\v Z}itkovi{\'c}.
\newblock {Incomplete Stochastic Equilibria with Exponential Utilities Close to
  Pareto Optimality}.
\newblock In G.~Yin and T.~Zariphopoulou, editors, {\em {Stochastic Analysis,
  Filtering, and Stochastic Optimization: A Commemorative Volume to Honor Mark
  H. A. Davis's Contributions}}, pages 267--292. Springer, Cham, 2022.
\newblock \href {https://doi.org/10.1007/978-3-030-98519-6_11}
  {\path{doi:10.1007/978-3-030-98519-6_11}}.

\bibitem{koch-medina:wenzelburger:18}
P.~Koch-Medina and J.~Wenzelburger.
\newblock {Equilibria in the CAPM with non-tradeable endowments}.
\newblock {\em Journal of Mathematical Economics}, 75:93--107, 2018.
\newblock \href {https://doi.org/10.1016/j.jmateco.2017.12.004}
  {\path{doi:10.1016/j.jmateco.2017.12.004}}.

\bibitem{levy:11}
H.~Levy.
\newblock {\em {The Capital Asset Pricing Model in the 21st Century:
  Analytical, Empirical, and Behavioral Perspectives}}.
\newblock Cambridge University Press, Cambridge, 2012.
\newblock \href {https://doi.org/10.1017/CBO9781139017459}
  {\path{doi:10.1017/CBO9781139017459}}.

\bibitem{lintner:65a}
J.~Lintner.
\newblock {Security Prices, Risk, and Maximal Gains From Diversification}.
\newblock {\em The Journal of Finance}, 20(4):587--615, 1965.
\newblock \href {https://doi.org/10.1111/j.1540-6261.1965.tb02930.x}
  {\path{doi:10.1111/j.1540-6261.1965.tb02930.x}}.

\bibitem{lintner:65b}
J.~Lintner.
\newblock {The Valuation of Risk Assets and the Selection of Risky Investments
  in Stock Portfolios and Capital Budgets}.
\newblock {\em The Review of Economics and Statistics}, 47(1):13--37, 1965.
\newblock \href {https://doi.org/10.2307/1924119} {\path{doi:10.2307/1924119}}.

\bibitem{markowitz:2010}
H.~M. Markowitz.
\newblock {Portfolio Theory: As I Still See It}.
\newblock {\em Annual Review of Financial Economics}, 2(1):1--23, 2010.
\newblock \href {https://doi.org/10.1146/annurev-financial-011110-134602}
  {\path{doi:10.1146/annurev-financial-011110-134602}}.

\bibitem{markowitz:2014}
H.~Markowitz.
\newblock {Mean--variance approximations to expected utility}.
\newblock {\em European Journal of Operational Research}, 234(2):346--355,
  2014.
\newblock \href {https://doi.org/10.1016/j.ejor.2012.08.023}
  {\path{doi:10.1016/j.ejor.2012.08.023}}.

\bibitem{martins:23}
D.~P.~T. Martins.
\newblock {\em {Aspects of quadratic utility: mean--variance hedging in rough
  volatility models, and CAPM-type equilibria}}.
\newblock Doctoral thesis no.\ 29897, ETH Z{\"u}rich, Z{\"u}rich, 2023.
\newblock \href {https://doi.org/10.3929/ethz-b-000661658}
  {\path{doi:10.3929/ethz-b-000661658}}.

\bibitem{mocha:westray:12}
M.~Mocha and N.~Westray.
\newblock {Quadratic Semimartingale BSDEs Under an Exponential Moments
  Condition}.
\newblock In C.~Donati-Martin, A.~Lejay, and A.~Rouault, editors, {\em
  {S{\'e}minaire de Probabilit{\'e}s XLIV}}, volume 2046 of {\em Lecture Notes
  in Mathematics}, pages 105--139. Springer, Berlin, Heidelberg, 2012.
\newblock \href {https://doi.org/10.1007/978-3-642-27461-9_5}
  {\path{doi:10.1007/978-3-642-27461-9_5}}.

\bibitem{mossin:66}
J.~Mossin.
\newblock {Equilibrium in a Capital Asset Market}.
\newblock {\em Econometrica}, 34(4):768--783, 1966.
\newblock \href {https://doi.org/10.2307/1910098} {\path{doi:10.2307/1910098}}.

\bibitem{nielsen:88}
L.~T. Nielsen.
\newblock {Uniqueness of Equilibrium in the Classical Capital Asset Pricing
  Model}.
\newblock {\em Journal of Financial and Quantitative Analysis}, 23(3):329--336,
  1988.
\newblock \href {https://doi.org/10.2307/2331071} {\path{doi:10.2307/2331071}}.

\bibitem{nielsen:90b}
L.~T. Nielsen.
\newblock {Equilibrium in CAPM Without a Riskless Asset}.
\newblock {\em The Review of Economic Studies}, 57(2):315--324, 1990.
\newblock \href {https://doi.org/10.2307/2297384} {\path{doi:10.2307/2297384}}.

\bibitem{nielsen:90a}
L.~T. Nielsen.
\newblock {Existence of equilibrium in CAPM}.
\newblock {\em Journal of Economic Theory}, 52(1):223--231, 1990.
\newblock \href {https://doi.org/10.1016/0022-0531(90)90076-V}
  {\path{doi:10.1016/0022-0531(90)90076-V}}.

\bibitem{radner:82}
R.~Radner.
\newblock {Equilibrium under uncertainty}.
\newblock In K.~J. Arrow and M.~D. Intriligator, editors, {\em {Handbook of
  Mathematical Economics}}, volume~2, chapter~20, pages 923--1006.
  North-Holland, Amsterdam, 1982.
\newblock \href {https://doi.org/10.1016/S1573-4382(82)02015-3}
  {\path{doi:10.1016/S1573-4382(82)02015-3}}.

\bibitem{Schweizer:10}
M.~Schweizer.
\newblock {Mean--Variance Hedging}.
\newblock In R.~Cont, editor, {\em {Encyclopedia of Quantitative Finance}},
  pages 1177--1181. Wiley, 2010.
\newblock \href {https://doi.org/10.1002/9780470061602.eqf14025}
  {\path{doi:10.1002/9780470061602.eqf14025}}.

\bibitem{sharpe:64}
W.~F. Sharpe.
\newblock {Capital Asset Prices: A Theory of Market Equilibrium under
  Conditions of Risk}.
\newblock {\em The Journal of Finance}, 19(3):425--442, 1964.
\newblock \href {https://doi.org/10.1111/j.1540-6261.1964.tb02865.x}
  {\path{doi:10.1111/j.1540-6261.1964.tb02865.x}}.

\bibitem{CS2002}
A.~N. Shiryaev and A.~S. Cherny.
\newblock {Vector Stochastic Integrals and the Fundamental Theorems of Asset
  Pricing}.
\newblock {\em Proceedings of the Steklov Institute of Mathematics}, 237:6--49,
  2002.
\newblock URL: \href{https://www.mathnet.ru/eng/tm323}{\path{https://www.mathnet.ru/eng/tm323}}.

\bibitem{stapleton}
R.~C. Stapleton and M.~G. Subrahmanyam.
\newblock {A Multiperiod Equilibrium Asset Pricing Model}.
\newblock {\em Econometrica}, 46(5):1077--1096, 1978.
\newblock \href {https://doi.org/10.2307/1911437} {\path{doi:10.2307/1911437}}.

\bibitem{stewart:69}
G.~W. Stewart.
\newblock {On the Continuity of the Generalized Inverse}.
\newblock {\em SIAM Journal on Applied Mathematics}, 17(1):33--45, 1969.
\newblock \href {https://doi.org/10.1137/0117004} {\path{doi:10.1137/0117004}}.

\bibitem{treynor:62}
J.~L. Treynor.
\newblock {Toward a Theory of Market Value of Risky Assets}.
\newblock Unpublished manuscript, 1962.
\newblock DOI identifies the edited reprint.
\newblock \href {https://doi.org/10.2139/ssrn.628187}
  {\path{doi:10.2139/ssrn.628187}}.

\bibitem{wenzelburger:10}
J.~Wenzelburger.
\newblock {The two-fund separation theorem revisited}.
\newblock {\em Annals of Finance}, 6(2):221--239, 2010.
\newblock \href {https://doi.org/10.1007/s10436-009-0144-8}
  {\path{doi:10.1007/s10436-009-0144-8}}.

\bibitem{xing:zitkovic:18}
H.~Xing and G.~{\v Z}itkovi{\'c}.
\newblock {A class of globally solvable Markovian quadratic BSDE systems and
  applications}.
\newblock {\em The Annals of Probability}, 46(1):491--550, 2018.
\newblock \href {https://doi.org/10.1214/17-AOP1190}
  {\path{doi:10.1214/17-AOP1190}}.

\bibitem{zitkovic:12}
G.~{\v Z}itkovi{\'c}.
\newblock {An example of a stochastic equilibrium with incomplete markets}.
\newblock {\em Finance and Stochastics}, 16(2):177--206, 2012.
\newblock \href {https://doi.org/10.1007/s00780-011-0161-0}
  {\path{doi:10.1007/s00780-011-0161-0}}.

\end{thebibliography}
\end{document}